\documentclass{amsart}
\usepackage[centertags]{amsmath}
\usepackage{amsfonts,amssymb}
\usepackage{enumerate}
\usepackage{amsfonts,amsmath,amssymb,cite,graphicx,color,ulem,amsthm}
\usepackage{fourier}
\usepackage[margin=1.5in]{geometry}
\usepackage{hyperref}

\def\eps{\varepsilon}

\def\RR{\mathbb{R}} \def\NN{\mathbb{N}} \def\ZZ{\mathbb{Z}}
\def\EE{\mathbb{E}} \def\PP{\mathbb{P}}
\newcommand{\tr}{\mathrm{tr}}
\renewcommand{\det}{\operatorname{det}}

\newcommand{\argmax}{\operatornamewithlimits{argmax}}
\newcommand{\argmin}{\operatornamewithlimits{argmin}}

\theoremstyle{plain}
\newtheorem{theorem}{Theorem}[section]
\newtheorem{proposition}[theorem]{Proposition}
\newtheorem{lemma}[theorem]{Lemma}
\newtheorem{assumption}[theorem]{Assumption}

\theoremstyle{remark}
\newtheorem{remark}{Remark}[section]

\def\<{\langle} \def\>{\rangle}

\title[Sharp Asymptotic Estimates in Stochastic Systems with Small
  Noise]{Sharp Asymptotic Estimates for Expectations, Probabilities, and
  Mean First Passage Times in Stochastic Systems with Small
  Noise}

\author{Tobias Grafke$^1$}
\address{$^1$Mathematics Institute, University of Warwick, Coventry CV4 7AL, United Kingdom}
\author{Tobias Sch{\"a}fer$^2$}
\address {$^2$Department of Mathematics, College of Staten Island 1S-215, 2800 Victory Blvd., Staten Island, NY \& Physics Program, CUNY Graduate Center, 365 Fifth Avenue New York, NY}
\author{Eric Vanden-Eijnden$^3$}
\address{$^3$Courant Institute, New York University, 251 Mercer Street, New York, NY.}

\begin{document}
\numberwithin{equation}{section}

\begin{abstract}
  Freidlin-Wentzell theory of large deviations  can be used to compute
  the likelihood  of extreme  or rare  events in  stochastic dynamical
  systems via  the solution of  an optimization problem.  The approach
  gives  exponential  estimates that  often  need  to be  refined  via
  calculation of  a prefactor. Here it  is shown how to  perform these
  computations in  practice. Specifically, sharp  asymptotic estimates
  are derived for expectations,  probabilities, and mean first passage
  times  in a  form that  is geared  towards numerical  purposes: they
  require solving  well-posed matrix  Riccati equations  involving the
  minimizer of  the Freidlin-Wentzell action as  input, either forward
  or backward  in time  with appropriate  initial or  final conditions
  tailored to the estimate at hand.  The usefulness of our approach is
  illustrated on  several examples.  In particular,  invariant measure
  probabilities and mean first passage  times are calculated in models
  involving    stochastic    partial   differential    equations    of
  reaction-advection-diffusion type.
\end{abstract}

\maketitle

\tableofcontents

\section{Introduction}
\label{sec:intro}

Rare events in stochastic dynamical systems tend to cluster around
their most likely realization. As a result they have predictable
features that can be calculated via some optimization problem. This
profound observation has been made in numerous fields, and used
e.g. to explain phase transitions in statistical
mechanics~\cite{ellis2006entropy}, derive Arrhenius' law in chemical
kinetics~\cite{arrhenius:1889} , or use semiclassical trajectories in
quantum field
theory~\cite{vainshtein-zakharov-novikov-etal:1982}. Large deviation
theory (LDT)~\cite{varadhan:2008} gives a mathematical justification
to these results and provide us with an action, or rate function, to
minimize in order to calculate paths of maximum likelihood, also known
as \textit{instantons}. The theory also gives exponential asymptotic
estimates of rare event probabilities. While this information is
already useful in many cases, more refined estimates are often
desirable.  These `prefactor' calculations attempt to quantify the
effect of Gaussian fluctuations around the instanton, a notion that
has also been separately rediscovered in the literature through
various means~\cite{berglund:2013}. For example, in the context of
chemical reaction rates, next order refinements of the exponential
reaction rate are known as the Eyring-Kramers law~\cite{eyring:1935,
  kramers:1940}. Similarly in quantum field theory, perturbing around
the semiclassical trajectory, the second order variations leads to a
Gaussian path-integral, which ultimately results in an additional
contribution in the form of a ratio of functional
determinants~\cite{vainshtein-zakharov-novikov-etal:1982}.

Over the last two decades, several computational methods have been
developed to calculate instantons. Among others, we refer to the
string method in the context of gradient
flows~\cite{e-ren-vanden-eijnden:2002, e-ren-vanden-eijnden:2007}, the
minimum action
method~\cite{e-ren-vanden-eijnden:2004,fogedby-ren:2009}, the adaptive
minimum action method (aMAM)~\cite{zhou-ren-e:2008} and the geometric
minimum action method (gMAM)~\cite{heymann-vanden-eijnden:2008,
  heymann-vanden-eijnden:2008-a, vanden-eijnden-heymann:2008}. These
methods are now efficient enough to be used in the context
high-dimensional systems, including stochastically driven partial
differential equations arising in fluid
dynamics~\cite{grafke-grauer-schindel:2015,
  grafke-grauer-schaefer:2015}.

In contrast, surprisingly little work has been done on the numerical
side of prefactor calculations (see
however~\cite{schorlepp-grafke-grauer:2021}). The main objective of this
paper is to show how to extend methods such as MAM or gMAM to
efficiently estimate prefactors in the context of the calculations of
expectations, probabilities, and mean first passage times.

\subsection{Large Deviation Theory and instantons}
\label{sec:inst-large-devi}

Consider a family of stochastic differential equations (SDEs) for
$X^\epsilon_t\in\RR^n$, with drift vector field~$b: \RR^n\to\RR^n$ and
diffusion matrix $a = \sigma\sigma^\top$, where
$\sigma\in\RR^{n\times n}$,
\begin{equation}
  \label{eq:sde}
  dX^\epsilon_t = b(X^\epsilon_t)\,dt + \sqrt{\epsilon}\sigma \,dW_t\,.
\end{equation}
Here, $W_t$ is an $n$-dimensional Wiener process, and we have
introduced a small parameter $\epsilon>0$ to characterize the strength
of the noise.  For simplicity we will assume that the diffusion matrix
$a\in\RR^{n\times n}$ is positive-definite (hence invertible) and
constant (but not necessarily diagonal), i.e.~the case of additive
Gaussian noise---the generalization of the methods presented below to
a covariance matrix $a$ that depends on $x$, i.e.~multiplicative
Gaussian noise, is straightforward.  Large deviations
theory~\cite{varadhan:2008, freidlin-wentzell:2012} indicates that, in
the limit as~$\epsilon\to 0$, the solutions to~\eqref{eq:sde} that
contribute most to the probability of an event or the value of an
expectation are likely to be close to the minimizer of the
Freidlin-Wentzell rate function $S_T$ subject to appropriate boundary
conditions. This action functional is given by
\begin{equation}
\label{eq:Freidlin-Wentzell_action}
S_T(\phi) = \int_0^T\,L(\phi,\dot \phi)\,dt
\end{equation}
with the Lagrangian
\begin{equation}
  L(\phi,\dot \phi) = \tfrac{1}{2}\langle \dot \phi-b(\phi),a^{-1}(\dot
  \phi-b(\phi))\rangle\,
  \equiv \tfrac{1}{2}|\dot \phi -b(\phi)|_a^2.
\end{equation}
Here $\langle x,y\rangle$ stands for the Euclidean scalar product
between the vectors $x$ and $y$ and we introduced the norm induced
by~$a$,
$|x|^2_a = \langle x,a^{-1}x\rangle \equiv \sum_{i,j=1}^n x_i
a_{i,j}^{-1} x_j$. If the diffusion matrix $a$ is the identity, this
norm becomes simply the Euclidean length.

The minimizer of the action (\ref{eq:Freidlin-Wentzell_action}) is
referred to as path of maximum likelihood or \textit{instanton}, and
it can be found by solving the corresponding Euler-Lagrange equations
\begin{equation}
\label{eq:Euler-Lagrange_original}
\frac{d}{dt}\frac{\partial L}{\partial \dot \phi}=\frac{\partial L}{\partial \phi}
\end{equation}
with boundary conditions appropriate to the event under
consideration. Thus, in the context of large deviation theory, the
leading order estimation of probabilities or expectations can be
reduced to the solution of the deterministic
system~\eqref{eq:Euler-Lagrange_original}.

Alternatively, there is a Hamiltonian formulation to the
problem. Taking the Legendre transform of the Lagrangian and
introducing the momentum $\theta = \partial L/\partial \dot \phi$, we
obtain the Hamiltonian
\begin{equation}
  \label{eq:hamiltonian_diffusion}
  H(\phi,\theta) = \langle b(\phi), \theta\rangle+\tfrac{1}{2}\langle \theta, a \theta \rangle,
\end{equation}
and the Euler-Lagrange equations for the instanton become
\begin{equation}
\label{eq:hamilton_time_parametrization}
\dot \phi = b(\phi)+a \theta, \qquad \dot \theta = -(\nabla b(\phi))^\top \theta\,.
\end{equation}
\subsection{Prefactor estimates}
\label{sec:prefactor-estimates}

The instanton gives the leading contribution to the exponential decay
of the probability for observing a rare event. In order to obtain
sharp estimates, one needs to furthermore consider prefactor
contributions.

Intuitively these prefactors can be calculated by accounting for the
effects of the fluctuations around the instanton~$\phi$, which can be
done by linearizing the solution of the SDE~\eqref{eq:sde}
around~$\phi$ and considering 
\begin{equation}
  \label{eq:Zeq0}
  dZ_t = \nabla b(\phi(t)) Z_t\,dt + \sigma\,dW_t.
\end{equation}
The solution to this equation defines a Gaussian process, and the
prefactor contribution to expectations, probabilities, or mean first
passage times can be calculated as specific expectations over this
process. In turns, these expectations are ratios between determinants
of specific positive-definite matrices or operators that can be
expressed in terms of the solutions of deterministic Riccati
equations, as can be intuited by analogy with results from optimal
control theory. Our objectives here are to: (i) formulate these
Riccati equations, including their boundary conditions, in the
specific cases of the calculation of expectations, probabilities, and
mean first passage time; and (ii) develop efficient numerical methods
for their solution.

\subsection{Related works}
\label{sec:related}

From a theoretical point of view, our approach builds on a large
corpus of works dealing with expansion beyond the exponential estimate
of LDT and the evaluation of quadratic path integrals or Wiener
integrals, as initiated by Schilder~\cite{schilder:1966}. Results from
probability theory and stochastic analysis in this direction include
for example the pioneering works by Kifer~\cite{kifer1977asymptotics}
and Azencott~\cite{azencott:1985}. On the analytical side, formal
asymptotic expansions were used e.g.
in~\cite{matkowsky1977exit,maier-stein:1997}, the WKB expansion
in~\cite{fleming-james:1992}, optimal control theory
in~\cite{barles-perthame:1988}, and more recently potential theory
in~\cite{bovier-eckhoff-gayrard-etal:2004, berglund:2013}.  This last
approach allows one to establish rigorously the Eyring-Kramers law for
reversible system~\cite{bovier-hollander:2015}, and has also been
extended to other problems, such as lattice gas
models~\cite{hollander:2004}. Similarly, prefactor calculations
recently included non-reversible
systems~\cite{bouchet-reygner:2016}. From a mathematical perspective
it is much harder to treat the infinite dimensional case, even though
recent breakthroughs have been made at least for systems in detailed
balance~\cite{faris-jona-lasinio:1982, berglund-gentz:2013,
  barret:2015, berglund-di_gesu-weber:2016}.

From a computational viewpoint, none of the references above deal with
explicit calculation of the instanton or the prefactors, and the
equations they derive for these objects, while often very general, do
not lend themselves automatically to numerical implementation.
Explicit (especially numerical) computation of prefactors is usually
confined to quantum field theory, where the evaluation of Gaussian
path integrals is a classical result~\cite{montroll:1952,
  gelfand-yaglom:1960, levit-smilansky:1977, forman:1987}, see
also~\cite{dunne:2008} for a recent review. However, the Lagrangian in
quantum mechanics is usually assumed to take a special form with no
first order time derivative in the Euler-Lagrange equation
(corresponding to a stochastic process in detailed balance in the
stochastic interpretation). The processes we are interested in are
considerably more general from that perspective, and require the more
general methods we develop here.  These methods built on formulas
derived by asymptotic analysis of backward Kolmogorov equations, and
are complementary to those recently proposed
in~\cite{schorlepp-grafke-grauer:2021} using a path integral approach.

\subsection{Main contributions}
\label{sec:maincontrib}

Our main results can be summarized as follows: (1) We provide formal
but short and simple proofs of propositions establishing sharp
estimates for expectations, probability densities, probabilities, exit
probabilities, and mean first passage times. While most of these
results can be found in some form in the literature, they are
scattered in many different papers, and we believe it is useful to
collect and summarize them in one place. The methods used in these
proofs can also be extended to more general situations not covered
here. (2) We phrase each of our theoretical statements in a way geared
towards numerical implementation, unlike what is usually found in the
literature on the topic. (3) We use the geometric approach from gMAM
to provide statements that are valid at infinite time (i.e. on the
invariant measure of the process), by explicitly computing this
infinite time limit via mapping of $t\in (-\infty,0]$ onto the
normalized arclength $s\in(0,1]$ along the instantons involved. (4) We
formally generalize our results to the infinite-dimensional setup,
with applications to stochastic partial equations of
reaction-advection-diffusion type. (5) We also generalize our result
to examples driven by non-Gaussian noise, specifically Markov jump
processes in specific limits. (6) We illustrate the applicability of
our method through tests cases in finite and infinite dimension, in
which we discuss how to perform the numerical calculation involved.

In terms of limitations, our work focuses on situations where the
stochastic system at hand has a single attracting point in the limit
of vanishing noise. This setup is of interest in several situations,
but it excludes the important problem of analyzing rare transitions
between metastable states. Some of the tools we develop here, in
particular the Riccati equations whose solutions enter the expressions
for the prefactors, may be useful to analyze these transitions, but we
will not consider them here.

\subsection{Assumptions and organization}
\label{sec:style}

As stated before, we are interested in obtaining sharp asymptotic
estimates for expectations, probabilities, exit probabilities, and
mean first passage times. We will do so under the generic assumption
that the LDT optimization problem associated with each of these
questions, i.e. the minimization of the action
in~\eqref{eq:Freidlin-Wentzell_action} subject to appropriate
constraints and/or boundary conditions, is strictly convex. This
simplifying assumption guarantees that the solution of the equations
presented below exists and is unique, and therefore allows us to avoid
dealing with local minimizers, flat minima, etc. that may require to
generalize/amend some of the statements below. While this is not
necessary difficult to do, at least formally, it leads to a zoology of
subcases that we want to avoid listing.

To make \eqref{eq:sde} well-posed, we will also make:
\begin{assumption}
  \label{as:1}
  The vector field $b$ is $C^2(\RR^n)$ and such that:
  $$\exists \alpha,\beta>0\ :\ \< b(x),x\> \le \alpha- \beta|x|^2 \quad \forall x\in\RR^n;$$
  and the matrix $a$ is such that:
  $$\exists \gamma,\Gamma \ \ \text{with} \ \ 0< \gamma< \Gamma < \infty\  :\  \gamma|x|^2 \le
  \langle x, a
  x\rangle < \Gamma |x|^2\quad \forall x \in \RR^n.$$
\end{assumption}
\noindent
This assumption guarantees~\cite{mattingly2002ergodicity} that the
solution to the SDE~\eqref{eq:sde} exists for all times and is
ergodic with respect to a unique invariant measure with a probability
density function $\rho_\eps:\RR^n \to (0,\infty)$. This density is the
unique solution to
\begin{equation}
  \label{eq:9}
  0 = -\nabla \cdot \left (b\rho_\eps\right) + \tfrac12\eps a : \nabla \nabla 
  \rho_\eps, \qquad \rho_\eps \ge 0, \qquad \int_{\RR^n} \rho_\eps(x)
  dx = 1, 
\end{equation}
where here and below the colon denotes the trace, i.e. $a:\nabla\nabla
= \tr(a [\nabla \otimes \nabla])$.
We will also make a stronger assumption:
\begin{assumption}
  \label{as:2}
  The ODE $\dot x = b(x)$ has a single fixed point located at $x_*$
  (i.e.~$x_*$ is the only solution to $b(x)=0$), which is linearly
  stable locally (i.e.~the real part of all eigenvalues of the matrix
  $\nabla b(x_*)$ are strictly negative), and globally attracting (i.e.~any
  solution to $\dot x = b(x)$ approaches $x_*$ asymptotically).
\end{assumption}
\noindent
This assumption implies that $\rho_\eps$ becomes atomic on $x=x_*$
as~$\eps\to0$, a property we will need when looking at expectations or
probability on the invariant measure.

\medskip

The remainder of this paper is organized as follows: In
Sec.~\ref{sec:expectations} we will first consider the problem of
calculating sharp estimates of expectations, both at finite time
(Secs.~\ref{sec:ftexpedct} and~\ref{sec:expect-via-girs}) and on the
invariant measure of~\eqref{eq:sde} (Sec.~\ref{eq:IM}). In
Sec.~\ref{sec:pdf}, we will show how to calculate sharp estimates of
probabilities densities at finite (Sec.~\ref{sec:pdffinitet}) and
infinite times (Sec.~\ref{sec:prob-dens-invar}). In
Sec.~\ref{sec:prob} we then build on these results to calculate
probabilities at finite times (Sec.~\ref{sec:prob-finite-t}) and on
the invariant measure of the process
(Sec.~\ref{sec:prob-dens-invar}). Finally, in Sec.~\ref{sec:mfpt} we
consider the problem of mean exit time calculation. These results are
illustrated on finite dimensional examples throughout the paper, and
on infinite dimensional examples in Sec.~\ref{sec:infinited}, where we
consider linear and nonlinear reaction-advection-diffusion equations with random
forcing.  The results in this paper are mostly for diffusions, but
they can be generalized to other set-ups, in particular Markov jump
processes: this is discussed in Sec.~\ref{sec:nonGaussian}. We end the
paper with some conclusions in Sec.~\ref{sec:conclu} and defer some
technical results to Appendices.

\section{Expectations}
\label{sec:expectations}

\subsection{Finite time expectations}
\label{sec:ftexpedct}

Given the observable $f:\RR^n\to\RR$ with $f \in C^2(\RR^n)$, consider
\begin{equation}
  \label{eq:expect}
  A_\eps(T,x) = \EE^x \exp\left(\eps^{-1} f(X^\eps_T)\right)\,,
\end{equation}
where the expectation is taken over samples of the SDE~(\ref{eq:sde})
conditioned on $X^\eps_0 = x$ and evaluated at the final time $t=T<\infty$.
We have the following proposition:

\begin{proposition}
  \label{thm:expectation}
  Let $(\phi_x(t),\theta_x(t))$ solve the instanton equations
  \begin{equation}
    \label{eq:4}
    \begin{aligned}
      \dot \phi_x &= a \theta_x+b(\phi_x), \qquad &&\phi_x(0) = x,\\
      \dot \theta_x &= -(\nabla b(\phi_x))^\top \theta, &&\theta_x(T) = \nabla
      f(\phi_x(T)),
    \end{aligned}
  \end{equation}
  and $W_x(t)$ be the solution to the Riccati equation
  \begin{equation}
    \label{eq:W}
    \dot W_x = -\nabla\nabla\langle b(\phi_x),\theta_x\rangle
    - (\nabla b(\phi_x))^\top W_x
    - W_x (\nabla b(\phi_x)) - W_x a W_x\,,\qquad W_x(T) =
    \nabla\nabla f(\phi_x(T))\,,
  \end{equation}
  integrated backwards in time from $t=T$ to $t=0$ along the instanton
  $\phi_x(t)$. Then the expectation~(\ref{eq:expect}) satisfies
  \begin{equation}
    \label{eq:limexpect}
    \lim_{\eps\to0} \frac{A_\eps(T,x)}{\bar A_\eps(T,x)} = 1,
  \end{equation}
  where
  \begin{equation}
    \label{eq:3}
    \bar A_\eps (T,x) = R(T,x) \exp\left(\eps^{-1}\left(f(\phi_x(T))
        -\frac12 \int_0^T \langle \theta_x(t),a\theta_x(t)\rangle\,dt\right)\right)\,,
  \end{equation}
  with
  \begin{equation}
    \label{eq:30}
    R(T,x) = \exp\left(\frac12
      \int_0^T \tr(aW_x(t))\,dt\right).
  \end{equation}
\end{proposition}

\begin{remark}
  The function $R(T,x)$ is typically referred to as the prefactor. LDT
  gives a rougher estimate
  \begin{equation}
    \label{eq:31}
    \lim_{\eps\to0} \frac{\log A_\eps (T,x)}{\log \bar A_\eps (T,x)} =1,
  \end{equation}
  which would be unaffected if we were to neglect the prefactor $R(T,x)$
  in $\bar A_\eps (T,x)$. Of course, this prefactor is key to get the
  more refined estimate in~\eqref{eq:limexpect}.
\end{remark}

\begin{remark}
  If the SDE~\eqref{eq:sde} is modified into
  \begin{equation}
    \label{eq:20}
    dX^\epsilon_t = b(X^\epsilon_t)\,dt + \eps \tilde
    b(X^\epsilon_t)\,dt  +
    \sqrt{\epsilon}\sigma \,dW_t\,.
  \end{equation}
  with $\tilde b:\RR^n \to\RR^n$ is $C^1(\RR^n)$ with bounded
  derivatives, Proposition~\ref{thm:expectation} can be generalized by
  replacing~\eqref{eq:30} with
  \begin{equation}
    \label{eq:limexpectadddrfift}
    R(T,x) = \exp\left(\frac12
      \int_0^T \tr(aW_x(t))\,dt+ \int_0^T \< \theta_x(t),\tilde
      b(\phi_x(t)\> dt\right),
  \end{equation}
  while leaving unchanged all the other equations in the
  proposition. The statements in the propositions below can similarly be
  straightforwardly amended to apply to~\eqref{eq:20}, but for the
  sake of brevity we will stick to~\eqref{eq:sde}.
\end{remark}

\begin{proof}[Proof of Proposition~\ref{thm:expectation}]
  Let
  \begin{equation}
    u_\eps(T-t,x) = {\EE}^{x}\exp\left(\eps^{-1} f(X^\eps_t) \right)
  \end{equation}
  so that
  \begin{equation}
    \label{eq:6}
    u_\eps(0,x) = A_\eps(T,x).
  \end{equation}
  It is well-known that $u_\eps$ satisfies the backward Kolmogorov equation
  (BKE)
  \begin{equation}
    \label{eq:1}
    \partial_t u_\eps+ L_\eps u_\eps=0, \qquad u_\eps(T,x) =  \exp\left(\eps^{-1}f(x)\right)\,,
  \end{equation}
  where $L_\eps$ is the generator of the process~\eqref{eq:sde}:
  \begin{equation}
    \label{eq:5}
    L_\eps = b(x) \cdot \nabla + \tfrac12 \eps a: \nabla \nabla 
  \end{equation}
  Look for a solution of~\eqref{eq:1} 
  of the form
  \begin{equation}
    \label{eq:2}
    u_\eps(t,x) = Z_\eps(t,x)\, \exp\left(\eps^{-1}S(t,x)\right),
  \end{equation}
  where $S(t,x)$ satisfies the Hamilton-Jacobi equation
  \begin{equation}
    \label{eq:hamilton_jacobi_w}
    \partial_t S + b(x)\cdot \nabla S + \tfrac12\langle \nabla S, a \nabla S \rangle
    = \partial_t S + H(x,\nabla S) = 0\,, \qquad S(T,x) = f(x).
  \end{equation}
  If $(\phi_x(t),\theta_x(t))$ solves the instanton equations
  in~\eqref{eq:4}, we have $\theta_x(t) = \nabla S(t,\phi_x(t))$ and a
  direct calculation shows that
  \begin{equation}
    \frac{d}{dt} S(t,\phi_x(t)) = \partial_t S + \dot \phi_x \cdot
    \nabla S
    = \tfrac{1}{2} \langle \theta_x(t),a\theta_x(t) \rangle\,,
  \end{equation}
  implying
  \begin{equation}
    S(T,\phi(T)) - S(0,\phi(0)) =
    \frac{1}{2} \int_0^T  \langle \theta,a\theta \rangle\,dt\,.
  \end{equation}
  Since $S(T,\phi(T)) = f(\phi(T))$, we have
  \begin{equation}
    \exp(\eps^{-1} S(0,\phi(0)))
    = \exp\left(\eps^{-1}\left(f(\phi(T))-
        \frac12\int_0^T\langle \theta,
        a\theta\rangle\,dt\right)\right).
  \end{equation}
  As a result, to show that~\eqref{eq:limexpect} holds, it remains to
  establish that the factor $Z_\eps(t,x)$ has a limit as $\eps\to0$
  with $\lim_{\eps\to0} Z_\eps(0,x)= R(T,x)$.  To this end, notice
  that $Z_\eps(t,x)$ satisfies
  \begin{equation}
    \label{eq:7}
    \partial_t Z_\eps + (b+a\nabla S) \cdot \nabla Z_\eps + \tfrac12
    Z_\eps a:\nabla\nabla S + \tfrac12\eps a : \nabla \nabla Z_\eps= 0,\qquad 
    Z_\eps(T,x) = 1. 
  \end{equation}
  Taking the limit as $\eps\to0$ on this equation, we formally deduce
  that $\lim_{\eps\to0} Z_\eps(t,x)= Z(t,x)$, where $Z(t,x)$ solves
  \begin{equation}
     \label{eq:7b}
    \partial_t Z + (b+a\nabla S) \cdot \nabla Z + \tfrac12  Z a:\nabla\nabla S = 0 \qquad 
    Z(T,x) = 1. 
  \end{equation}
  Setting $G_x(t) = Z(t,\phi_x(t))$ on the instanton path, so that $G_x(0) =
  Z(0,\phi_x(0)) = Z(0,x)$, and using the
  instanton equation $\dot \phi_x = b(\phi_x) + a\theta_x$, we find as evolution
  equation for $G_x$
  \begin{equation} \label{eq:find_G}
    \dot G _x = -\tfrac12 G_x\, {\mathrm{tr}}(aW_x) \,,\qquad G_x(T)=1\,,
  \end{equation}
  where the matrix $W_x(t)$ is  defined as the Hessian of $S(t,x)$
  evaluated along the characteristics (i.e.~the instanton path
  $\phi_x(t)$): $W_x(t)= \nabla \nabla S(t,\phi_x(t))$. In order to
  solve the equation (\ref{eq:find_G}) for $G_x(t)$, we need an
  equation for $W_x(t)$. Differentiating the Hamilton-Jacobi equation
  (\ref{eq:hamilton_jacobi_w}) twice with respect to $x$ and
  evaluating the result at $x=\phi_x(t)$, it is easy to show that
  $W_x$ solves the Riccati equation in~\eqref{eq:W}. Therefore, by
  integrating~(\ref{eq:find_G}), we deduce
  \begin{equation}
    G_x(0) = Z(0,x)
    = \exp\left(\frac12\int_0^T\tr(a W_x(t))\,dt\right)\equiv R(T,x)\,,
  \end{equation}
  which terminates the proof.
\end{proof}

\subsection{Expectations via Girsanov theorem}
\label{sec:expect-via-girs}

An alternative justification of the Riccati equation~\eqref{eq:W} can
be given by introducing a stochastic process that samples the Gaussian
fluctuations around the instanton. This approach opens up the
possibility to alternatively compute the prefactor contribution as an
expectations via sampling techniques. This result is well-known (see
e.g.~\cite{freidlin-wentzell:2012}) and can be phrased as:

\begin{proposition}
  \label{thm:expectationb}
  The prefactor~\eqref{eq:30} satisfies
  \begin{equation}
    \label{eq:limexpectb}
    R(T,x)
    = \EE^0 \exp\left(\frac12 \int_0^T
      \nabla\nabla\langle b(\phi_x(t)),\theta_x(t) \rangle : Z_t Z_t\,dt
      +\tfrac12\nabla\nabla f(\phi_x(T)) : Z_T Z_T\right)\,,
  \end{equation}
  where $(\phi_x(t),\theta_x(t))$  are defined as in
  Proposition~\ref{thm:expectation}, $Z_t$ solves
\begin{equation}
  \label{eq:Zeq}
  dZ_t = \nabla b(\phi_x(t)) Z_t\,dt + \sigma\,dW_t,
\end{equation}
and the expectation $\EE^0$ is taken over realizations
of~\eqref{eq:Zeq} conditioned on $Z_0=0$.
\end{proposition}

\begin{proof}
  Let $Y^\eps_t\in\RR^n$ satisfy
\begin{equation}
  \label{eq:dY}
  dY^\eps_t = \dot\phi_x(t)\,dt + \sqrt{\epsilon} \sigma\,dW_t, 
\end{equation}
where $\phi_x(t)$ is the instanton solution to~\eqref{eq:4}. By invoking
Girsanov's theorem, we can write $A_\epsilon(x)$ as
\begin{equation}
  A_\epsilon(T,x)=\EE^x M^\eps_T \exp\left( \eps^{-1}f(Y^\eps_T)\right)
\end{equation}
where $M^\eps_T$ is the Radon-Nikodym density
\begin{equation}
  M^\eps_T =
  \exp\left(-\frac1{2\epsilon}\int_0^T|\dot\phi_x-b(Y^\eps_t)|_a^2\,dt
    - \frac1{\sqrt{\epsilon}}\int_0^T \langle
    \sigma^{-1}(\dot\phi_x-b(Y^\eps_t)),
    dW_t\rangle\right)\,.
\end{equation}
Since $Y^\eps_t= \phi_x(t) +\sqrt{\eps} \sigma dW_t$, after expanding
both $b(Y^\eps_t)$ and $f(Y^\eps_t)$ in $\epsilon$, it is easy to see
that the leading order contribution is precisely given by
\begin{align*}
  &\exp\left(-\frac1{2\eps}
  \int_0^T|\dot\phi_x -  b(\phi_x)|^2_a
  \,dt + \eps^{-1} f(\phi_x(T))\right)\\
  &= \exp\left(-\frac1{2\eps}\int_0^T \langle
    \theta_x(t),a\theta_x(t)\rangle\,dt
    + \eps^{-1} f(\phi_x(T))\right)
\end{align*}
The next order vanishes due to the criticality of the minimizer. The
first correction term in the exponential is therefore $O(\epsilon^0)$, i.e.~it gives the
prefactor $R(T,x)$,  and reads
\begin{equation}
  \label{eq:prefactor_expectation}
  \begin{split}
    R(T,x)=\EE
    \exp\bigg(&-\frac12 \int_0^T |\nabla b(\phi_x(t)) U_t|_a^2\,dt
    + \frac12\int_0^T \nabla\nabla\langle b(\phi_x),\theta_x\rangle : U_t U_t\,dt\\
    & + \int_0^T \langle\nabla b(\phi_x) U_t, dU_t\rangle_a
    +\tfrac12\nabla\nabla f(\phi_x(T)) : U_T U_T\bigg)\,,
  \end{split}
\end{equation}
where $U_t = \sigma W_t$. Noticing that the term
\begin{equation}
  \exp\bigg(-\frac12 \int_0^T |\nabla b(\phi_x) U_t|_a^2\,dt
  + \int_0^T \langle\nabla b(\phi_x) U_t, dU_t\rangle_a\bigg)
\end{equation}
is a itself Radon-Nikodym density for the change of measure from the
random process $Z_t$ defined in~\eqref{eq:Zeq} to $U_t=\sigma W_t$, we
can therefore alternatively write the right hand-side
of~\eqref{eq:prefactor_expectation} as in~\eqref{eq:limexpectb}.
\end{proof}

Note that the formula~\eqref{eq:limexpectb} immediately tells us that
the prefactor, defined as the limit as $\eps\to0$ of the ratio between
$A_\eps(T,x)$ and $\bar A^\eps(T,x)$, is unity when both the drift~$b$
and the observable~$f$ are linear. Note also that computing the
expectation by using the change of measure from the original process
to one representing fluctuations around the instanton can be seen as
an approximated way (up to terms of order $O(\eps)$) of performing
importance sampling via Monte-Carlo method: how to do this importance
sampling exactly is harder in general, as discussed
e.g. in~\cite{vanden-eijnden-weare:2012}.

Finally, note that another, more direct, proof
Proposition~\ref{thm:expectationb} goes as follows. It is easy to see
that
\begin{equation}
  \label{eq:8}
  \EE^0 \exp\left(\frac12 \int_0^T
    \nabla\nabla \langle b(\phi_x),\theta_x\rangle : Z_t Z_t\,dt
    +\tfrac12\nabla\nabla f(\phi_x(T)) : Z_T Z_T\right) = v(0,0)
\end{equation}
where $v(t,z)$ solves
\begin{equation}
  \label{eq:pde}
  \partial_t v + \langle \nabla b(\phi_x) z, \nabla v\rangle +
  \tfrac12 a :\nabla \nabla v + q(t,z) v = 0,\qquad v(T,z)
  = \exp\left(\tfrac12\nabla\nabla f(\phi_x(T)) : zz\right)\,,
\end{equation}
with
\begin{equation}
  q(t,z) = \tfrac12\nabla\nabla\<  b(\phi_x(t)), \theta_x\rangle : z z\,.
\end{equation}
The solution of~\eqref{eq:pde} can be expressed as
\begin{equation}
  v(t,z) = G_x(t) \exp\left(\tfrac12 \langle z, W_x(t) z\rangle\right)\,,
\end{equation}
where $W_x(t)$ and $G_x(t)$ solve~\eqref{eq:W} and~\eqref{eq:find_G},
respectively.  Therefore $v(0,0) = G_x(0)$, consistent with the
statement in Proposition~\ref{thm:expectation}. In
Appendix~\ref{sec:riccatiexpect} we list a few more expressions that
relate the solution to Riccati equations like~\eqref{eq:W} to
expectations over the solution of a linear SDE
like~\eqref{eq:Zeq}. These expressions are useful as they give a
possible route to solve the Riccati equation~\eqref{eq:W} via sampling, which is
less accurate in general but simpler than solving the Riccati equation
itself. 

\subsection{Expectations on the invariant measure}
\label{eq:IM}

Given the observable $f:\RR^n\to\RR$ with $f\in C^2(\RR^n)$, consider
the expectation
\begin{equation}
  \label{eq:expectinf}
  B_\eps = \int_{\RR^n} \exp\left(\eps^{-1}f(y)\right)\rho_\eps(y) dy
\end{equation}
where $\rho_\eps(y)$ denotes the invariant density solution
to~\eqref{eq:9}. \eqref{eq:expectinf} can also be written as
\begin{equation}
  \label{eq:10}
  B_\eps = \lim_{T\to\infty} A_\eps(T,x)
\end{equation}
Assumption~\ref{as:1} guarantees that this limit exists for
appropriate $f$ (e.g. if $|f|$ is bounded), and is independent
on~$x$.

The next proposition shows how to calculate a sharp estimate
for~\eqref{eq:expectinf} using the procedure of gMAM that allows us to
compactify the physical time interval $[0,\infty)$ onto $[0,1]$.

\begin{proposition}
  \label{thm:expectationinf}
  Let $(\hat\phi(s),\hat\theta(s))$ solve the geometric instanton equations
  \begin{equation}
    \label{eq:4infpdf}
    \begin{aligned}
      \lambda \hat\phi' &= a \hat\theta+b(\hat\phi), \qquad &&\hat\phi(0) = x_*,\\
      \lambda \hat\theta' &= -(\nabla b(\hat\phi))^\top \hat\theta,
      &&\hat\theta(1) = \nabla f(\hat\phi(1)),
    \end{aligned}
  \end{equation}
  with $\lambda = |b(\hat\phi)|_a/|\hat\phi'|_a$, and $\hat W(s)$ be the
  solution to the Riccati equation
  \begin{equation}
    \label{eq:Winf}
    \lambda \hat W' = -\nabla\nabla \langle b(\hat\phi), \hat\theta\rangle
    - (\nabla b(\hat\phi))^\top
    \hat W - \hat W (\nabla b(\hat\phi)) - \hat W a \hat W\,,
    \qquad \hat W(1) = \nabla\nabla f(\hat\phi(1))\,,
  \end{equation}
  integrated backwards in time from $s=1$ to $s=0$ along the instanton
  $\hat\phi(s)$.  Then the expectation~(\ref{eq:expectinf}) satisfies
  \begin{equation}
    \label{eq:limexpectinf}
    \lim_{\eps\to0} \frac{B_\eps}{\bar B_\eps} = 1,
  \end{equation}
  where
  \begin{equation}
    \label{eq:3inf}
    \bar B_\eps = \hat R \exp\left(\eps^{-1}\left(f(\hat\phi(1))-
        \frac12 \int_{0}^1 \lambda^{-1}(s)\langle
        \hat\theta(s),a\hat\theta(s)\rangle\,ds\right)\right)\,,
  \end{equation}
  with
  \begin{equation}
    \label{eq:32}
    \hat R = \exp\left(\frac12
      \int_{0}^1 \lambda^{-1}(s) \tr(a\hat W(s))\,ds\right)\,.
  \end{equation}
\end{proposition}

For more details about the geometric instanton equations we refer the
reader to~\cite{heymann-vanden-eijnden:2008,
  grafke-grauer-schaefer-etal:2014,
  grafke-schaefer-vanden-eijnden:2017} as well as
Appendix~\ref{sec:init-cond-geom}. Note that the parametrization of
$\hat \phi$ can be chosen arbitrarily: In the calculations, it is
convenient to use normalized arc-length, i.e. impose that
$ |\hat\phi'(s)|_a=L$, where the constant $L$ is the length of the
instanton.

\begin{proof}
  In order to speak meaningfully about the limit $T\to\infty$, we
  introduce a reparametrization of time, $t(s):[0,1]\to\RR^+$, to
  compactify the infinite ``physical'' time interval. In particular,
  we choose $t(s)$ such that
  \begin{equation}
    \frac{d}{ds}(\phi \circ t) = \left(\int_0^T |\dot\phi|\,dt\right)^{-1}\,,
  \end{equation}
  and denote $\hat\phi(s) = (\phi\circ t)(s)$.  We then have
  \begin{equation}
    \hat \phi'(s) = \frac{d}{ds}(\phi\circ t)(s) = \frac{dt}{ds}
    \dot\phi\big|_{t=t(s)} 
    \equiv \lambda ^{-1} (s)\dot\phi\big|_{t=t(s)}
  \end{equation}
  where we introduced $\lambda(s) = ds/dt$. As a consequence,
  $\hat\phi(s)$ fulfills the geometric instanton
  equations~(\ref{eq:4infpdf}). These equations can be used to show that, in
  the limit as $T\to\infty$, the initial condition becomes irrelevant,
  and we can consider $\hat\phi(0)=x_*$. First,
  from~(\ref{eq:Freidlin-Wentzell_action}) we know that
  $S_T = \tfrac12\int_0^T|\dot x - b(x)|^2\,dt$, so that in the limit
  $T\to\infty$, we must have $\dot x = b(x)$ for an infinite amount of
  time if the action is to remain finite. As a consequence, for
  $T\to\infty$, the global minimizer will decay towards the unique
  fixed point $x_*$, as all initial points are attracted towards it.

  In order to find the global minimum, we can consider the two
  separate problems of first approaching the fixed point, and
  subsequently leaving again. For this purpose, consider the
  trajectory $\eta(t)$, and corresponding reparametrized trajectory
  $\hat\eta(s)$ with $\dot \eta = \lambda \hat\eta'$, $s\in[-1,1]$,
  and
  \begin{equation}
    \hat\eta(s) = 
    \begin{cases}
      \hat \eta_1(s) & \text{for}\quad s\in[-1,0]\,,\\
      \hat \eta_2(s) = \hat\phi(s) & \text{for}\quad s\in[0,1]\,,
    \end{cases}
  \end{equation}
  where
  \begin{equation}
    \lambda \hat \eta_1' = b(\hat\eta_1),\quad \hat \eta_1(-1)
    = x,\quad \hat\eta_1(0) = x_*\,.
  \end{equation}
  It follows that $\hat\eta(s)$ corresponds to the trajectory that
  deterministically decays into the fixed point $x_*$ starting from
  $x$, and then corresponds to the minimizer $\hat\phi$, solution to
  equations~(\ref{eq:4infpdf}) from then on.

  Since $x_*$ is the unique fixed point and all $x\in\RR^n$ are
  attracted to it, such an $\hat\eta_1$ exists and is unique for all
  $x$. On the other hand, since $(\hat\eta_1, \hat\theta_{\eta_1})$
  fulfill instanton equations on $s\in[-1,0]$, and $\lambda\hat\eta_1'
  = b(\hat\eta_1)$, we have that $\hat\theta_{\eta_1} = 0$ and the
  corresponding action vanishes on the interval
  $s\in[-1,0]$. Therefore, the action associated with $\hat\eta(s)$ is
  equal to the action associated with $\hat\phi(s)$. The problem of
  finding a global minimizer starting at $x$ reduces to the problem of
  solving (\ref{eq:4infpdf}).

  Another consequence of taking the limit $T\to\infty$ is that the
  Hamiltonian is zero along the 
  instanton, i.e. $H(\hat \phi(s),\hat \theta(s)) =0$ $\forall
  s\in[0,1]$.  As a result
  \begin{equation}
    \label{eq:37}
    |a \hat\theta + b(\hat\phi)|^2_a = \<\hat \theta, a\hat \theta\> + 2 \< \hat
    \theta, b(\hat \phi)\> + |b(\hat\phi)|^2_a = 2H(\hat \phi,\hat
    \theta) + |b(\hat\phi)|^2_a = |b(\hat\phi)|^2_a\,.
  \end{equation}
  Since $\dot
  \phi\big|_{t(s)} = \lambda(s) \hat\phi'(s)$, this allows us to
  deduce the following expression for $\lambda$:
  \begin{equation}
    \label{eq:lamb11}
    \lambda(s) =
    \frac{|\dot\phi\big(t(s))|_a}{|\hat\phi'(s)|_a}
    = \frac{|a \hat\theta + b(\hat\phi)|_a}{|\hat\phi'|_a}
    = \frac{|b(\hat\phi)|_a}{|\hat\phi'|_a}\,.
  \end{equation}
  This is the expression stated in the proposition.

  It remains to be shown that the reparametrization of the Riccati
  equation,~(\ref{eq:Winf}), is well posed. To this end, notice that
  as $s\to0$, we have $\hat\phi(s)\to x_*$ and
  $\hat\theta(s)\to0$. Therefore,
  \begin{equation}
    \lambda \hat W' = -(\nabla b(x_*))^\top \hat W
    - \hat W(\nabla b(x_*)) - \hat W a \hat W \quad\text{for}\quad 0<s\ll1\,,
  \end{equation}
  which leads to the conclusion that, since $\nabla b(x_*)$ is
  negative definite by definition, we have $\hat W(s)\to0$ as $s\to0$.
\end{proof}

\subsubsection{Example: Gradient system}
\label{sec:exampl-grad-syst}

An easy example that can be computed explicitly is the case of
diffusion in a potential landscape,
\begin{equation}
  \label{eq:gradient-SDE}
  dX_t^\eps = -\nabla U(X_t^\eps) \,dt + \sqrt{2\eps} dW_t\,,
\end{equation}
where $X\in\RR^n$, and $U:\RR^n \to \RR$ is a potential with a unique
minimum $x_*$, such that $\int_{\RR^n} e^{-\eps^{-1}U(x)} dx < \infty$ for
all $\eps>0$---we can set $U(x_*)=0$ without loss of generality.
Given an observable $f:\RR^n \to \RR$ such that $f(y)$ grows
sufficiently slower
than $U(y)$ as $|y|\to\infty$, we want to obtain a sharp estimate of
the expectation
\begin{equation}
  \label{eq:Beps1}
  B_\eps = \int_{\RR^n} \exp(\eps^{-1} f(y) ) \rho_\eps(y) dy\,,
\end{equation}
where $\rho_\eps(y)$ is the density of the invariant measure
of~(\ref{eq:gradient-SDE}), given by
\begin{equation}
  \label{eq:43}
  \rho_\eps(y) = Z_\eps^{-1} \exp(-\eps^{-1} U(y)) \qquad \text{with}
\qquad Z_\eps = \int_{\RR^n} \exp(-\eps^{-1} U(y))  dy\,.
\end{equation}
Combining~\eqref{eq:Beps1} and~\eqref{eq:43} we see that $B_\eps$ is
given by
\begin{equation}
  \label{eq:Beps2}
  B_\eps = \frac{\int_{\RR^n} \exp(\eps^{-1}( f(y) -U(y))) dy}
  {\int_{\RR^n} \exp(-\eps^{-1}U(y)) dy}\,,
\end{equation}
and both integrals can be estimated by Laplace's method. The result is
that $\lim_{\eps\to0} B_\eps/\bar B_\eps = 1$ with $\bar B_\eps$ given
by
\begin{equation}
  \label{eq:Beps3}
  \bar B_\eps = \left(\frac{\det H(x_*)}{\det\big(H(x_f)-\nabla\nabla f(x_f)\big)}\right)^{1/2} \exp\left(\eps^{-1}\left(f(x_f) - U(x_f)\right)\right)\,,
\end{equation}
where $H(x)=\nabla\nabla U(x)$ is the Hessian of the
potential and the point $x_f\in\RR^n$ is the solution of the optimization problem
\begin{equation}
  \label{eq:xfdef}
  x_f = \argmin_{y\in\RR^n} \left(U(y) - f(y)\right)\,.
\end{equation}
We will now show that Proposition~\ref{thm:expectationinf} yields the
same result.

First, we know explicitly that the instanton fulfills
\begin{equation}
  \label{eq:instantongrad}
  \lambda \hat\phi '(s) = -\nabla U(\hat\phi(s)) \quad\text{and}\quad \hat\theta(s) = \nabla U(\hat\phi(s))\,,
\end{equation}
where primes again denote derivatives with respect to
arclength. Further, the endpoints of the instanton are $x_*=\hat
\phi(0)$ and $x_f=\hat\phi(1)$: the first is by definition, and the
second since the final condition for the equation for $\hat\theta(s)$
gives
\begin{equation}
  \label{eq:44}
  \hat \theta(1) = \nabla f(\hat \phi(1)) = \nabla U(\hat\phi(1))\,,
\end{equation}
where the second equality follows from the second equation
in~\eqref{eq:instantongrad}. Since~\eqref{eq:44} is also the
Euler-Lagrange equation for the minimization problem
in~\eqref{eq:xfdef}, we deduce $x_f = \hat \phi(1)$. Therefore, using
\begin{equation}
  \lambda^{-1}(s) \langle \hat \theta, \hat \theta\rangle = \langle \nabla U(\hat\phi(s)), \hat\phi'(s)\rangle = \frac{d}{ds} \nabla U(\hat\phi(s))
\end{equation}
we deduce that the exponential term in the expectation is given by
\begin{equation}
  \exp\left(\eps^{-1} \left(f(x_f)-\int_0^1 \lambda^{-1}(s)|\hat\theta(s)|^2\,ds\right)\right) = \exp\left(\eps^{-1} \left(f(x_f)-U(_f)\right)\right)\,.
\end{equation}

Second, we have
\begin{equation}
  \nabla\nabla\langle b, \hat\theta\rangle = -\nabla\nabla\nabla U(\hat\phi) \nabla U(\hat \phi) = -\nabla\nabla\nabla U(\hat\phi)\lambda\hat\phi' = \lambda H'\,.
\end{equation}
This means that the Riccati equation~(\ref{eq:Winf}) is here given by
\begin{equation}
  \lambda \hat W' = \lambda H' + H\hat W + \hat WH + 2\hat W^2\,,\qquad \hat W(1)=\nabla\nabla f(x_f)\,.
\end{equation}
Let us look for a solution of the form
\begin{equation}
  \label{eq:WD}
  \hat W = \lambda D^{-1} D'\,,
\end{equation}
for a matrix $D\in\RR^{n\times n}$, with $D(1)=\mathrm{Id}$, and $\lambda(1)
D'(1)=\nabla\nabla f(x_f)$ to fulfill the boundary conditions
for $\hat W$. Equation~\eqref{eq:WD} can be written as $D' = \lambda^{-1} D \hat
W$ which, from Liouville's formula, $\Psi'(s)
= A(s) \Psi(s)$ implies $\det\Psi(s) = \det \Psi(0)
\exp\left(\int_0^s \tr A(\tau)d\tau\right)$ for $A,\Psi \in
\RR^{n\times n}$, yields
\begin{equation}
  \label{eq:D0}
  \det D(0) = \exp\left(\int_0^1 \lambda^{-1}(s)\tr \hat W(s)\,ds\right)\,.
\end{equation}
From equation~\eqref{eq:WD} it also follows that
\begin{equation}
  \lambda \hat W' = \lambda D^{-1}(\lambda D')' - \lambda^2 D^{-1}D'D^{-1}D'\,,
\end{equation}
which we can use in conjunction with the Riccati equation to get
\begin{equation}
  \lambda (\lambda D')' = \lambda (DH)' + \lambda(DH-\lambda D')D^{-1}D'
\end{equation}
or equivalently
\begin{equation}
  \lambda (\lambda D'-DH)' = -(\lambda D'-DH) W\,.
\end{equation}
Using again Liouville's formula (this time for $\Psi = \lambda
D'-DH$) yields the relation
\begin{equation}
  \det\big(\lambda(0) D'(0) - D(0) H(x_*)\big) = \det\big(\lambda(1) D'(1) - D(1)H(x_f)\big) \exp\left(-\int_0^1 \lambda^{-1}(s) \tr \hat W(s)\,ds\right)\,,
\end{equation}
into which we can insert the boundary conditions $\lambda(1) D'(1) =
\nabla\nabla f(x_f)$, $D(1)=\mathrm{Id}$ and $\lambda(0)=0$,
as well as $\det D(0)$ from equation~(\ref{eq:D0}) to obtain
\begin{equation}
  \det H(x_*) = \det\big(H(x_f)-\nabla\nabla f(x_f)\big) \exp\left(-2\int_0^1 \lambda^{-1}(s) \tr \hat W(s)\,ds\right)\,.
\end{equation}
This gives eventually the prefactor contribution,
\begin{equation}
  \hat R = \exp\left(\int_0^1 \lambda^{-1}(s) \tr \hat W(s)\,ds\right) = \left(\frac{\det H(x_*)}{\det\big(H(x_f)-\nabla\nabla f(x_f)\big)}\right)^{1/2}\,.
\end{equation}
Therefore we recover $\bar B_\eps$ given in~\eqref{eq:Beps3}, as
needed.

Note that the above results can be generalized to a diffusion in a
potential landscape with mobility matrix, i.e.~systems of the form
\begin{equation}
  dX = -M \nabla U(X)\,dt + \sqrt{2\eps} M^{1/2}\,dW\,,
\end{equation}
where $M\in \RR^{n\times n}$ is the symmetric, positive definite
mobility matrix, and $M^{1/2}\in \RR^{n\times n}$ is a symmetric
matrix with $M^{1/2}M^{1/2} = M$. The procedure above can be repeated
by replacing $\hat W$ with $M^{-1/2} \hat W M^{-1/2}$ and $H$ with
$M^{-1/2} H M^{-1/2}$. The result~(\ref{eq:Beps3}) is unchanged.

\section{Probability densities}
\label{sec:pdf}

\subsection{Probability densities at finite time}
\label{sec:pdffinitet}

Here we estimate the probability density function
of $X^\eps_t$ in the limit of small $\eps$. Denote this density by
$\rho^x_\eps(t,y)$, so that, given any suitable $F:\RR^n\to\RR$,
\begin{equation}
  \label{eq:11}
  \int_{\RR^n} F(y) \rho^x_\eps(T,y) dx = \EE^x F(X^\eps_T)
\end{equation}
The next proposition shows how to get a sharp estimate of
$\rho^x_\eps(t,y)$ at any $y$ by purely local considerations:
\begin{proposition}
  \label{thm:pdfs}
  Let $(\phi_{x,y}(t),\theta_{x,y}(t))$ solve
  the instanton equations
  \begin{equation}
    \label{eq:4pdf}
    \begin{aligned}
      \dot \phi_{x,y} &= a \theta_{x,y}+b(\phi_{x,y}), \qquad &&\phi_{x,y}(0) = x\,,\\
      \dot \theta_{x,y} &= -(\nabla b(\phi_{x,y}))^\top \theta, &&\phi_{x,y}(T) = y\,,
    \end{aligned}
  \end{equation}
  and let $Q_{x,y}(t)$ be the solution to  the forward Riccati equation
  \begin{equation}
    \label{eq:Q}
    \dot Q_{x,y} =
    Q_{x,y} K_{x,y}
    Q_{x,y} + Q_{x,y} (\nabla b(\phi_{x,y}))^\top +
    (\nabla b(\phi_{x,y})) Q_{x,y}  + a,\quad Q_{x,y}(0)=0\,.
  \end{equation}
  where we denote
  \begin{equation}
    K_{x,y}(t) = \nabla \nabla \<
    b(\phi_{x,y}(t)),\theta_{x,y}(t)\>\,,
  \label{eq:25}
\end{equation}
In addition let $I_x(T,y)$ be defined as
\begin{equation}
  \label{eq:Iexplicit}
  I_x(T,y) = \frac12\int_0^T \langle
  \theta_{x,y}(t),a\theta_{x,y}(t)\rangle\,dt \ge 0\,.
\end{equation}
Then the
  probability density function $\rho^x_\eps(T,y)$ of satisfies
  \begin{equation}
    \label{eq:thm2}
    \lim_{\eps\to0} \frac{\rho^x_\eps(T,y)}{\bar \rho^x_\eps(T,y)}= 
    1 \qquad \text{pointwise in $y$},
  \end{equation}
  where $\bar \rho^x_\eps(T,y)$ is given by
  \begin{equation}
    \label{eq:12}
    \bar \rho^x_\eps(T,y) = (2\pi\eps)^{-n/2}|\det
    Q_{x,y}(T)|^{-1/2}
    \exp\left(\frac12 \int_0^T \tr(K_{x,y}(t) Q_{x,y}(t) )\,dt
      -\frac1{2\eps}I_x(T,y) \right)\,.
  \end{equation}

\end{proposition}

\begin{remark}
  Equation~(\ref{eq:Q}) for $Q_{x,y}$ is structurally equivalent to
  the equation~(\ref{eq:W}) for $W_{x,y}=Q_{x,y}^{-1}$ except for the
  boundary conditions. The boundary conditions necessitate solving
  $Q_{x,y}$ forward in time. Notably, this direction of integration is
  also the one in which the equation for $Q_{x,y}$ is well-posed and
  numerically stable, as can be intuited for example by considering
  the Ornstein-Uhlenbeck process $b(x)=-\gamma x$, for $\gamma>0$ (see
  Sec.~\ref{sec:exampl-one-dimens} below). This feature will become
  even more apparent in the context of stochastic partial differential
  equations, see the discussion after Proposition~\ref{thm:SPDEs} in
  Sec.~\ref{sec:infinited}.
\end{remark}

\begin{remark}
  In Appendix~\ref{sec:riccatiexpect} we discuss how to
  solve~\eqref{eq:Q} via sampling of the solution of a nonlinear (in
  the sense of McKean) SDE.
\end{remark}

Intuitively, the local estimation of of $\rho^x_\eps(t,y)$ implied by
Proposition~\ref{thm:pdfs} is possible because in the limit
$\eps\to0$, this probability density is dominated by the instanton,
and the prefactor contributions can again be estimated from the
Gaussian fluctuations around this instanton. More concretely, we will
see in the proof of Proposition~\ref{thm:pdfs} below (see
\eqref{eq:24}) that $Q_{x,y}(T)$
is the inverse of the Hessian of the action $I_x(T,y)$, i.e.
\begin{equation}
  \label{eq:22}
  Q_{x,y}(T) = [\nabla _y \nabla _y I_x(T,y)]^{-1}
\end{equation}
This implies that we can also write~\eqref{eq:12} as
\begin{equation}
  \label{eq:21}
  \bar \rho^x_\eps(T,y) = (2\pi\eps)^{-n/2}|\det \nabla _y \nabla _y I_x(T,y)|^{1/2}
  \exp\left(\frac12 \int_0^T \tr(K_{x,y}(t) Q_{x,y}(t) )\,dt-\frac1{2\eps}I_x(T,y)\right).
\end{equation}
This form shows that $\bar \rho^x_\eps(T,y) $ is normalized in the
limit as $\eps\to0$, a result we state as
\begin{lemma}
  \label{th:lem1}
  The density $\bar \rho^x_\eps(T,y)$ defined in~\eqref{eq:12} is
  asymptotically normalized, i.e.
  \begin{equation}
  \label{eq:230}
  \lim_{\eps\to0} \int_{\RR^n} \bar \rho_\eps^x(T,y) dy= 1
  \end{equation}
\end{lemma}

\begin{proof}
  Starting from expression~\eqref{eq:21} for $\bar \rho^x_\eps(T,y)$,
  which as we show in the proof of Proposition~\ref{thm:pdfs} is
  equivalent to~\eqref{eq:12}, let us evaluate the integral
  $\int_{\RR^n} \bar \rho_\eps^x(T,y) dy$ by Laplace method. To this
  end, we must first identify the point where $I_x(T,y)$ is
  minimal. It is easy to see that this point is $y = y_x(T)$, where
  $y_x(t)$ is the solution to the ODE $\dot y_x = b(y_x)$,
  $y_x(0)= x$. Indeed, $y_x(t)$ is also the instanton obtained when
  $\theta_{x,y_x(T)}(t)=0$, and from~\eqref{eq:Iexplicit} it gives
  $I_x(T,y_x(T))=0$, which implies that the minimum at $y_x$ is
  necessarily a global minimum. Similarly, $\theta_{x, y_x(T)}(t)=0$
  implies that $K_{x, y_x(T)}(t)=0$. Therefore
\begin{equation}
  \label{eq:23}
  \begin{aligned}
    & \lim_{\eps\to0} \int_{\RR^n} \bar \rho_\eps^x(T,y) dy\\
    &=
    \lim_{\eps\to0} (2\pi\eps)^{-n/2}|\det \nabla _y \nabla _y
    I_x(T,y_x(T) )|^{1/2}
    \int_{\RR^n} \exp\left(-\frac1{2\eps}I_x(T,y)\right) dy\\
  &= (2\pi)^{-n/2}|\det \nabla _y \nabla _y
  I_x(T,y_x(T) )|^{1/2} \int_{\RR^n}
  \exp\left(-\frac1{2}\<u,\nabla_y\nabla_y I_x(T,y_x(T))
    u\> \right) du\\
  & = 1
\end{aligned}
\end{equation}
where to get the second equality we used $u = (z-y_x(T))/\sqrt{\eps}$
as new integration variable and only keep the zeroth order term in
$\eps$ that contributes to the limit.
\end{proof}
  
In the proof of this lemma, only the behavior of
$\bar \rho_\eps^x(T,y)$ near its maximum at $y_x(T)$ matters, but let
us stress that the expression for $\bar \rho_\eps^x(T,y)$ given
in~\eqref{eq:12} offers a much finer approximation of this density
valid arbitrary far away from $y_x(T)$. This will allow us to use
$\bar \rho_\eps^x(T,y)$ to estimate expectations dominated by tail
events, as shown in Sec.~\ref{sec:expectrevisit}.

\begin{proof}[Proof of Proposition~\ref{thm:pdfs}]
  Consider the expectation
  \begin{equation}
    C^x_\eps(\eta) = \EE^x \exp(\eps^{-1}\langle\eta, X^\eps_T\rangle),
    \qquad \eta \in \RR^n\,.
  \end{equation}
  We can alternatively express this expectation as
  \begin{equation}
    \label{eq:pdf-expect}
    C^x_\eps(\eta) =
    \int_{\RR^n} \exp\left(\eps^{-1} \langle\eta ,y\rangle\right)
    \rho^x_\eps(T,y)\,dy\,,
  \end{equation}
  We will show that $\bar \rho^x_\eps(T,y)$ allows us to estimate this
  expectation for all $\eta \in \RR^d$.  Assuming that
  $\rho^x_\eps(T,y)$ is of the form
  \begin{equation}
    \label{eq:14}
    \rho^x_\eps(T,y) = (2\pi\eps)^{-n/2}\left(R_T^x(y) +
    \mathcal O(\eps)\right) 
    \exp(-\eps^{-1} I_x(T,y))\,,
  \end{equation}
  we can evaluate the integral in~\eqref{eq:pdf-expect} by Laplace
  method to obtain
  \begin{equation}
    \label{eq:C1}
    C^x_\eps(\eta) = \left(R_T^x(y) +
    \mathcal O(\eps)\right) |\det \nabla_y\nabla_y I_x(T,y)|^{-1/2}
    \exp(\eps^{-1}(\langle\eta, y\rangle - I_x(T,y)))\,,
  \end{equation}
  where
  \begin{equation}
    y = \argmax_{z\in \RR^n} (\langle\eta, z\rangle- I_x(T,z))\,.
  \end{equation}
  From Proposition~\ref{thm:expectation} we also know that
  \begin{equation}
    \label{eq:C2}
    C^x_\eps(\eta) =
    \left(\exp\left(\frac12 \int_0^T \tr(aW_{x}(t))\,dt\right) +
      \mathcal O(\eps)\right)
    \exp\left(\eps^{-1}\left(\langle\eta, \phi_{x}(T)\rangle
        -\frac12 \int_0^T \langle \theta_{x},a\theta_{x}\rangle\,dt\right)\right)\,,
  \end{equation}
  where $(\phi_{x}(t),\theta_x(t))$ solve the instanton
  equations~\eqref{eq:4} and $W_x(t)$ solves the Riccati
  equation~\eqref{eq:W}; note that, since $f(x) = \< \eta,x\>$
  here, the boundary conditions at $t=T$ reduce to
  $\theta_x(T) = \eta$ and $W_x(T) = 0$. Comparison
  between~\eqref{eq:C1} and~\eqref{eq:C2} implies that $\phi_x(T) = y$,
  i.e.~the instanton equations for
  $(\phi_x(t),\theta_x(t))\equiv(\phi_{x,y}(t),\theta_{x,y}(t)) $
  reduce to the form in~\eqref{eq:4pdf}. In addition,
  $I_x(T,y)$ is given by~\eqref{eq:Iexplicit} and
  \begin{equation}
    \label{eq:Rexplicit}
    R_T^x(y) = |\det \nabla_y \nabla_y I_x(T,y)|^{1/2}
    \exp\left(\frac12\int_0^T \tr(a W_{x}(t))\,dt\right)\,.
  \end{equation}
  Using Proposition~\ref{th:riccati1link} with $K(t) = K_{x,y}(t)$, we
  can identify $W(t)=W_x(t)$, and $Q(t)=Q_{x,y}(t)$ to deduce that
  \begin{equation}
    \label{eq:53}
    \int_0^T \tr(a W_{x}(t)) dt= \int_0^T \tr(K_{x,y}(t) Q_{x,y}(t)) dt
  \end{equation}
  since $W_x(T)=0$ and $Q_{x,y}(0)=0$ and so
  $\det(\text{Id} -W_x(T) Q_{x,y}(T))=\det(\text{Id} -W_x(0)
  Q_{x,y}(0)) = 1$.  Therefore, to finish the proof it remains to evaluate
  $\det [\nabla_y\nabla_y I_x(T,y)]$. To this end, notice first that,
  since $I_x(T,y)$ is given by~\eqref{eq:Iexplicit}, it satisfies
  \begin{equation}
    \label{eq:13}
    \partial_t  I_x(t,y) +H(y,\nabla_y I_x(t,y)) = 0, \qquad
    \lim_{t\to0}t I_x(t,y) = \tfrac12 \< (y-x), a^{-1} (y-x)\>\,.
  \end{equation}
  Indeed, the solution to this equation can be expressed as
  in~\eqref{eq:Iexplicit}, and the boundary condition follows from the
  fact that, for small $T$, to leading order the solution to the
  instanton equations~\eqref{eq:4pdf} reads
  \begin{equation}
    \label{eq:15}
    \phi_{x,y}(t) = x + (y-x) t/T +O(T), \qquad  \theta_{x,y}(t) = a^{-1}
    (y-x)/T +O(T), \qquad t\in[0,T].
  \end{equation}
  This implies that
  \begin{equation}
    \label{eq:16}
    I_x(T,y) = \frac12 \int_0^T \langle
    \theta_{x,y}(t),a\theta_{x,y}(t)
    \rangle\,dt =\tfrac12
    T^{-1} \< (y-x), a^{-1} (y-x)\> +O(1)\quad\text{for}\quad T\ll1\,,
  \end{equation}
  consistent with the boundary condition in~\eqref{eq:13}. Therefore,
  if we introduce
  \begin{equation}
    [\nabla_y\nabla_y I_x(t,\phi_{x,y}(t))]^{-1} =
    Q_{x,y}(t)\label{eq:24}
\end{equation}
so that $\nabla_y\nabla_y I_x(T,y) = Q_{x,y}^{-1}(T)$, then an
equation for $Q_{x,y}(t)$ can be derived by (i) differentiating
\eqref{eq:13} twice with respect to $y$ and evaluating the result at
$y=\phi_{x,y}(t)$ to obtain an equation for
$\nabla_y\nabla_y I_x(t,\phi_{x,y}(t))$, and (ii) using this equation
to derive one for $Q_{x,y}(t)$. The result is the Riccati
equation~\eqref{eq:Q}, in which the boundary condition follows from
$\lim_{t\to0} t \nabla_y \nabla_y I_x(t,y) = \tfrac12 a^{-1}$,
i.e.~$[\nabla_y \nabla_y I_x(t,y) ]^{-1} = O(t)$ and hence
$Q_{x,y}(t)= [\nabla_y \nabla_y I_x(t,\phi_{x,y}(t)) ]^{-1} = O(t)$.
\end{proof}

\subsubsection{Example: The one-dimensional Ornstein-Uhlenbeck process}
\label{sec:exampl-one-dimens}

Consider equation~(\ref{eq:sde}) with $n=1$, $\sigma=1$, and
$b(x)=-x$, i.e.
\begin{equation}
  \label{eq:1dOU}
  dX_t^\eps = -X_t^\eps\,dt + \sqrt{\eps} dW_t\,,\quad X_0^\eps = x\,.
\end{equation}
This one-dimensional and linear case is the easiest possible
non-trivial scenario, and in particular we know explicitly that
\begin{equation}
  \label{eq:ou-exact}
  \rho_\eps^x(T,y) = \sqrt{\frac1{\pi\eps(1-e^{-2T})}}
  \exp\left(-\frac{\left|y-xe^{-t}\right|^2}{\eps(1-e^{-2T})}\right)\,.
\end{equation}
We want to compare this analytical result to the approximation
$\bar\rho_\eps^x(T,y)$ given in equation~(\ref{eq:12}). In fact it
turns out that $\bar\rho_\eps^x(T,y)$ of proposition~\ref{thm:pdfs} is
exact in this case, since there is no higher order contribution in
$\eps$ in the prefactor. To show this, we have the instanton equations
\begin{equation}
  \begin{cases}
    \dot\phi = -\phi + \theta,&\phi(0)=x\,,\\
    \dot\theta = \theta,& \phi(T) = y\,,
  \end{cases}
\end{equation}
which are solved by
\begin{equation}
  \left\{\begin{aligned}
    \phi(t) &= \frac{y-xe^{-T}}{1-e^{-2T}}\left(e^{t-T}-e^{-t-T}\right)\\
    \theta(t) &= \frac{2\left(y-xe^{-T}\right)}{1-e^{-2T}} e^{t-T}\,.
  \end{aligned}\right.
\end{equation}
The exponential estimate therefore yields
\begin{equation}
  \exp\left(-\frac1{2\eps} \int_0^T \theta^2(t)\,dt\right)
  = \exp\left(-\frac{\left|y-xe^{-T}\right|^2}{\eps\left(1-e^{-2T}\right)}\right)\,.
\end{equation}
For the prefactor, we have the Riccati equation
\begin{equation}
    \dot Q = -2Q + 1, \qquad  Q(0)=0\,,
\end{equation}
which is solved by
\begin{equation}
  Q(t) = \frac12(1-e^{-2t})\,.
\end{equation} Therefore
\begin{equation}
  |Q(T)|^{-1/2} = \left(\frac12(1-e^{-2T})\right)^{-1/2}\,,
\end{equation}
and since $\nabla \nabla b(x) = 0$ in this example, we obtain
\begin{equation}
  \bar\rho_\eps^x(T,y) = \sqrt{\frac1{\pi\eps(1-e^{-2T})}}
  \exp\left(-\frac{\left|y-xe^{-T}\right|^2}{\eps(1-e^{-2T})}\right)\,,
\end{equation}
which is precisely the analytical result.

\subsection{Probability density of the invariant measure}
\label{sec:prob-dens-invar}

We can generalize Proposition~\ref{thm:pdfs} to calculate the
probability density of the invariant measure, using again the
procedure of gMAM to compactify physical time: 
\begin{proposition}
  \label{thm:pdfinf}
  Let $(\hat \phi_{y}(s),\hat \theta_{y}(s))$
  solve the geometric instanton equations
  \begin{equation}
    \label{eq:4inf}
    \begin{aligned}
      \lambda \hat \phi_{y}' &= a \hat\theta_{y}+b(\hat\phi_{y}),
      \qquad &&\hat\phi_{y}(0) = x_*,\\
      \lambda \hat \theta'_{y} &= -(\nabla b(\hat\phi_{y}))^\top
      \theta_{y}, &&\hat\phi_{y}(1) = y\,.
    \end{aligned}
  \end{equation}
  where $\lambda = |a \hat\theta_{y}+b(\hat\phi_{y})|/|\phi'_{y}|$. In
  addition let $Q_*$ be the solution to  the Lyapunov equation
    \begin{equation}
      \label{eq:28}
      0 = Q_* (\nabla b(x_*))^\top + (\nabla b( x_*)) Q_*  + a\,,
    \end{equation}
  and let $\hat Q_y(s)$ be the solution to the forward Riccati equation
  \begin{equation}
    \label{eq:Qinf}
     \lambda \hat  Q'_{y}  = \hat Q_{y} \hat K_y \hat Q_{y} +
      \hat Q_{y} (\nabla b(\hat \phi_{y}))^\top  + 
      (\nabla b(\hat \phi_{y})) \hat Q_{y}
      + a,\qquad \hat Q_y(0) = Q_*\,,
    \end{equation}
    where we denote 
\begin{equation}
\hat K_y(s) = \nabla\nabla\<b(\hat \phi_y(s)),\hat
\theta_y(s)\>\,,\label{eq:26}
\end{equation}
Finally, let $\hat I(y)$ be given by 
\begin{equation}
  \label{eq:56}
  \hat I(y) = \frac1{2\eps}  \int_0^1 \lambda^{-1}(s) 
  \langle \hat\theta_{y}(s),a \hat\theta_{y}(s)\rangle ds\,,
\end{equation}
Then the probability density function $\rho_\eps(y)$ of satisfies
\begin{equation}
  \label{eq:thm2inf}
  \lim_{\eps\to0} \frac{\rho_\eps(y)}{\bar \rho_\eps(y)}= 1,
\end{equation}
in which  $\bar \rho_\eps(y)$ is given by either one of the two equivalent expressions
\begin{equation}
  \label{eq:12inf}
  \begin{aligned}
    \bar \rho_\eps(y) & = (2\pi\eps)^{-n/2}|\det \hat Q_{y}(1)|^{-1/2}
    \exp\left(\frac12 \int_0^1 \lambda^{-1}(s) \tr(\hat K_y(s)\hat
      Q_{y}(s))\,ds-\frac1{2\eps} \hat I(y)\right)\\
    & = (2\pi\eps)^{-n/2}|\det Q_*|^{-1/2} \exp\left(-\int_0^1
      \lambda^{-1}(s) \left(\nabla \cdot b(\hat
        \phi_y(s))+\tfrac12\tr(a\hat
        Q^{-1}_{y}(s))\right)\,ds-\frac1{2\eps} \hat I(y)\right).
  \end{aligned}
\end{equation}

\end{proposition}

The function $\hat I(y)$ defined in~\eqref{eq:56} is important: it is
the quasipotential of $y$ with respect to $x_*$. Interestingly
$ \hat Q_y(s) = [\nabla \nabla \hat I(\hat \phi_y(s))]^{-1}$ and so~\eqref{eq:12}
can also be written as
\begin{equation}
  \label{eq:27}
  \begin{aligned}
    \bar \rho_\eps(y) &= (2\pi\eps)^{-n/2}|\det \nabla\nabla\hat
    I(y)|^{1/2} \exp\left(\frac12 \int_0^1 \lambda^{-1}(s) \tr(\hat
      K_y(s)\hat Q_{y}(s)) \,ds-\frac1{2\eps} \hat I(y)\right)\\
      & = (2\pi\eps)^{-n/2}|\det \nabla\nabla \hat I(x_*)|^{1/2}
      \exp\left(- \int_0^1 \lambda^{-1}(s) \hat L(y,s)\,ds
        -\frac1{2\eps} \hat I(y)\right)\,,
  \end{aligned}
\end{equation}
where we defined
\begin{equation}
  \label{eq:57}
  \hat L(s,y) = \nabla
  \cdot b(\hat \phi_y(s))+\tfrac12\tr(a \nabla \nabla \hat I(\hat
  \phi_y(s)))
\end{equation}
The second expression in~\eqref{eq:27} is consistent with the one
derived in~\cite{bouchet-reygner:2016}. Note that we can find yet
another illuminating form to write~(\ref{eq:27}) by use of the
\textit{transverse decomposition} that decomposes the drift $b(x)$
into a gradient of the quasipotential, and a transverse portion
$l(x)$, defined as
\begin{equation}
  \label{eq:transverse-decomposition}
  l(x) = b(x) +\tfrac12a\nabla \hat I(x)\qquad\text{with}\qquad
  \langle \nabla I(x), l(x)\rangle=0\,.
\end{equation}
Such a decomposition is guaranteed to exist under the assumption that
the quasipotential is continuously
differentiable~\cite[chap.~3.4]{freidlin-wentzell:2012}. In terms of
$l(x)$~(\ref{eq:57}) reduces to \mbox{$\hat L(s,y) = \nabla\cdot
  l(\hat\phi_y(s))$}. Note that for a reversible diffusion, i.e.~a
gradient flow, we have $U(x) = 2\hat I(x)$, and so $l(x)=0$ and $\hat
L(s,y)$ vanishes. More generally, in flows that have $\nabla\cdot
l(x)=0$, the invariant measure remains the Gibbs-measure,
$\rho_\eps(y) = (2\pi\eps)^{-n/2} |\det \nabla\nabla \hat
I(x_*)|^{1/2} \exp\left(-\frac1{2\eps} \hat I(y)\right)$, leading the
authors of~\cite{bouchet-reygner:2016} to call the quantity
$\exp\big(-\int_0^1 \lambda^{-1}(s) \nabla\cdot l(x)\,ds\big)$ the
\textit{non-Gibbsianness} of the flow. It is important to remark,
though, that this simplification is only available in the infinite
time case on the invariant measure, as for finite times the
corresponding quantity $\nabla_y\nabla_y I_x(t,y)$ becomes singular at
$t=0$. Further, since $\nabla\cdot l$ is generally not available
explicitly, the first form of~(\ref{eq:27}) is the one most easily
used for numerical computation.

We can use~\eqref{eq:27} to show that $\lim_{\eps\to0}\int_{\RR^n}
\bar \rho_\eps(y) dy = 1$, i.e. the equivalent of Lemma~\ref{th:lem1}
holds in the infinite time limit using the geometric expressions
above. Basically, this is because the dominating point in this
integral is $y=x_*$, for which $\hat \theta_{x_*}(s)=0$ and $\hat
Q_{x_*}(s) = Q_*= \nabla \nabla\hat I^{-1}(x_*)$. This also shows
that, $O(\sqrt{\eps})$ away from $x_*$, $\rho_\eps(y)$ can be
approximated by the Gaussian density with mean $x_*$ and covariance
$\eps Q_*$ given by
\begin{equation}
  \label{eq:29}
  (2\pi\eps)^{-n/2}|\det
    Q_*|^{-1/2}
    \exp\left(-\frac1{2\eps} \< (y-x_*),Q_*^{-1} (y-x_*)\>\right),
\end{equation}
and this density is normalized. However, for locations~$y$ that are
$O(1)$ away from $x_*$, this Gaussian approximation is no longer
valid, and the full expression~\eqref{eq:12inf} must be used as
estimate of the probability density on the invariant measure, as shown
in Sec.~\ref{sec:expectrevisit}.

\begin{proof}[Proof of Proposition~\ref{thm:pdfinf}]
  To get the first equality in~\eqref{eq:12inf}, we can mimic the
  steps in the the proof of Proposition~\ref{thm:pdfs}. The only difference is that
  we need to show the validity of the boundary conditions
  for $\hat Q(0)$ given in equation~(\ref{eq:28}). It was established
  before that, regardless of the initial conditions $x$, the instanton
  $\hat\eta(s)$ first decays into the unique fixed point $x_*$ on
  $s\in[-1,0]$. For $s=0$, we then have $\hat\phi(0)=x_*$,
  $\hat\theta(0)=0$, and $\lambda(0) = |
  b(x_*)|_a/|\hat \phi'(0)|_a = 0$, so that we deduce from the arc-length
  reparametrization of equation~(\ref{eq:Q}), given
  in~(\ref{eq:Qinf}), that
  \begin{equation}
    0 = \hat Q(0) (\nabla b(x_*))^\top  + (\nabla b(x_*))\hat Q(0)  + a\,.
  \end{equation}
  This shows that $\hat Q_y(0) = Q_*$ with $Q_*$ solution
  to~\eqref{eq:28}.

  To get the second equality in~\eqref{eq:12inf}, start
  from~\eqref{eq:Qinf} for $\hat Q_y$ and use
  Liouville's formula to deduce that
  \begin{equation}
    \label{eq:58}
    \begin{aligned}
      \lambda \frac{d}{ds} \log \det \hat Q_y & = \tr \left[ \hat Q_y^{-1}(
      \hat Q_{y} \hat K_y \hat Q_{y} + \hat Q_{y} (\nabla b(\hat
      \phi_{y}))^\top + (\nabla b(\hat \phi_{y})) \hat Q_{y} + a)
    \right]\\
    & = \tr (\hat K_y \hat Q_{y}) + 2\nabla
    \cdot b + \tr (a \hat Q_y^{-1}) 
  \end{aligned}
\end{equation}
Integrating this equation on~$s\in[0,1]$ using $\hat Q_y(0) = Q_*$
gives
\begin{equation}
  \label{eq:59}
  \det \hat Q_y(1) = \det Q_* \exp\left( \int_0^1 \lambda^{-1}(s)
    \left(\tr (\hat K_y(s) \hat Q_{y}) + 2\nabla\cdot b
      (\hat \phi_y(s))
      + \tr (a \hat Q_y^{-1}) \right) \right)
\end{equation}
i.e.
\begin{equation}
  \label{eq:60}
  \begin{aligned}
    & |\det \hat Q_y(1)|^{-1/2} \exp\left( \frac12\int_0^1
      \lambda^{-1}(s) \tr
      (\hat K_y(s) \hat Q_{y}) \right)\\
    & = |\det Q_* |^{-1/2} \exp\left( -\int_0^1 \lambda^{-1}(s)
      \left(\nabla \cdot b(\hat \phi_y(s))+\tfrac12\tr (a \hat
        Q_y^{-1})\right) \right)
  \end{aligned}
\end{equation}
This shows that the two expressions for $\bar\rho_\eps(y)$
in~\eqref{eq:12inf} are identical.
\end{proof}
\begin{remark}
  Note that this is consistent with the intuitive interpretation that
  $\hat Q$ quantifies the covariance of fluctuations around the
  instanton. For $T\to\infty$ the fluctuations will ``thermalize''
  around the fixed point, which corresponds to considering the
  linearized (Ornstein-Uhlenbeck) dynamics around $x_*$, the
  covariance of which solves the Lyapunov equation~(\ref{eq:28}).
\end{remark}

\begin{remark}
  \label{remark:geometric-riccati-initial}
  We seemingly encounter a practical problem with the forward Riccati
  equation~(\ref{eq:Qinf}) at $s=0$, as also discussed
  in~\cite{bouchet-reygner:2021}: Since the instanton remains at
  the fixed-point $x_*$ for an infinite amount of time, we have
  $\lambda(0)=0$, as well as $\hat\theta_y(0)=0$, and thus the forward
  Riccati equation~(\ref{eq:Qinf}) reads $0=0$ at $s=0$---a similar issue
  arises with the equation for $\hat \phi_y$ in~\eqref{eq:4inf} and is
  discussed in Appendix~\ref{sec:init-cond-geom}. This problem is only
  apparent however and the limit of $\hat Q_y'(s)$ as $s\to0$ can be
  straightforwardly obtained by writing~(\ref{eq:Qinf}) as
  \begin{equation}
    \label{eq:Qinfaa}
     \hat Q'_{y} = \lambda^{-1} \left(\hat Q_{y} \hat K_y \hat Q_{y} +
     \hat Q_{y} (\nabla b(\hat \phi_{y}))^\top + (\nabla b(\hat
     \phi_{y})) \hat Q_{y} + a\right),
  \end{equation}
  and sending $s\to0$ at both sides using l'H\^opital rule to compute
  the limit of the right hand side. Accounting for the fact that $\hat
  Q_y(0)= Q_*$ and $\hat K_y(0) =0$ since $\hat\theta_y(0)=0$, the
  result is the following equation for $\hat Q'_y(0) $:
  \begin{equation}
    \label{eq:Qinfaa_}
    \begin{aligned}
      \hat Q'_{y}(0) &= [\lambda'(0)]^{-1} \left(Q_{*} (\nabla \nabla b(x_*)\hat \theta_y'(0)) 
        Q_{*} + \hat Q'_{y}(0) (\nabla b(x_*))^\top + (\nabla b(x_*))
        \hat Q'_{y}(0) \right.\\
      &\quad \left. + Q_{*} (\nabla \nabla b(x_*) \hat
        \phi_y'(0))^\top + (\nabla \nabla b(x_*)\hat \phi_y'(0)) Q_{*}
      \right).
    \end{aligned}
  \end{equation}
  All the quantities in this equation except $\hat Q_y'(0)$ are
  available from the solution of the geometric instanton
  equation~\eqref{eq:4inf} as well as the Lyapunov
  equation~\eqref{eq:28} for $Q_*$, and \eqref{eq:Qinfaa_} can also be
  written as a Lyapunov equation
  \begin{equation}
    \label{eq:Q-prime-Lyapunov}
    \mathfrak C\, \hat Q_y'(0) + \hat Q_y'(0)\, \mathfrak C^\top = \mathfrak K \,,
  \end{equation}
  where we defined
  \begin{equation}
    \begin{aligned}
      \mathfrak C &= \tfrac12 \lambda'(0) \mathrm{Id} - \nabla b(x_*)\\
      \mathfrak K &= Q_{*} (\nabla \nabla b(x_*)\hat \theta_y'(0)) 
      Q_{*} +Q_{*} (\nabla \nabla b(x_*) \hat
      \phi_y'(0))^\top + (\nabla \nabla b(x_*)\hat \phi_y'(0)) Q_{*}\,.
    \end{aligned}
  \end{equation}
  Since $\lambda'(0)>0 $ and $- \nabla b(x_*)$ is positive-definite as
  $x_*$ is a stable fixed point of $\dot x = b(x)$ by assumption, the
  matrix $\mathfrak C$ is positive-definite and invertible, meaning
  that~\eqref{eq:Q-prime-Lyapunov} has a unique solution~$\hat
  Q_y'(0)$. Once $\hat Q_y'(0)$ has been calculated, one can
  initialize the integration of $\hat Q_y(s)$ with arc-length stepsize
  $\Delta s>0$ using e.g.
  \begin{equation}
    \hat Q_y(\Delta s) = Q_* + \Delta s \hat Q_y'(0) +O(\Delta s^2)\,.
  \end{equation}
  For the examples presented in
  sections~\ref{sec:exampl-invar-meas-R2-nonlin}
  or~\ref{sec:exampl-nonl-irrev-R2-nonlin}, it turns out that
  $\nabla\nabla b(x_*)=0$ and thus $Q_y'(0)=0$. Therefore,
  approximating $\hat Q_y(\Delta s)$ by $Q_*$ is correct to
  $\mathcal O(\Delta s^2)$, which is precise enough for the numerical
  scheme we employ, and is therefore what we used in
  practice. However, we note that it is straighforward to go to higher
  order if necessary: for example, the derivation to obtain
  $\hat Q_y''(0)$ when $\hat Q'_y(0)=0$ is given in
  appendix~\ref{sec:geom-ricc-equat}.
\end{remark}

\subsubsection{Example: Invariant measure of gradient diffusions}
\label{sec:exampl-invar-meas-1d-gradient}

Similar to the expectations for gradient diffusion, discussed in
section~\ref{sec:exampl-grad-syst}, we can also derive a formula for
the small $\eps$ approximation of the invariant density of the
gradient diffusion process~(\ref{eq:gradient-SDE}). The result is
known to be 
\begin{equation}
  \label{eq:grad-rho-inf}
  \bar\rho_\eps(y) = (2\pi\eps)^{-n/2} \left(\det H(x_*)\right)^{1/2} \exp(-\eps^{-1} U(y))\,,
\end{equation}
where the only approximation made is on the prefactor $Z_\eps$ that
can be evaluated by Laplace's method (using $U(x_*)=0$):
$\lim_{\eps\to0} Z_\eps = (2\pi\eps)^{n/2} \left(\det
  H(x_*)\right)^{-1/2} $. Let us show that the small $\eps$
approximation of Proposition~\ref{thm:pdfinf} recovers this result,
including the normalization factor.

First, the instanton contribution yields
\begin{equation}
  \begin{aligned}
    \lambda \hat \phi_{y}'  &= 2\hat\theta_{y}-\nabla U(\hat\phi_{y}), \qquad &&\hat\phi_{y}(0) = x_*,\\
    \lambda \hat \theta'_{y} &= H(\hat\phi_{y}) \theta_{y}, &&\hat\phi_{y}(1) = y,
  \end{aligned}
\end{equation}
which is solved by
\begin{equation}
  \lambda \hat\phi_y'(s) = \nabla U(\hat\phi_y(s)) = \hat\theta_y(s)
\end{equation}
so that the exponential large deviation contribution is given by
\begin{equation}
  \hat I(y) = \int_0^1\lambda^{-1}(s) |\hat \theta_y|^2(s)\,ds = \int_0^1 \langle \nabla U(\hat\phi(s)),\hat\phi'(s)\rangle\,ds = U(y)\,,
\end{equation}
which recovers the exponential part of $\bar \rho_\eps(y)$. If we
additionally notice that for a gradient system we know explicitly the
transverse decomposition to be trivial, $l(x)=0$ and thus $\hat
L(s,y)=0$, the second form of (\ref{eq:12inf}) in
proposition~\ref{thm:pdfinf} immediately yields the limiting density
of~(\ref{eq:grad-rho-inf}).

It is instructive to recover the same result by explicitly solving the
Riccati equation. Notice that $\hat K_y(s) = \nabla \nabla \< b(\hat
\phi_y(s), \hat \theta_y(s) \> = - \lambda (s) H'(\hat \phi_y(s))$, so
that forward Riccati equation~(\ref{eq:Qinf}) reduces to
\begin{equation}
  \label{eq:Q-gradient-case}
  \lambda \hat Q_y' = -\lambda\hat Q_y H'\hat Q_y - H\hat Q_y - \hat Q_yH + 2\textrm{Id}\,,\qquad \hat Q_y(0) = Q_*\,.
\end{equation}
We can directly solve the Lyapunov equation~(\ref{eq:28}) for $Q_*$,
\begin{equation}
  Q_* = \hat Q_y(0) = H(x_*)^{-1}\,.
\end{equation}
Given this fact, note that equation~(\ref{eq:Q-gradient-case}) is
solved by
\begin{equation}
  \hat Q_y(s) = H^{-1}(\hat\phi_y(s))\,.
\end{equation}
Therefore, we can read off
\begin{equation}
  |\det \hat Q_y(1))|^{-1/2} = |\det H(y))|^{1/2}\,.
\end{equation}
and
\begin{equation}
  \begin{aligned}
    & \exp\left(\frac12\int_0^1 \lambda^{-1} (s) \tr (\hat K_y(s) \hat Q_y(s)) ds\right)\\
    = &  \exp\left(-\frac12\int_0^1 \tr ( H^{-1}(\hat\phi_y(s)) H'(
      \hat\phi_y(s)) ds\right) = \left| \frac{\det \hat Q_y(0)}{\det
        \hat Q_y(1)}\right|^{1/2} = \left| \frac{\det H(y)}{\det H_* }\right|^{1/2}
  \end{aligned}
\end{equation}
Combining these results we indeed recover exactly the limiting density
given in~(\ref{eq:grad-rho-inf}).

\subsubsection{Example: Invariant measure of a nonlinear, irreversible process in $\RR^2$}
\label{sec:exampl-invar-meas-R2-nonlin}

For all above examples, analytical results were available from
case-specific calculations, since the systems are very simple. The
easiest example of a system were no analytical results can be easily
derived is a system which is nonlinear, and further the drift is not
given by a gradient of a potential. Since the latter is always the
case in 1D, we need to consider a state space of at least dimension
two. In this case, in order to compute expectations, probability
densities, or probabilities with our approach, we need to solve the
corresponding equations, both instanton equation and Riccati
equation, by numerical means.

As a concrete example, consider the non-gradient drift given by
\begin{equation}
  \label{eq:2d-nonlin-drift}
  b(x_1,x_2) = (-\alpha x_1 - \gamma x_1^3 + \beta x_2, -\alpha x_2 -
  \gamma x_2^3 - \beta x_1)\,.
\end{equation}
For positive $\alpha,\gamma\in\RR$, this drift corresponds to a
nonlinear attractive force towards the fixed-point $x_*=(0,0)$. The
parameter $\beta\ne0$ adds a swirl on the drift, so that the total
dynamics are no longer gradient. The behavior of this drift can be
seen by the streamlines in Fig.~\ref{fig:CompTwOU} (left). Note in
particular that the system is not rotationally symmetric in the
presence of cubic terms (i.e.~$\gamma\ne0$).

In order to estimate numerically the density $\bar \rho^\eps(y)$
defined in~\eqref{eq:12inf}, we
need to:
\begin{enumerate}
\item compute the instanton $\hat \phi_y(s)$ connecting $(0,0)$ to
  $y$, as well as the parametrization $\lambda(s)$, using gMAM; and
\item solve the forward Riccati equation (\ref{eq:Qinf}).
\end{enumerate}

\begin{figure}
  \begin{center}
    \includegraphics[height=180pt]{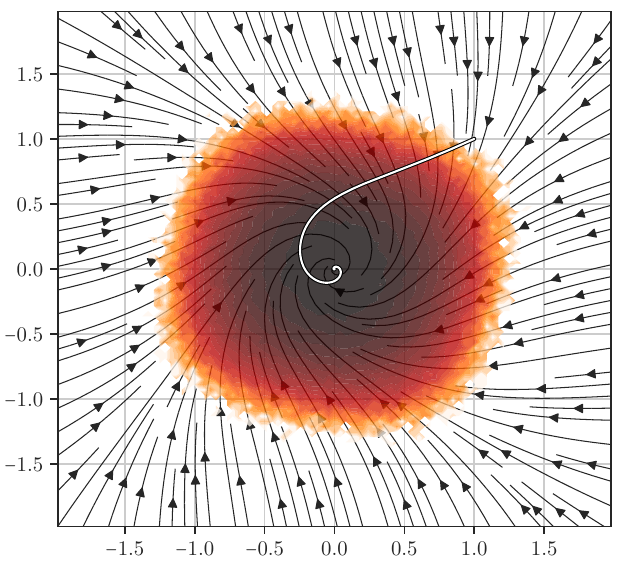}
    \includegraphics[height=180pt]{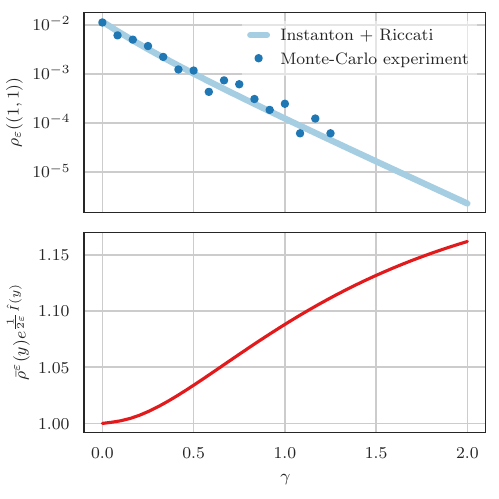}
  \end{center}
  \caption{Left: Instanton (solid curve) connecting the fixed point
    $x_*=(0,0)$ to the point $(1,1)$. The shaded contours indicate the
    invariant measure, obtained from sampling the process for long
    times. The flow field depicts the drift $b(x)$. Here,
    $\alpha=0.5$, $\beta=1$, $\gamma=1$, $\eps=0.25$. Top right:
    Comparison of the pdf $\rho_\eps((1,1))$ between the computation
    via proposition~\ref{thm:pdfinf} (light blue solid) and random
    sampling of the stochastic process (dark blue dots) for varying
    $\gamma\in[0,2]$. The pdf is reproduced at every point from
    instanton and fluctuations, including its normalization
    factor. Bottom right: Prefactor of the invariant density,
    $\bar \rho^\eps(y) e^{\frac12 \eps^{-1} \hat I(y)}$, at $y=(1,1)$
    as a function of the nonlinearity parameter $\gamma$, obtained by
    numerically integrating the Riccati equation along the instanton
    using proposition~\ref{thm:pdfinf}. Clearly, the strength of the
    nonlinearity influences the prefactor. The other parameters are
    again $\alpha = 0.5$, $\beta = 1$.}
  \label{fig:CompTwOU}
\end{figure}
The results of this procedure are shown in
Fig.~\ref{fig:CompTwOU}. For a given arbitrary endpoint $(1,1)$, we
can obtain the invariant density by sampling the process for long
times, and approximating the density by a normalized histogram. The
result is depicted as shaded contours on the left. Note that for this
procedure, not only do we need many samples hitting close to the point
$(1,1)$, but furthermore statistics everywhere else in state space,
because these determine the normalization constant. The instanton
connecting $(0,0)$ to $(1,1)$ is depicted here as a solid line, while
the drift is given by the flowlines.

We now want to establish how the prefactor contribution,
$\bar \rho^\eps(y) e^{\frac12 \eps^{-1} \hat I(y)}$, changes when
changing the parameter $\gamma$, which determines the strength of the
nonlinear term. In particular, for the linear case $\gamma=0$, we know
that
$\bar \rho^\eps(y) e^{\frac12 \eps^{-1} \hat I(y)}=1\ \forall\
y\in\RR^2$, since the invariant measure is Gaussian, and its
covariance can be computed by solving a Lyapunov equation. For the
nonlinear case, $\gamma>0$, this is no longer possible, and we need to
resort to numerical results instead. In the right panel of
Fig.~\ref{fig:CompTwOU} we show the computation of the prefactor through
proposition~\ref{thm:pdfinf}, computing the instanton and solving the
Riccati equation. For $\gamma=0$, this reproduces the known linear
result, but we obtain values for finite values of $\gamma$ as
well. These results are in agreement with results from sampling the
whole process, and looking at the normalized histogram at the point
$(1,1)$ for different $\gamma\in[0,2]$. The numerical parameters here
are $\alpha=0.5, \beta=1, \gamma\in[0,2], \eps=0.25$, and the
histogram binning is done with square bins of side-lengths
$\Delta x=0.02$ and with $10^7$ samples per value of $\gamma$.

\subsection{Calculation of expectations revisited}
\label{sec:expectrevisit}

We can use Proposition~\ref{sec:pdffinitet} to give an expression
alternative to that in Proposition~\ref{thm:expectation} for finite-time
expectations:

\begin{proposition}
  \label{thm:expectationrev}
  Let 
  $(\phi_x(t),\theta_x(t))$ solve the instanton equations in \eqref{eq:4},
  and $Q_x(t)$ be the solution to the forward Riccati equation
  \begin{equation}
    \label{eq:Qrev}
    \dot Q_x = Q_xK_x Q_x
    + Q_x(\nabla b(\phi_x))^\top 
    +(\nabla b(\phi_x)) Q_x +a\,,\qquad Q_x(0) =
    0\,,
  \end{equation}
  where we denote
  $K_x(t) = \nabla \nabla \< b(\phi_x(t)),\theta_x(t) \>$.  Then the
  factor $R(T,x)$ defined in \eqref{eq:30} of
  Proposition~\ref{thm:expectation} can also be expressed as
  \begin{equation}
    \label{eq:30rev}
    R(T,x) = |\det (\text{Id} - \nabla \nabla f(\phi_x(T))
    Q_x(T))|^{1/2} \exp\left(\frac12
      \int_0^T \tr(K_x(t)Q_x(t))\,dt\right).
  \end{equation}
  
\end{proposition}

\begin{proof}
  We can arrive at expression~\eqref{eq:30rev} for $R(T,x)$ in two
  ways. We can use Proposition~\ref{th:riccati1link} with $W_x(t)$
  solution to~\eqref{eq:W} and $Q_x(t)$ solution to~\eqref{eq:Qrev} to
  deduce from~\eqref{eq:39} that
  \begin{equation}
    \label{eq:55}
     |\det (\text{Id} - \nabla \nabla f(\phi_x(T))
    Q_x(T))|^{1/2} \exp\left(\frac12
      \int_0^T \tr(K_x(t)Q_x(t))\,dt\right)= \exp\left(\frac12
      \int_0^T \tr(aW_x(t)\,dt\right)
  \end{equation}
  The left hand side is~\eqref{eq:30rev} and the right hand side is
  \eqref{eq:30}, showing that these two equations contain identical
  expressions.  Alternatively, we can evaluate
  \begin{equation}
    \label{eq:54}
    \int_{\RR^n} \exp(\eps^{-1} f(x)) \bar \rho^\eps(T,x)
      dx
  \end{equation}
  by  Laplace method, using the expression in~\eqref{eq:12} for $\bar
  \rho^\eps(T,x)$. This calculation show that this integral is
  asymptotically equivalent to $\bar A^\eps(T,x)$ in~\eqref{eq:3} with
  $R(T,x)$ given by \eqref{eq:30rev} instead of~\eqref{eq:30rev},
  thereby establishing again that these two equations give the same $R(T,x)$
  since this factor is independent on $\eps$. 
\end{proof}

Similarly, we can use Proposition~\ref{thm:pdfinf} to give an expression
alternative to that in Proposition~\ref{thm:expectationinf} for
expectations on the invariant measure:

\begin{proposition}
  \label{thm:expectationrevinf}
  Let $(\hat \phi(s),\hat \theta(s))$ solve the instanton equations
  in~\eqref{eq:4infpdf}, and $\hat Q(s)$ be the solution to the
  forward Riccati equation
  \begin{equation}
    \label{eq:Qrevinf}
    \lambda (s) \hat Q' = \hat Q\hat K \hat Q
    + \hat Q(\nabla b(\hat \phi))^\top 
    +(\nabla b(\hat \phi)) \hat Q +a\,,\qquad \hat Q(0) =
    Q_*\,.
  \end{equation}
  where $Q_*$ solves~\eqref{eq:28} and we denote
  $\hat K(s) = \nabla \nabla \< b(\hat \phi(s)),\hat \theta(s)
  \>$. Then the factor $\hat R$ defined in~\eqref{eq:32} of
  Proposition~\ref{thm:expectationinf}, can also be expressed as
  \begin{equation}
    \label{eq:30revinf}
    \hat R = |\det (\text{Id} - \nabla \nabla f(\hat \phi(1))
    \hat Q(1))|^{1/2} \exp\left(\frac12
      \int_0^1 \lambda^{-1}(s) \tr(\hat K(s)\hat Q(s))\,ds\right).
  \end{equation}
\end{proposition}
The proof of this proposition is similar to that of
Proposition~\ref{thm:expectationrev}.

\section{Probabilities}
\label{sec:prob}

\subsection{Probabilities at finite time}
\label{sec:prob-finite-t}

Having access to pointwise estimates of the probability density allows
us to estimate probabilities by integration. For example, let
$f:\RR^n\to\RR$ be some observable, and assume we want to compute the
probability that $f(X_T^\eps)$ exceeds some value $a\in \RR$. This
probability can be expressed as
\begin{equation}
  \label{eq:Pa}
  \PP^x(f(X^\eps_T)\ge a  ) = \int_{f(y)\ge a} \rho^x_\eps(T,y)\,dy\,,
\end{equation}
and calculating for various values of $a$ gives the complementary
cumulative distribution function (aka tail distribution) of the random
variable $f(X_T^\eps)$. For concreteness, we will focus on the
calculation of this tail probability for one value of $a$ which,
without loss of generality, we can set to zero. In this case,
\eqref{eq:Pa} can also be interpreted as the probability that the
stochastic process~(\ref{eq:sde}) hits the set
\begin{equation}
  \label{eq:setAdef}
  A = \{z\in\RR^n | f(z) \ge 0\}\,,
\end{equation}
at time $t=T$. This is a problem interesting in its own right, and we
will denote this probability as
\begin{equation}
  \label{eq:PA}
  P_\eps^A(T,x) = \PP^x(X^\eps_T\in A) = \int_A \rho^x_\eps(T,y)\,dy\,.
\end{equation}
To make the problem interesting, we will work under:
\begin{assumption}
  \label{as:set}
  The function $f:\RR^n\to\RR$ is in $C^2$, $x\notin A$ (i.e.~$f(x)<0$),
  and the vector field $b$ points outward $A$ everywhere on $\partial
  A$ (i.e.~$\langle\nabla f(z),b(z)\rangle <0 \ \forall\ z \in \partial A
  = \{ z\in\RR^n | f(z)=0\}$).
\end{assumption}
\noindent
This assumption implies that the event $X^\eps_T$ is noise-driven, and
as a result $P_\eps^A(T,x) \to0$ as $\eps\to0$, which is the
nontrivial case. We also need
\begin{assumption}
  The point on $\partial A$ with minimal $I_x(T,z)$,
  \begin{equation}
    y = \argmin_{z\in \partial A} \int_0^T
    \langle \theta_{x,z}(t), a \theta_{x,z}(t)\rangle\,dt
  \end{equation}
  is unique, and so is the instanton leading to it.
\end{assumption}
\noindent
Intuitively we demand that the set $A$ does not admit multiple
distinct points that can be reached by competing instantons of
identical action. Then, we have:
\begin{proposition}
  \label{thm:probability}
  Let $(\phi_{x,y}(t),\theta_{x,y}(t))$
  solve the instanton equations in~\eqref{eq:4} with $y$ specified as
  \begin{equation}
    \label{eq:ydef}
    y = \argmin_{z\in \partial A} \int_0^T
    \langle \theta_{x,z}(t), a \theta_{x,z}(t)\rangle\,dt,\,
  \end{equation}
  and let $Q_{x,y}(t)$ solve the forward Riccati equation
  in~\eqref{eq:Q}. Let also $F(T,x)$ and $V(T,x)$ be defined as
  \begin{equation}
    \label{eq:33}
    \begin{aligned}
      F(T,x) &= \text{Id} - |\theta_{x,y}(T)| |\nabla f(y)|^{-1} Q_{x,y}(T) \nabla\nabla f(y),\\
      V(T,x) &= \langle\theta_{x,y}(T), Q_{x,y}(T) \theta_{x,y}(T)\rangle^{-1/2}
      \exp\left(\frac12\int_0^T \tr( K_{x,y}(t) Q_{x,y}(t))\,dt\right)\,,
    \end{aligned}
  \end{equation}
  where
  $K_{x,y}(t)= \nabla \nabla \<
  b(\phi_{x,y}(t)),\theta_{x,y}(t)\>$. Finally denote by
  $\hat n = \nabla f(y)/|\nabla f(y)|$ the inward pointing surface
  normal of $A$ at $y$. Then the probability $P^A_\eps(T,x)$
  satisfies
  \begin{equation}
    \label{eq:Paas}
    \lim_{\eps\to0}\frac{P^A_\eps(T,x)}{\bar P^A_\eps(T,x)}=
    1,
  \end{equation}
  in which 
  \begin{equation}
    \label{eq:17}
    \bar P^A_\eps(T,x) =  (2\pi)^{-1/2} \eps^{1/2} \left(\frac{\langle
        \hat n, F(T,x)\hat n\rangle}{\det F(T,x)}\right)^{1/2} V(T,x)
    \exp\left(-\frac1{2\eps}
      \int_0^T \langle \theta_{x,y}(t), a \theta_{x,y}(t)\rangle\,dt \right)\,.
  \end{equation}
  In addition , the unit normal $\hat n$ can be expressed as
  $\hat n= \theta_{x,y}(T)/|\theta_{x,y}(T)|$.
\end{proposition}

\begin{figure}
  \begin{center}
    \includegraphics[width=200pt]{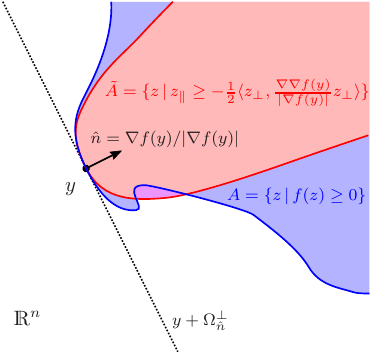}
  \end{center}
  \caption{Schematic representation of the situation in
    proposition~\ref{thm:probability}. To obtain the probability to
    hit the set $A$, we replace the integral of the density over $A$
    with an integral over the paraboloid $\tilde A$, which is
    tangential $A$ at the point of maximum likelihood $y$ and shares
    its curvature.}
  \label{fig:paraboloid}
\end{figure}
\begin{remark}
  As we will see in the proof, the factor involving $F(T,x)$ describes the
  effect of the geometry (specifically, the curvature) of the set $A$
  around the maximum likelihood hitting point $y$: If the set is
  planar, then $\nabla\nabla f(y) = 0$ and the factor evaluates to
  unity. Equivalently, this factor disappears for linear observables
  $f$. As will also be clear from the proof, in the definition of
  $F(T,x)$ in~\eqref{eq:33} we can replace $|\nabla f(y)|^{-1} \nabla
  \nabla f(y)$ by $\nabla \hat n(y)$, i.e.  in~\eqref{eq:17} we can
  replace $F(T,x)$ with
  \begin{equation}
    \label{eq:33b}
      F_\perp(T,x) = \text{Id} - |\theta_{x,y}(T)|   Q_{x,y}(T) \nabla \hat n(y)
  \end{equation}
  This is because
  $\nabla \hat n=|\nabla f|^{-1} \nabla \nabla f + |\nabla f)|^{-1}
  \hat n (\nabla \nabla f\hat n)^T$ and the extra factor
  $ |\nabla f|^{-1} \hat n (\nabla \nabla f \hat n)^T$ does not
  contribute to the ratio $\<\hat n, F(T,x)\hat n\>|\det F(T,x)|^{-1}$
  in~\eqref{eq:17}. This shows that this expression is intrinsic to
  the set $A$, i.e. it does not depend on the way we parametrize its
  boundary by the zero level-set of $f$, as it should be. For
  computational purposes, using the parametrized version
  in~\eqref{eq:33} is more convenient, however.
\end{remark}

\begin{proof}
  Let us evaluate the integral in~\eqref{eq:PA} by Laplace method
  using the ansatz~\eqref{eq:14} for $\rho_\eps^x(T,y)$. To this end
  notice that the point $y$ specified in~\eqref{eq:ydef} is also given
  by
  \begin{equation}
    y = \argmin_{z\in \partial A} I_x(T,z) \in \partial A\,.
  \end{equation}
  Notice also that $\theta_{x,y}(T) = \nabla_y I_x(T,y)$, and that
  $\hat n = \theta_{x,y}(T)/|\theta_{x,y}(T)|$ is the inward pointing
  unit normal to $\partial A$ at $y$. Therefore, to perform the
  integral
  \begin{equation}
    \bar P_\eps^A(T,x) = \int_A \bar\rho_\eps^x(T,z)\,dz\,,
  \end{equation}
  we split the integration variable into components
  parallel to $\hat n$ and perpendicular to $\hat n$, as
  \begin{equation}
    z-y = \eps z_\parallel \hat n + \sqrt{\eps} z_\perp\,,
  \end{equation}
  where $z_\parallel\in\RR$ and $z_\perp \in \Omega_{\hat n}^\perp$
  with
  \begin{equation}
    \Omega_{\hat n}^\perp = \{z\in\RR^d : \langle\hat n, z\rangle = 0\}\,.
  \end{equation}
  For any $z\in A$ we have $f(z)\ge0$. Further using $f(y)=0$, we can
  expand $f$ around $y$ to obtain
  \begin{equation}
    f(z) = f(y) + \sqrt{\eps}\langle z_\perp, \nabla f(y)\rangle
    + \eps z_\parallel \langle \hat n, \nabla f(y) \rangle
    + \frac{\eps}2 \langle z_\perp, \nabla\nabla f(y) z_\perp\rangle
    + \mathcal O(\eps^{3/2}) \ge 0\,.
  \end{equation}
  Since $\langle z_\perp, \nabla f(y)\rangle = 0$ by definition, we
  have for points in $A$ around $y$ that
  \begin{equation}
    z_\parallel \ge -\tfrac12\langle z_\perp, \frac{\nabla\nabla f(y)}{|\nabla f(y)|} z_\perp\rangle\,.
  \end{equation}
  Effectively, to the relevant order in $\eps$, $A$ can be approximated by a paraboloid
  \begin{equation}
    \tilde A = \{z\in\RR^n\,|\, z_\parallel \ge -
    \tfrac12\langle z_\perp, |\nabla f(y)|^{-1}\nabla\nabla f(y) z_\perp\rangle\}\,,
  \end{equation}
  where $\nabla\nabla f(y)/|\nabla f(y)|$ is the (normalized)
  curvature of the set $A$ at $y$.
  
  We can now evaluate the integral in~\eqref{eq:PA} by Laplace method, 
  \begin{equation}
  \begin{aligned}
    & \bar P_\eps^A(T,x) \\
    &= (2\pi\eps )^{-n/2} R_\eps(y)\exp(-\eps^{-1} I_x(T,y)) \
    \eps^{(n-1)/2} \int_{\Omega_{\hat n}^\perp} e^{-\tfrac12 \langle
      z_\perp, \nabla\nabla I_x(T,y)
      z_\perp\rangle} \\
    & \qquad \times \eps \int_{-\tfrac12\langle z_\perp, |\nabla f(y)|^{-1}\nabla\nabla f(y) z_\perp\rangle}^\infty e^{-z_\parallel |\nabla I_x(T,y)|}dz_\parallel\,dz_\perp \,\\
    &= (2\pi)^{-1/2} \eps^{1/2} R_\eps(y)\exp(-\eps^{-1} I_x(T,y))
    |\nabla I_x(T,Y)|^{-1} \\
    & \qquad \times \int_{\Omega_{\hat n}^\perp} e^{-\tfrac12
      \langle z_\perp, \nabla\nabla I_x(T,y)
      z_\perp\rangle + \tfrac12\langle z_\perp, |\nabla f(y)|^{-1}|\nabla I_x(T,y)|\nabla\nabla f(y) z_\perp\rangle}\,dz_\perp \,\\
    &= (2\pi)^{-1/2} \eps^{1/2} R_\eps(y) \left|D(T,y)\right|^{-1/2}
    |\theta_{x,y}(T)|^{-1} \exp(-\eps^{-1} I_x(T,y))\,,
  \end{aligned}\label{eq:19}
\end{equation}
  Here
  \begin{equation}
    \label{eq:18}
    \begin{aligned}
      R_\eps(y) &= \left(|\det Q_{x,y}(T)|^{-1/2}
        \exp\left(\frac12\int_0^T\tr(K_{x,y}(t) Q_{x,y}(t)))\,dt\right)+\mathcal
        O(\eps)\right)\\
      &= \left(|\det \nabla\nabla I_x(T,y)|^{1/2}
        \exp\left(\frac12\int_0^T\tr(K_{x,y}(t) Q_{x,y}(t)))\,dt\right)+\mathcal
        O(\eps)\right)
    \end{aligned}
  \end{equation}
  and
  \begin{equation}
    \label{eq:36}
    D(T,y) = \det_\perp
    \left(\nabla\nabla I_x(T,y) - |\nabla f(y)|^{-1}|\nabla
      I_x(T,y)|\nabla\nabla f(y) \right)
  \end{equation}
  where $\det_\perp$ is defined as follows: Given an invertible,
  positive definite, symmetric $H\in \RR^{n\times n}$ and a unit
  vector $\hat n$ we define $\det_\perp H$ via
  \begin{equation}
    \label{eq:detperp}
    (2\pi)^{(n-1)/2} \left| {\det_\perp} H \right|^{-1/2}
    =\int_{\Omega_{\hat n}^\perp}
    e^{-\frac12 \< y,H y\>} dy\,.
  \end{equation}
  In other words, $\det_\perp$ is the determinant evaluated only in
  the space $\Omega_{\hat n}^\perp$ perpendicular to a
  vector~$\hat n$. As shown in Appendix~\ref{appendix:comp_det_perp}
  we have
  \begin{equation}
    \det_\perp H = \langle \hat n, H^{-1} \hat n\rangle \det H\,,
  \end{equation}
   so that
  \begin{equation}
    \begin{aligned}
      &|\det \nabla\nabla I_x(T,y)|^{1/2} \left|D(T,y)\right|^{-1/2}\\
      &=\left|\det_\perp
        \left(1 - |\nabla f(y)|^{-1}|\nabla I_x(T,y)| Q_{x,y}(T) 
          \nabla\nabla f(y) \right)\right|^{-1/2}
      \langle \hat n, Q_{x,y}(T) \hat n\rangle^{-1/2}
    \end{aligned}
  \end{equation}
  and thus, inserting $R_\eps$ in~\eqref{eq:19}, this equation
  gives~\eqref{eq:Paas}. Finally, note that by definition of
  $\det_\perp$ we can replace
  $|\nabla f(y)|^{-1}\nabla\nabla f(y) $ by
  $\nabla \hat n(y)$ in \eqref{eq:36}, which justifies the alternative
  expression in~\eqref{eq:33b} for $F(T,x)$.
\end{proof}

\subsection{Probabilities on the invariant measure}
\label{sec:probinf}

The equivalent of~\eqref{eq:PA} at infinite time is
\begin{equation}
  \label{eq:PAinf}
  P_\eps^A= \int_A \rho_\eps(y)\,dy\,.
\end{equation}
where $\rho_\eps(y)$ is the invariant density defined by~\eqref{eq:9}
and the set $A$ is defined as before. Then we have:
\begin{proposition}
  \label{thm:probabilityinf}
  Let
  $(\hat \phi_{y}(s),\hat \theta_{y}(s))$ solve the instanton
  equations in~\eqref{eq:4inf} with $y$ specified as
  \begin{equation}
    \label{eq:ydefinf}
    y = \argmin_{\hat\phi_{y}(1)\in \partial A} \int_0^1\lambda^{-1}(s)
    \langle \hat\theta_{y}(s), a \hat\theta_{y}(s)\rangle\,ds;
  \end{equation}
  and let $\hat Q_{y}(s)$ solve the forward Riccati equation
  in~\eqref{eq:Qinf}. Let also $\hat F$ and $\hat V$ be given by
  \begin{equation}
    \label{eq:33inf}
    \begin{aligned}
      \hat F &= \text{Id}- |\hat\theta_{y}(1)| |\nabla
      f(y)|^{-1}\hat Q_{y}(1) 
      \nabla\nabla f(y),\\
      \hat V &= \langle\hat\theta_{y}(1), \hat Q_{y}(1) \hat
      \theta_{y}(1)\rangle^{-1/2} \exp\left(\frac12\int_0^1
        \lambda^{-1}(s) \tr(\hat K_{y}(s) \hat Q_{y}(s)))\,dt\right),
    \end{aligned}
  \end{equation}
  where
  $\hat K_{y}(s) = \nabla \nabla \< b(\hat \phi_{y}(s)),\hat
  \theta_{y}(s)\>$. Finally denote by
  $\hat n = \nabla f(y)/|\nabla f(y)|$ the inward pointing surface
  normal of $A$ at $y$.  Then the probability $P^A_\eps$ satisfies
  \begin{equation}
    \label{eq:Paasinf}
    \lim_{\eps\to0}\frac{P^A_\eps}{\bar P^A_\eps}=
    1,
  \end{equation}
  where
  \begin{equation}
    \label{eq:17inf}
    \bar P^A_\eps =  (2\pi)^{-1/2} \eps^{1/2} \left(\frac{\langle
        \hat n, \hat F\hat n\rangle}{\det \hat F}\right)^{1/2} \hat V 
    \exp\left(-\frac1{2\eps}
      \int_0^1 \lambda^{-1}(s)\langle \hat\theta_{y}(s), a \hat
      \theta_{y}(s)\rangle\,ds
    \right)\,.
  \end{equation}
  In addition , the unit normal $\hat n$ can be expressed as
  $\hat n= \hat \theta_{y}(1)/|\hat \theta_{y}(1)|$.
\end{proposition}

The proof is similar to that of Proposition~\ref{thm:probability}.

\subsubsection{Example: Nonlinear, irreversible process in $\RR^2$ revisited}
\label{sec:exampl-nonl-irrev-R2-nonlin}

\begin{figure}
  \begin{center}
    \includegraphics[height=170pt]{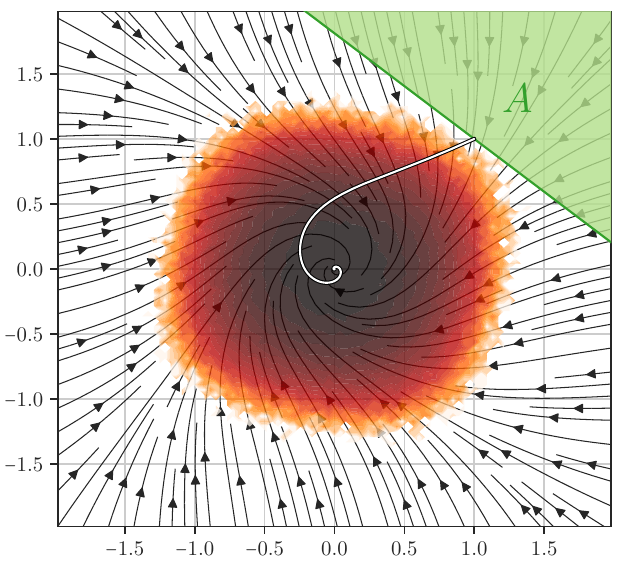}
    \includegraphics[height=170pt]{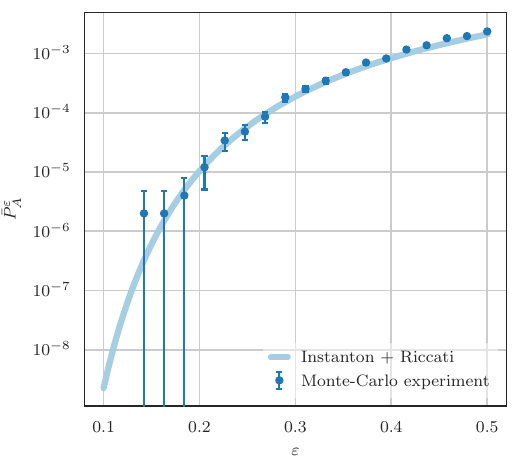}
  \end{center}
  \caption{Left: Instanton (solid curve) connecting the fixed point
    $x_*=(0,0)$ to the set $A$, with minimal action at boundary point
    $(1,1)$. The shaded contours indicate the invariant measure,
    obtained from sampling the process for long times, and the set $A$
    is depicted in green. The flow field depicts the drift
    $b(x)$. Here, $\eps=0.25$. Right: Comparison of the probability to
    hit the set $A$ between the computation via
    proposition~\ref{thm:probabilityinf} (light blue solid) and random
    sampling of the stochastic process (dark blue dots) for varying
    $\eps$. The probability is reproduced at every point from
    instanton and fluctuations, including its normalization
    factor. The parameters are $\alpha=0.5, \beta=1, \gamma=0.5$.}
  \label{fig:2d-nonlinear-prob}
\end{figure}

To show the applicability of these propositions to an actual system,
we re-use the nonlinear, irreversible process in $\RR^2$, as defined
in equation~(\ref{eq:2d-nonlin-drift}), as a simple example for which
the solution is not readily accessible by analytical
considerations. Instead of computing the invariant density, as in
section~\ref{sec:exampl-invar-meas-R2-nonlin} we use
proposition~\ref{thm:probabilityinf} to compute the probability of
hitting the set $A$ on the invariant measure, where $A$ is the
half-space defined by
\begin{equation}
  A = \{ x \in \RR^2 | \langle\hat n, x - (1,1)\rangle \ge 0\}\,,
\end{equation}
with $\hat n \approx (0.6304, 0.7762)$ chosen specifically such that the
point $y=(1,1)$ is the maximum likelihood point on $\partial
A$. Because of this, the solutions to the instanton equations and the
 Riccati equation do not need to be re-computed, and we can insert
their results into equation~(\ref{eq:17inf}) to obtain the instanton
and prefactor estimate for the hitting probability on the invariant
measure.

The results of this experiment are shown in
Fig.~\ref{fig:2d-nonlinear-prob}: While the dynamics and instanton
in Fig.~\ref{fig:2d-nonlinear-prob} (left) look identical to the
original problem in section~\ref{sec:exampl-invar-meas-R2-nonlin}, we
are now trying to estimate the probability of hitting the set $A$
denoted by the light green shading. The result shown in the right
panel of Fig.~\ref{fig:2d-nonlinear-prob} confirms that the asymptotic
prediction of proposition~\ref{thm:probabilityinf} agrees with the
Monte-Carlo simulations for different values of $\eps$. The parameters
are $\alpha=0.5, \beta=1, \gamma=0.5$, and we are taking
$N_{\text{samples}} = 10^6$ samples each value of $\eps$.

\section{Exit probabilities and mean first passage times}
\label{sec:mfpt}

Let $A^c \subset \RR^n$ be the complement of the set $A$ defined
in~\eqref{eq:setAdef} i.e.
\begin{equation}
  A^c = \{ z \in \RR^n \,|\, f(z) < 0\}
\end{equation}
Under Assumption~\ref{as:set}, we know from the argument given after
Proposition~\ref{thm:pdfinf} that
\begin{equation}
  \label{eq:34}
  \lim_{\eps\to0} \int_{A^c} \rho_\eps(y) dy = \lim_{\eps\to0} \int_{A^c} \bar
  \rho_\eps(y) dy = 1
\end{equation}
We wish to estimate the exit probability
\begin{equation}
  \label{eq:35}
  P^{\partial A}_\eps= \int_{\partial A} \rho_\eps(y) d\sigma(y) 
\end{equation}
in the limit as $\eps\to0$, since this surface integral enters the
limiting expression for the mean first passage time (MFPT) of the
process to set~$A$. We have

\begin{proposition}
  \label{thm:probflux}
  Let  $(\hat \phi_{y}(s),\hat \theta_{y}(s))$ solve to the instanton
  equations in~\eqref{eq:4inf} with $y$ specified as
  \begin{equation}
    \label{eq:ydefmfpt}
    y = \argmin_{z\in \partial B} \int_0^1\lambda^{-1}(s)
    \langle \hat\theta_{z}(s), a \hat\theta_{z}(s)\rangle\,ds;
  \end{equation}
  and let $\hat F$ and $\hat V$ be given by~\eqref{eq:33inf}. Denote
  by $\hat n = \nabla f(y)/|\nabla f(y)|$ the inward pointing unit
  normal vector on $\partial A$ at $y$.  Then the exit probability
  $P^{\partial A}_\eps$ satisfies
  \begin{equation}
    \label{eq:mfpt1}
    \lim_{\eps\to0}\frac{P^{\partial A}_\eps}{\bar P^{\partial A}_\eps}=
    1,
  \end{equation}
  where
  \begin{equation}
    \label{eq:17mfpt}
    \bar P^{\partial A}_\eps =  (2\pi\eps)^{-1/2} |\hat\theta_{y}(1)|\left(\frac{\langle
        \hat n, \hat F\hat n\rangle}{\det \hat F}\right)^{1/2}\hat V
    \exp\left(-\frac1{2\eps}
      \int_0^1 \lambda^{-1}(s)\langle \hat\theta_{y}(s), a \hat
      \theta_{y}(s)\rangle\,ds
    \right)
  \end{equation}
 
\end{proposition}

\begin{proof}
  The proof has the same ingredients as the proof of
  proposition~\ref{thm:probabilityinf}, except instead of integrating
  over the paraboloid approximation $\tilde A$, we integrate over its
  boundary $\partial\tilde A$. No normal integral is performed and
  therefore the scaling in $\eps$ and $|\hat\theta_{y}(1)|$ differs.
\end{proof}

With this knowledge, we can now compute exit times. Consider the entrance
into the domain $A$ for a particle
starting at $x\in A^c$. The mean exit time $T_\epsilon:A^c \to \RR^+$
is given by the solution of the problem
\begin{equation}
  \begin{aligned}
    L_\eps T_{\epsilon} &= -1, \quad &&\text{if} \ \ x\in
    A^c, \label{eq:basic_exit_time} \\ T_{\epsilon} &= 0,
    \qquad && \text{if} \ \ x \in {\partial A}\,,
  \end{aligned}
\end{equation}
where $L_\eps$ is the generator of our process defined in~\eqref{eq:5}. In the
following, we quickly review the standard approach to use a boundary
layer expansion in order to approximate the exit time
$T_{\epsilon}$. Multiplying the first equation by the invariant
measure $\rho$ and integrating we find---after applying Green's
theorem and making use of the fact that
$L^*_\eps\rho_\eps=0$~\cite{gardiner:2009}:
\begin{equation} \label{eq:sol_condition}
  \epsilon \int_{\partial A} \rho_\eps \langle\hat n,
  \nabla T_\epsilon\rangle\,d\sigma = - \int_{A^c} \rho_\eps\, dx\,.
\end{equation}
Here, $\hat n(z) = \nabla f(z)/|\nabla f(z)|$ denotes the inward
pointing normal unit vector of the surface $\partial{A}$ at $z\in
\partial{A}$. An approximation of $\langle \hat n(z), \nabla
T_\eps(z)\rangle$ for $z \in \partial {A}$ can be obtained by boundary
layer analysis. For this purpose, we first expand
\begin{equation} \label{eq:T_and_tau}
T_{\epsilon}(x) = {e}^{C/\epsilon}\tau(x)
\end{equation}
such that we find from (\ref{eq:basic_exit_time}) the corresponding
equation for $\tau$
\begin{equation}
L_\eps \tau = {e}^{-C/\epsilon}\approx 0\,.
\end{equation}
Since we assumed that $\langle\hat n(z), b(z)\rangle < 0$ for all
$z\in\partial A$, the appropriate scaling of the boundary layer is $x
= z - \epsilon\eta\hat n $. In the scaled variables, the equation for
$\tau$ becomes
\begin{equation}
-\langle b(z), \hat n(z)\rangle \tau_{\eta} + \tau_{\eta\eta} = 0
\end{equation}
with the solution
\begin{equation}
\tau = \tilde C\left(1-{e}^{\langle\hat n(z), b(z)\rangle \eta}\right)
\end{equation}
This means that we obtain for the exit time $T_{\epsilon}$ the
expression
\begin{equation}
T_{\epsilon} = \tilde C\,
{e}^{C/\epsilon}\,\left(1-{e}^{\langle\hat n(z), b(z)\rangle
  \eta}\right)
\end{equation}
and therefore
\begin{equation}
\nabla T_{\epsilon} = \frac{\tilde C}{\epsilon} \,
       {e}^{C/\epsilon}\, \langle\hat n(z), b(z)\rangle\hat n(u)
\end{equation}
which we can use in the solvability condition
(\ref{eq:sol_condition}). From there, for $x$ away from the boundary
layer, we obtain
\begin{equation}
  \label{eq:Teqn}
T_{\epsilon}(x) = \tilde C\, {e}^{C/\epsilon} = \frac{-
  \int_{A^c} \rho_\eps(z) \, dz}{ \int_{\partial A} \rho_\eps (z)
  \langle\hat n(z), b(u)\rangle d\sigma(z)}
\end{equation}
The following proposition shows how to estimate $T_\epsilon(x)$ in the
limit as $\eps\to0$:
\begin{proposition}
  \label{thm:mfpt}
  Let $(\hat \phi_{y}(s),\hat \theta_{y}(s))$ solve the instanton
  equations in~\eqref{eq:4inf} with $y$ as specified in
  \eqref{eq:ydefmfpt},  and let $\hat F$ and $\hat V$ be given
  by~\eqref{eq:33inf}.Denote by $\hat n = \nabla f(y)/|\nabla f(y)|$
  the inward pointing unit normal vector on $\partial A$ at $y$.  Then
  for any $x\in A^c$, the
  mean first passage time $T_\eps(x)$ satisfies
  \begin{equation}
    \lim_{\eps\to0} \frac{T_\eps(x)}{\bar T_\eps} = 1
  \end{equation}
  where
  \begin{equation}
    \bar T_\eps= (2\pi\eps)^{1/2} |\langle \hat n, b(y)\rangle|^{-1} |\hat\theta_y(1)|^{-1} \left(\frac{\langle
        \hat n, \hat F\hat n\rangle}{\det \hat F}\right)^{-1/2}\hat V ^{-1} \exp\left(\frac1{2\eps}\int_0^1 \lambda^{-1}(s)\langle \hat\theta_{y}(s), a \hat
      \theta_{y}(s)\rangle\,ds
    \right)\,.
  \end{equation}
\end{proposition}

\begin{proof}
  From~\eqref{eq:Teqn} it follows with the use of
  proposition~\ref{thm:probflux} that
  \begin{equation}
    \bar T_\eps= -\left(\langle \hat n(y), b(y)\rangle \bar P_\eps^{\partial B}\right)^{-1}\,,
  \end{equation}
  where we additionally used that $\lim_{\eps\to0} \int_{A^c} \rho_\eps(z)\,dz = 1$.
\end{proof}

\begin{remark}
  Another interesting case is when we demand $\langle \hat n(z),
  b(z)\rangle = 0$ everywhere on $\partial A$, such that $A^c$
  corresponds exactly to the basin of attraction of the process. In
  this case, the situation is more complicated. Generally, one expects
  a different scaling of the form
  \begin{equation}
    T_\eps(x) \sim C \eps^{-1/2} \left(\bar P_\eps^{\partial A}\right)^{-1}\,.
  \end{equation}
  The underlying assumptions need to make sure that the quasipotential
  is twice differentiable at the exit point, which is true for
  gradient systems, but not generally true for an arbitrary drift. We
  refer to~\cite{maier-stein:1997, bouchet-reygner:2016} for details.
\end{remark}

\subsubsection{Example: Ornstein-Uhlenbeck process}
\label{sec:exampl-ornst-uhlenb-mfpt}

For the 1D Ornstein-Uhlenbeck process,
\begin{equation}
  \label{eq:1D-OU-mfpt}
  dX_t^\epsilon = -\gamma X_t^\epsilon \,dt + \sqrt{\epsilon} \,dW_t,\qquad X_0 = 0\,,
\end{equation}
$X_t\in\RR$, the formula for the MFPT is known
exactly~\cite{thomas:1975}. Concretely, the expected time for the
process~(\ref{eq:1D-OU-mfpt}) to leave the set $A^c=[-\infty, z]$ is
given by
\begin{equation}
  T_\eps = \frac1\gamma \sqrt{\frac\pi2} \int_0^{z\sqrt{2\gamma/\epsilon}} \left(1+\text{erf}\left(\frac t{\sqrt{2}}\right)\right) \exp\left(\frac{t^2}2\right)\,dt\,.  
\end{equation}
In the limit $\eps\to0$, truncating the prefactor at $\mathcal
O(\eps^{3/2})$, this yields the limiting result
\begin{equation}
  \label{eq:tau-approx}
  \tilde T_\epsilon = \left(\frac1z\sqrt{\frac{\pi\epsilon}{\gamma^3}} + O(\epsilon^{3/2})\right) e^{\epsilon^{-1}\gamma z^2}\,.
\end{equation}

Since necessarily $y=z$, and there is no perpendicular direction in
1D, Proposition~\ref{thm:mfpt} tells us that $\bar T_\eps = C
\exp(\eps^{-1}\gamma z^2)$ with $C$ given by
\begin{equation}
  C = (2\pi\eps)^{-1} \left( |b(y)| |\hat \theta_y(1)| \hat V|\right)^{-1}, 
\end{equation}
Using
\begin{align*}
  \hat\theta_y(1) &= 2\gamma z, & |b(y)| &= \gamma z\\
  \hat Q^* &= (2\gamma)^{-1},& \hat Q_y(t) &= (2\gamma)^{-1}\\
  \hat V &= \left(\hat \theta_y(1)^2 \hat
                       Q_y(1)\right)^{-1/2}
                       = (2\gamma z^2)^{-1/2}
\end{align*}
we obtain
\begin{equation*}
  C = \frac1z \sqrt{\frac{\pi\eps}{\gamma^3}}\,,
\end{equation*}
in agreement with the analytical result in (\ref{eq:tau-approx}).

\subsubsection{Example: Exit from a displaced circle}
\label{sec:example:-exit-from-circle}

Consider instead the situation of $X_t\in\RR^2$, but still
\begin{equation}
  \label{eq:2D-OU-mfpt}
  dX_t^\epsilon = -\gamma X_t^\epsilon \,dt + \sqrt{\epsilon} \,dW_t,\qquad X_0 = 0\,.
\end{equation}
We are interested in the expected time to enter the complement of the translated circle
\cite{grasman-herwaarden:1999} of radius $r$ around $(z-r,0)$, i.e. exit
\begin{equation}
  \label{eq:B-circ}
  A^c_r = \{ x\in\RR^2 \ |\  |x-(z-r,0)|\le r\}\,.
\end{equation}
We want to compare this result against the translated half-plane,
\begin{equation}
  \label{eq:B-plane}
  A^c_| = \{ x\in\RR^2 \ |\ x_1 \le z\}\,,
\end{equation}
where $x_1$ is the first component of $x=(x_1,x_2)$. The most likely
exit point, located at $y=(z,0)$, is identical in both cases, as is
the instanton, but we expect to find different prefactors due to the
difference in curvature at $y$ between the two
sets. In the left panel of Fig.~\ref{fig:curvature} we illustrate the problem
setup. We will focus on the case $r>z$ when the instanton lies on the
$x$-axis: at $r=z$ the instanton becomes degenerate and we obtain
\textit{caustics}, i.e.~due to rotational symmetry every point on the
circle is equally likely to be the exit point.

When $r>z$, we have, basically identically to
section~\ref{sec:exampl-ornst-uhlenb-mfpt},
\begin{align*}
  \hat \theta_y(1) &= 2\gamma z,& |\langle \hat n, b(y)\rangle |  &= \gamma z\\
  \hat Q^* &= (2\gamma)^{-1}\mathrm{Id},& \hat Q_y(t) &= (2\gamma)^{-1}\mathrm{Id}\\
  \hat V &= (2\gamma z^2)^{-1/2}\,.
\end{align*}
The only difference between the two cases~(\ref{eq:B-circ})
and~(\ref{eq:B-plane}) are the curvature contributions $\hat F_r$ and
$\hat F_|$, respectively. For the half-space $A^c_|$, the curvature at
$y$ is 0, and thus $\hat F_|=\mathrm{Id}$. For the circle, instead, we can
choose
\begin{equation*}
  f(x) = |x-(z-r,0)|^2 - r^2
\end{equation*}
so that $A^c_r$ is the zero level-set of $f$. Then,
\begin{equation*}
  |\nabla f(y)| = 2r
\end{equation*}
and
\begin{equation*}
  \nabla\nabla f(y) = 2\mathrm{Id}\,,
\end{equation*}
and thus
\begin{align*}
  \hat F_r &= \mathrm{Id} - \hat Q_y(1) |\hat\theta_y(1)||\nabla f(y)|^{-1} \nabla\nabla f(y)\\
  &=(1-\frac{z}{r}) \mathrm{Id}\,.
\end{align*}
Defining the scalar contribution of the curvature to the prefactor
$c=(\det_\perp \hat F)^{-1/2}$ as $c_|$ for the planar case and $c_r$
for the case of a circle with radius $r$, respectively, we obtain
\begin{equation}
  c_| = 1,\qquad\text{and}\qquad c_r =  \left(\det_\perp (1-\frac zr)\mathrm{Id}\right)^{1/2} = \sqrt{\frac{r-z}{r}}
\end{equation}

\begin{figure}
  \begin{center}
    \includegraphics[height=160pt]{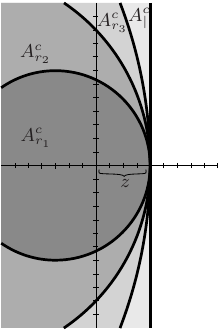}
    \hspace{2em}
    \includegraphics[height=160pt]{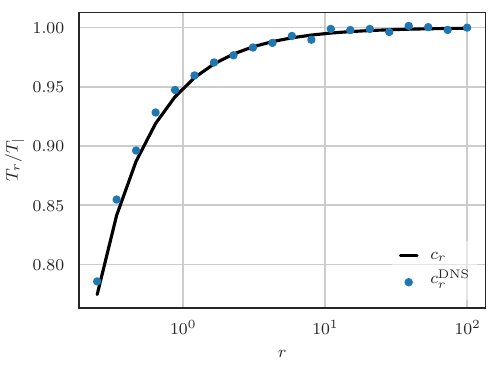}
  \end{center}
  \caption{Left: Schematic representation of the curvature
    experiment. An Ornstein-Uhlenbeck process is started at the origin
    (center), and we are interested in the time $T_r$ at which it
    leaves the circular set $A^c_r$, or time $T_|$ at which it leaves the
    half-plane $A^c_|$. The most likely exit point is always $(z,0)$,
    but the curvature of the set boundary differs. Right: Ratios of
    $T_r^{\text{DNS}}$ to $T_|^{\text{DNS}}$ for different $r$, in
    comparison to the theoretical
    prediction~(\ref{eq:curvature-factor}). For $r\to\infty$, this
    ratio should converge to $1$, but for small radii the curvature
    has a measurable effect. For example, at a radius $r=\frac14$, the
    measured exit time is roughly 25\% smaller than for the
    planar exit.\label{fig:curvature}}
\end{figure}

In order to measure this curvature prefactor experimentally, we
performed the following numerical experiment: For different radii $r$
from $r=0.25$ to $r=100$, we performed $N=2\cdot 10^5$ Monte-Carlo
simulations each, and measured the mean time $T_r^{\text{DNS}}$ to
exit the set $A^c_r$. We performed further $2\cdot 10^5$ Monte-Carlo
simulations in the planar case $A^c_|$ to obtain $T_|^{\text{DNS}}$. The
ratio of the measured passage times of the planar case to the circular
case yields a numerical estimate of the curvature component only:
\begin{equation*}
  c_r^{\text{DNS}} = \frac{T_r^{\text{DNS}}}{T_|^{\text{DNS}}}
\end{equation*}
which can be compared against the analytical prediction
\begin{equation}
  \label{eq:curvature-factor}
  c_r =\sqrt{(r-z)/r}\,.
\end{equation}
This comparison is shown in Fig.~\ref{fig:curvature}. For
$r\to\infty$, $c_r$ indeed converges to $1$, but for small radii the
curvature has a measurable effect, and the measured discrepancy to the
planar prediction agrees with the correction term given
in~(\ref{eq:curvature-factor}). The other parameters are $\gamma=1$,
$z=0.1$, and $\eps=5\cdot10^{-3}$.

\section{Infinite-dimensional examples}
\label{sec:infinited}

In the previous sections, we derived prefactor estimates and sharp
limits in finite dimension. All of these estimates have a counterpart
in infinite dimension, i.e.~when applied to stochastic partial
differential equations.

For concreteness we will focus on reaction-advection-diffusion
equations of the type (for more general equations see
Appendix~\ref{sec:general})
\begin{equation}
  \label{eq:SPDE-generic}
  \partial_t c + v(x) \cdot
  \nabla c = \nabla \cdot \left( D(x) \nabla c\right) + r(x) c  + f(c)
  + \sqrt{\eps}\eta\,,\quad c(0)=c_0\,,
\end{equation}
Here $t\in [0,\infty)$, $x\in \Omega \subset \RR^d$, with $\Omega$
compact, and $c: [0,\infty)\times \Omega \to \RR$; the vector field
$v: \Omega \to \RR^d$ is in $C^2(\Omega)$; the diffusion tensor
$D :\Omega \to \RR^d \times \RR^d$ is in $C^2(\Omega )$, symmetric,
$D^T(x) = D(x)$, and positive-definite for all $x\in\Omega$;
$r: \Omega \to \RR$ is in $C^2(\Omega)$; the reaction term
$f:\RR \to \RR$, is in $C^2(\Omega)$, nonlinear in general; and the
noise $\eta$ is white-in-time Gaussian with covariance
\begin{equation}
  \label{eq:SPDE-noise}
  \EE \eta(t,x) \eta(t',x') = \delta(t-t') \chi(x,x')\,.
\end{equation}
where $\chi : \Omega\times \Omega \to \RR$ given and in
$C^2(\Omega\times \Omega)$. If $\Omega$ is a rectangular domain in
$\RR^d$, we can impose periodic boundary condition on $c$, assuming
that all the other functions in~\eqref{eq:SPDE-generic} are also
periodic; otherwise, denoting by $\hat n(x)$ the unit normal to
$\partial \Omega$, we impose that $r(x) =0$ and
$\hat n(x) \cdot v(x) = 0$ on $\partial \Omega$, $\chi(x,y)=0$ for all
$x\in \partial\Omega$ and $y\in\Omega$ or $y\in \partial\Omega$ and
$x\in\Omega$, and also that
\begin{equation}
  \label{eq:40}
  \hat n(x) \cdot D(x) \nabla c(t,x) = 0 \qquad \text{for all\ \ $x\in \partial
    \Omega$ \ and \ $t\ge 0$}\,.
\end{equation}

In this infinite-dimensional setting, the determinants must be
replaced by functional determinants, the instanton equations become
PDEs, and the Riccati equation must be replaced by a functional
variant as well. While these changes require a whole new set of
techniques to rigorously prove validity, our results and methods remain
formally applicable in infinite dimension. For example, considering
probabilities for a linear observable, i.e.
\begin{equation}
  \label{eq:probspde}
  P^{c_0}_\eps(T,z) = \mathbb P^{c_0} \Big( \int_\Omega \phi(x) c(T,x)\,dx\ge z\Big)\,,
\end{equation}
for some test function $\phi: \Omega \to \RR$, we obtain a proposition
analogous to Proposition~\ref{thm:probability}:
\begin{proposition}[Probabilities for SPDEs]
  \label{thm:SPDEs}
  Let the fields $c(t,x)$, $\theta(t,x)$ solve the instanton
  equations
  \begin{equation}
    \label{eq:instanton:SPDE}
    \left\{\begin{aligned}
        \partial_tc &= \nabla \cdot \left( D(x) \nabla c\right) - v(x) \cdot
        \nabla c + r(x) c  + f(c)  + \int_\Omega
        \chi(x,y)\theta(t,y)\,dy\,,\quad
        &&c(0) = c_0\,,\\
        \partial_t \theta &= -\nabla \cdot \left( D(x) \nabla \theta
        \right) - \nabla \cdot\left(v(x) \theta\right) - r(x) \theta  -
          f'(c)\theta,&&
          \theta(T) = \phi\,,
    \end{aligned}\right.
  \end{equation}
  with either periodic boundary conditions, or
  \begin{equation}
    \label{eq:42}
    \hat n(x) \cdot D(x) \nabla c(t,x) = \hat n(x) \cdot  D(x) \nabla
    \theta(t,x) =0  \qquad \text{for all \ \ $x\in \partial
      \Omega$ \  and  \ $t\in[0,T]$};
  \end{equation}
  Let also the field $\mathcal Q(t,x,y) = \mathcal Q(t,y,x)$ solve
  \begin{equation}
    \label{eq:Q-SPDE}
    \begin{aligned}
      \partial_t \mathcal Q =& \nabla_x \cdot \left( D(x) \nabla_x
        \mathcal Q \right) + \nabla_y \cdot \left( D(y) \nabla_y
        \mathcal Q \right) - v(x) \cdot \nabla_ x\mathcal Q - v(y)
      \cdot \nabla_ y\mathcal Q  + r(x) \mathcal Q + r(y) \mathcal Q \\
      & + f'(c(t,x)) \mathcal Q + f'(c(t,y)) \mathcal Q + \int_\Omega
      \mathcal Q(t,x,z) \mathcal Q(t,z,y) f''(c(t,z)) \theta(t,z) \,dz+
      \chi(x,y)\,,
    \end{aligned}
  \end{equation}
  with initial condition $\mathcal Q(0) = 0$, and either periodic
  boundary conditions in $x$ and $y$, or
  \begin{equation}
  \label{eq:42Q}
  \begin{aligned}
    &\hat n(x) \cdot D(x) \nabla_x \mathcal Q(t,x,y)= 0 \qquad
    &&\text{for all \  $x\in \partial
      \Omega$, \ $y\in \Omega$, \  and \ $t\in[0,T]$\,,}\\
    &\hat n(y) \cdot D(y) \nabla_y \mathcal Q(t,x,y)= 0 \qquad
    &&\text{for all \ $y\in \partial \Omega$, \ $x\in \Omega$, \ and \
      $t\in[0,T]$\,.}
  \end{aligned}
  \end{equation}
  Then the probability $P^{c_0}_\eps(T,z)$ in~\eqref{eq:probspde} satisfies
  \begin{equation}
    \lim_{\eps\to0} \frac{P^{c_0}_\eps(T,z)}{\bar P^{c_0}_\eps(T,z)} = 1\,,
  \end{equation}
  where
  \begin{equation}
    \bar P^{c_0}_\eps(T,z) = (2\pi)^{-1/2} \eps^{1/2} \mathcal V(T,c_0)
    \exp\left(-\frac1{2\eps} \int_0^T \int_{\Omega^2} \theta(t,x)
      \chi(x,y) \theta(t,y)\,dx\,dy\,dt\right)\,,
  \end{equation}
  with
  \begin{equation}
    \label{eq:VPDEdef}
    \begin{aligned}
      \mathcal V(T,c_0) &= \left(\int_{\Omega^2} \theta(T,x) \mathcal
        Q(T,x,y) \theta(T,y)\,dx\,dy\right)^{-1/2}\\
      & \times \exp\left(\frac12
        \int_0^T \int_{\Omega} f''(c(t,x)) \theta(t,x)\mathcal
        Q(t,x,x)\,dx\,dt\right)\,.
    \end{aligned}
  \end{equation}
\end{proposition}

We will omit the proof of this proposition since it essentially
amounts to translating the equations in
Proposition~\ref{thm:probability} to the infinite dimensional setting,
and using the fact that the term involving the tensor $F(T,x)$
disappears since the equivalent of $\nabla \nabla f$ is zero for a
linear observable as in~\eqref{eq:probspde}. Similar reformulations of
our other propositions are straightforward as well, and for the sake of
brevity we will therefore omit writing them down.

Note the instanton equation for $c(t,x)$ in~\eqref{eq:instanton:SPDE}
as well as the Riccati equation in~\eqref{eq:Q-SPDE} for $\mathcal
Q(t,x,y)$ are well-posed as formulated, i.e. forward in time;
similarly the instanton equation for $\theta(t,x)$
in~\eqref{eq:instanton:SPDE} is well-posed as formulated,
i.e. backward in time.

In Secs.~\ref{sec:line-advect-react}
and~\ref{sec:nonlin-advect-diffu-react}, we confirm numerically that
Proposition~\ref{thm:SPDEs}, or its counterparts for other quantities,
produces results that agree with those obtained via direct sampling of
the SPDE in~\eqref{eq:SPDE-generic}. As example problems, we consider
two special cases of~\eqref{eq:SPDE-generic}: First, in
Sec.~\ref{sec:line-advect-react} we study a linear
advection-reaction-diffusion equation with spatially non-homogeneous
forcing, for which some analytical results can be derived. The
probability that the concentration exceeds a threshold at a given
location can then be estimated by a mixed analytical-numerical
approach. Second, in Sec.~\ref{sec:nonlin-advect-diffu-react}, we
study a reaction-advection-diffusion equation with cubic nonlinearity,
where we need to resort to fully solving numerically all involved
instanton and Riccati equations.

\subsection{Linear reaction-advection-diffusion equation with
  non-local forcing}
\label{sec:line-advect-react}

Here we consider the stochastic one-dimensional
advection-diffusion-reaction equation,
\begin{equation}
  \label{eq:adv-diff}
  \partial_t c = \kappa \partial_x^2 c  - \partial_x (v(x) c) - \alpha c + \sqrt{\eps}\eta\,.
\end{equation}
where $x\in[0,1]$ periodic. This a special case
of~\eqref{eq:SPDE-generic} with $D = \kappa$, $r(x) = -\partial_x v$,
and $f(c) =-\alpha c$. For the covariance of the white-in-time forcing
$\eta(t,x)$ we take
\begin{equation}
  \label{eq:cov:loc}
  \chi(x,x') = \psi_{x_1}^\delta(x) \psi_{x_1}^\delta(x')\,.
\end{equation}
where $\psi_{x_1}^\delta(x)$ is a mollifier of length $\delta<1$
concentrated around $x_1\in[0,1]$. For the observable, we take 
\begin{equation}
  \PP \left(\int_0^1 c(x) \psi_{x_2}^\delta(x)\,dx \ge z\right)\,,
\end{equation}
where $x_2\not =x_1$ and the expectation is taken on the invariant
measure of the solution to~\eqref{eq:adv-diff}. Intuitively, the
scenario we are investigating is therefore that of a pollutant, the
density of which is described by $c(t,x)$, along a one-dimensional periodic
channel. The pollutant is randomly emitted into the environment at a
spatial location $x_1$, and gets advected and diffused conservatively,
but decays over time with rate $\alpha$. We are interested in
measuring extreme concentrations of the pollutant around the location
$x_2$ somewhere else in the channel.

\subsubsection{Finite-dimensional analogous case}
\label{sec:finite-d-analog}

Equation~(\ref{eq:adv-diff}) is a linear SPDE, an infinite-dimensional
generalization of the (non-normal, $n$-dimensional) Ornstein-Uhlenbeck
process,
\begin{equation}
  \label{eq:OUadr}
  dX^\eps_t = -\Gamma X^\eps_t\,dt + \sqrt{2\eps}\sigma\,dW_t\,,\quad t\ge0\,,
  \quad X^\eps_t\in\RR^n\,,
\end{equation}
for $\Gamma\in \RR^{n\times n}$, where $\Gamma\ne\Gamma^\top$
(non-symmetric), $\Gamma\Gamma^\top\ne\Gamma^\top\Gamma$ (non-normal),
$W$ is an $n$-dimensional Wiener process, $\sigma$ not necessarily
invertible, and $a = \sigma \sigma^\top$. In analogy to the above
scenario, we can ask for probabilities on the invariant measure of the form
\begin{equation}
  P_\eps^{A_z}=\PP(X^\eps\in  A_z)\,,\quad\text{where}\quad
  A_z = \{ x\in\RR^n | \langle k, x\rangle \ge z \}\,.
\end{equation}
Intuitively, we want to estimate the probability that the process
reaches the far-side of a plane with normal $k$, distance $z$ away
from the origin.

If we define the symmetric, positive semi-definite matrix $C$ as the solution of the
Lyapunov equation
\begin{equation}
  \label{eq:lyapunov}
  \Gamma C + C \Gamma^\top = 2 a\,,
\end{equation}
the quasi-potential of~(\ref{eq:OUadr}) is given by
\begin{equation}
  V(x) = \langle x, C^{-1} x\rangle\,.
\end{equation}
Since $C$ is not necessarily invertible, we interpret $w = C^{-1}v$ to
be the solution of $v = Cw$ if it exists, and otherwise the
quasi-potential is set to infinity. The final point of the geometric instanton
$\hat \phi(s)$ for observable value $z$ must be given by
\begin{equation}
  \label{eq:argmin}
  y = \argmin_{x\in\partial A_z}( V(x))
  = \frac{z}{\langle k, Ck\rangle} C k
\end{equation}
Here, we must assume that $k$ is not in the kernel of $C$, which is
the same as saying that $k$ is in the support of the invariant measure
of~(\ref{eq:OUadr}). The action is
\begin{equation}
  \label{eq:nonnOU-I}
  I(z) = \frac{z^2}{\langle k, Ck\rangle}\,,
\end{equation}
which follows from evaluating $V(x)$ at $x=y$ given by~\eqref{eq:argmin}.

To estimate the prefactor, we need to solve the $Q$-equation with
appropriate boundary condition. Because the equation~\eqref{eq:OUadr}
is linear, $\hat Q(s) = \tfrac12 C$. As a result
\begin{equation}
  \label{eq:nonnOU-P}
  \lim_{\eps\to 0} \frac{P_\eps^{A_z}}{\bar P_\eps^{A_z}}  =1\qquad
  \text{with} \qquad \bar P_\eps^{A_z}= \sqrt{\frac{\eps \langle k, Ck\rangle}{4\pi
      z^2}}
  \exp\left(-\frac{z^2}{\eps} \langle k, C k\rangle^{-1}\right)
\end{equation}
where we used 
\begin{equation}
  \label{eq:nonnOU-hatV}
  \hat V = \langle \hat \theta_y(1), \hat Q_y(1) \hat
  \theta_y(1)\rangle^{-1/2}
  = \frac1{\sqrt{2}z} \langle k, Ck\rangle^{1/2} 
\end{equation}
since
\begin{equation}
  \label{eq:nonnOU-theta}
  \hat\theta(s) = 2C^{-1}\hat\phi(s), \quad\text{so that}\quad \hat\theta(1) = 2C^{-1}y = \frac{2z}{\langle k, Ck\rangle} k.
\end{equation}
Due to the linearity of the system, only the end location of the
instanton at $s=1$ plays a role.

\subsubsection{Infinite dimensional setting}
\label{sec:advdiff1}

Coming back to the infinite-dimensional case, we can proceed
similarly, noticing that~(\ref{eq:adv-diff}) can be written as
\begin{equation}
  \partial_t c = -\mathcal G c + \eta\,,
\end{equation}
where
\begin{equation}
  \mathcal G = - \kappa \partial_x^2 + v(x)\partial_x +  (\partial_x v) + \alpha
\end{equation}
is a linear differential operator, acting on functions in
$L^2$. Notably, $\mathcal G$ is not normal, and we need to solve the
 Lyapunov equation
\begin{equation}
  \label{eq:op-lyapunov}
  \mathcal G C + C \mathcal G^\top = 2 \psi^\delta_{x_1}(x) \psi^\delta_{x_1}(y)\,,
\end{equation}
where $C(x,y)$ is symmetric and positive semi-definite in $L^2$ and
$C\mathcal G^\top = (\mathcal G C)^\top$, i.e. it is the differential
operator acting on the second variable of
$C(x,y)$. Explicitly~(\ref{eq:op-lyapunov}) reads
\begin{equation}
  \label{eq:op-lyapunov-advdiffreac}
  -\kappa (\partial_x^2 + \partial_y^2) C  -  (v\partial_x +
  v\partial_y) C
  - (\partial_x v + \partial_y v) C 
  + 2\alpha C
  = 2 \psi^\delta_{x_1}(x)\psi^\delta_{x_1}(y)\,.
\end{equation}

By extending the argument for the finite dimensional case to the
functional setting, we also deduce that the endpoint at $s=1$ of the
geometric instanton $\hat \phi(s,x)$ for hitting the set $\mathcal A_z$ is given by
\begin{equation}
  \hat \phi(1,x) = \argmin_{\xi \in\partial \mathcal A_z} V(\xi)
  = \frac{z \int_0^1 C(x,y) \psi_{x_2}^\delta(y)\,dy}{\int_0^1\int_0^1\psi_{x_2}^\delta(x) C(x,y) \psi_{x_2}^\delta(y)\,d\!x\,d\!y}\,,
\end{equation}
in analogy to equation~(\ref{eq:argmin}) by replacing the inner
product with the $L^2$ inner product. The corresponding action is
(compare~(\ref{eq:nonnOU-I}))
\begin{equation}
  I(z) = z^2\left(\int_0^1\int_0^1\psi_{x_2}^\delta(x) C(x,y)
    \psi_{x_2}^\delta(y)\,d\!x\,d\!y\right)^{-1}\,.
\end{equation}

Concerning the prefactor, we have (compare (\ref{eq:nonnOU-theta}))
\begin{equation}
  \theta(t=0,x) = 2z \left(\int_0^1\int_0^1\psi_{x_2}^\delta(x)
    C(x,y) \psi_{x_2}^\delta(y)\,d\!x\,d\!y\right)^{-1} \psi_{x_2}^\delta(x)
\end{equation}
and thus (compare~(\ref{eq:nonnOU-hatV}))
\begin{equation}
  \hat {\mathcal V} =
  \frac{1}{2z}\left(\int_0^1\int_0^1\psi_{x_2}^\delta(x)
    C(x,y) \psi_{x_2}^\delta(y)\,d\!x\,d\!y\right)^{1/2}\,,
\end{equation}
so that in total,
\begin{equation}
  \label{eq:prob}
   \lim_{\eps\to0} \frac{\PP \left(\int_0^1 c(x)
       \psi_{x_2}^\delta(x)\,dx \ge z\right)}{P_\eps(z)} = 1
 \end{equation}
 with
\begin{equation}
  P_\eps(z)= \sqrt{\frac{\eps}{4\pi z^2}}
  \left(\int_0^1\int_0^1\psi_{x_2}^\delta(x) C(x,y)
    \psi_{x_2}^\delta(y)\,
    d\!x\,d\!y\right)^{1/2} \exp\left(-\eps^{-1}I(z)\right)
\end{equation}
as in equation~(\ref{eq:nonnOU-P}). Clearly, this probability is
Gaussian, as expected when we consider a linear system with Gaussian
input. 

\begin{figure}
  \begin{center}
    \includegraphics[width=0.48\textwidth]{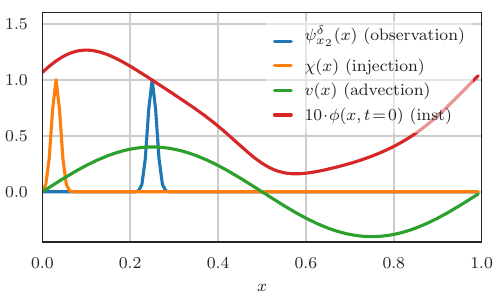}
    \includegraphics[width=0.48\textwidth]{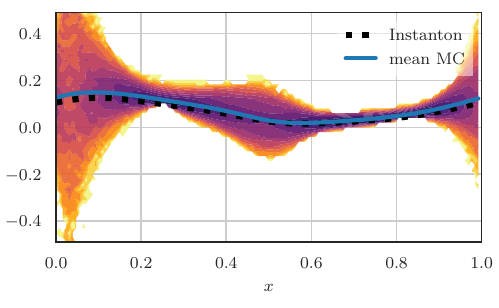}
  \end{center}
  \caption{\label{fig:inst} \textit{Left:} Instanton configuration
    (rescaled, red) for the advection-reaction-diffusion
    equation~(\ref{eq:adv-diff}) with velocity field $v(x)$ (green),
    conditioning on observing a high concentration in the region
    specified by $\psi_{x_2}^\delta$ (blue), where concentration is
    released randomly around location $x_1$ (yellow). Here,
    $\kappa=10^{-2}, \alpha=1$, $v(x) = 0.4\sin(2\pi x)$,
    $\eps = 5\cdot 10^{-2}$, $z=0.0025$, $x_1=1/32$, $x_2=1/4$, and
    $\delta=10^{-2}$, as well as $N_x=128$. \textit{Right:} Comparison
    between instanton and the result from direct numerical
    simulations, conditioning on hitting the set $\mathcal
    A_z$. Notably, the mean realization recovers the instanton, but
    variances are very high around $x=1/32$, where the concentration
    fluctuations are inserted, and also high around $x=1/2$, where the
    velocity field compresses the fluctuations of the concentration
    field.}
\end{figure}
\begin{figure}
  \begin{center}
    \includegraphics[width=0.48\textwidth]{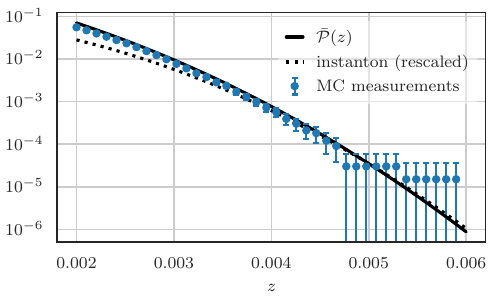}
  \end{center}
  \caption{\label{fig:prob} Probabilities for the
    advection-diffusion-reaction problem with localized forcing
    obtained by direct numerical simulations (blue), from
    formula~(\ref{eq:prob}) (black). Also shown is the instanton
    prediction rescaled by a factor $10^{-2}$ (which cannot be
    obtained \textit{a~priori} and was adjusted to get the best
    fit). The difference between the instanton prediction
    and $\bar{\mathcal P}(z)$ is noticeable.}
\end{figure}
In order to obtain $C$, we numerically solve the operator Lyapunov
equation~(\ref{eq:op-lyapunov}) for a finite difference representation
of $\mathcal G$. As an example, we take the advection-diffusion-reaction
equation~(\ref{eq:adv-diff}), with a velocity field
\begin{equation}
  v(x) = A \sin(2\pi x) \qquad (A>0)\,,
\end{equation}
i.e.~with negative divergence at the center of the domain, $x=1/2$. We
expect (if we were not to condition on any outcome, and would force
homogeneously in space) that in a typical configuration the
concentration has higher variance at $x=1/2$ and lower around $x=0$,
where it is depleted by the velocity.  We also choose $x_1=1/32$ and
$x_2=1/4$, , i.e.~the forcing is localized on the very left of the
domain, where concentration is released randomly, and we are sensing
concentration further down the channel, where it is transported by
advection and diffusion. We also check the results above numerically
by comparing the instanton prediction derived here with the result of
a direct simulation of the advection reaction diffusion
equation~\eqref{eq:adv-diff} with stochastic forcing localized
according to~\eqref{eq:cov:loc}, and counting the number of times the
observable exceeds a given threshold.

In the left panel of Fig.~\ref{fig:inst} we show the corresponding
instanton: It has a non-trivial dependence on both the location of the
forcing as well as the location of the observation. Additionally shown
are the localized functions defining the forcing and the observation,
as well as the velocity field. In the right panel of
Fig.~\ref{fig:inst} we shows a comparison against direct simulation
results of the SPDE. In particular we note that the mean of all
observed events within $\mathcal A_z$ resembles the
instanton. Further, the highest variance in the direct simulation is
clearly observed where noise is injected, as well as close to $x=1/2$,
at the sink of the velocity field.

Depending on $z$ and $\eps$, the probability can be computed
via~(\ref{eq:prob}). In Fig.~\ref{fig:prob} we show a comparison of these
probabilities against numerics, obtained by integrating the
SPDE~(\ref{eq:adv-diff}) for a long time, and observing how often each
threshold $z$ was exceeded. We note that the prediction indeed
captures not only the scaling, but also the correct prefactor.

\subsection{Nonlinear reaction-advection-diffusion equation}
\label{sec:nonlin-advect-diffu-react}

For the linear reaction-advection-diff\-u\-sion equation of
section~\ref{sec:line-advect-react}, we were able to harness linearity
to avoid explicitly using Proposition~\ref{thm:SPDEs} and numerically
solving the instanton and Riccati equations. This is no longer
possible if we choose the system to be nonlinear, for example by
taking a cubic reaction term. In this case, we no longer have access
analytically to the invariant measure, the instanton trajectory, or
the solution to the prefactor terms.

Concretely, we consider the nonlinear stochastic partial differential
equation
\begin{equation}
  \label{eq:SPDE}
  \partial_t c = \kappa
  \partial_x^2 c - v(x) \partial_x c -\alpha c - \gamma c^3 + \sqrt{2\epsilon} \eta\,,
\end{equation}
where the spatial variable $x$ is periodic on the domain $[-L/2,L/2]$
and $t\in [-T,0]$. This equation is a special case
of~\eqref{eq:SPDE-generic} with $D(x) = \kappa$, $r(x) = 0$, and
$f(c) = -\alpha c -\gamma c^3$: the parameters $\alpha$ and $\gamma$ control the
linear and nonlinear part of the reaction term, respectively. As
velocity field we pick
\begin{equation}
  \label{eq:SPDE-velocity}
  v(x) = 4 + 2\sin(4\pi x/L)\,,
\end{equation}
which always advects to the right, but with spatially varying
speed. As a consequence, the whole equation is no longer translation
invariant. We will assume that the noise $\eta$ is white in space and time, with
covariance
\begin{equation}
  \label{eq:38}
  \EE (\eta(t,x)\eta(t',x') ) = \delta(t-t') \delta(x-x')\,.
\end{equation}
This is a rougher noise than the one in the
SPDE~\eqref{eq:instanton:SPDE}, which is allowed because
\eqref{eq:SPDE} is well-posed in 1D with this
forcing~\cite{faris-jona-lasinio:1982}.

We are interested in the probability that a sample on the invariant
measure of~\eqref{eq:SPDE} exceeds the threshold $z$ at the location
$x=0$, i.e.
\begin{equation} \label{eq:mollified_exp_pde}
  \PP(c(x=0)\ge z)\,,
\end{equation}
Note that this problem can either be viewed as a nonlinear version of
the reaction-advection-diffusion equation of
section~\ref{sec:line-advect-react}, or as an infinite dimensional
version of the $\RR^2$ process given in
section~\ref{sec:exampl-nonl-irrev-R2-nonlin}.

To apply our method we need to first solve the geometric instanton
equations (\ref{eq:4infpdf}) using the appropriate boundary
conditions. Since we are focusing in this example on the limit $T\to
\infty$, an efficient way of numerically solving these equations using
an iterative scheme and arclength parametrization has been discussed
in previous work~\cite{grafke-grauer-schaefer-etal:2014}. Since for
equation~(\ref{eq:SPDE}) the Hamiltonian is
\begin{equation}
  \label{eq:37a}
  H(c,\theta) = \int_{-L/2}^{L/2} \left( ( \kappa
    \partial_x^2 \hat c - v(x) \partial_x \hat c -\alpha\hat  c
    - \gamma \hat c^3  )\hat \theta(x) + |\hat \theta(x)|^2 \right)dx
\end{equation}
we immediately obtain the (geometric) instanton equations in this case as
\begin{equation}
  \label{eq:inst-nonlinSPDE}
  \begin{cases}
    \lambda \partial_s \hat c = \kappa
    \partial_x^2 \hat c-  v(x) \partial_x \hat c-\alpha \hat c - \gamma
    \hat c^3+2 \hat\theta\,,\qquad & \hat c(0) = 0\\
    \lambda \partial_s \hat\theta = - \kappa \partial_x^2 \hat\theta - \partial_x(v(x)
    \hat\theta) + \alpha \hat\theta + 2\gamma \hat c^2 \hat\theta \,,&
    \hat \theta(1) = \delta(x)\,.
  \end{cases}
\end{equation}
Once the instanton is found, we need to solve the corresponding
(geometric variant of the) forward Riccati equation for $\hat{\mathcal
  Q}(s,x,y)$ given by (\ref{eq:Q-SPDE}).

For the specific SPDE~(\ref{eq:SPDE}), the equation for
$\hat{\mathcal Q}(s,x,y)= \hat{\mathcal Q}(s,y,x)$ is
\begin{equation}
  \label{eq:Q-nonlinSPDE}
  \begin{aligned}
    \lambda \partial_s \hat{\mathcal Q}=& \kappa\partial_x^2 \hat
    {\mathcal Q} + \kappa \partial_y^2 \hat {\mathcal Q} -v(x)
    \partial_x \hat{\mathcal Q} - v(y) \partial_y \hat{\mathcal Q}
    - 2\alpha \hat {\mathcal Q} -3\gamma \left(\hat c^2(s,x)+ \hat c^2(s,y)\right)\hat{\mathcal Q} \\
    & -6\int_{-L/2}^{L/2} \hat{\mathcal Q}(s,x,z) \hat c(s,z)\hat\theta(s,z)
    \hat{\mathcal Q}(s,z,y)\,dz + 2\delta(x\!-\!y)
  \end{aligned}
\end{equation}
to be solved with the initial condition, $\hat {\mathcal Q}(0) = \hat {\mathcal Q}_*$, solution to the Lyapunov equation
\begin{equation}
  \label{eq:spde-lyapunov}
  -\tfrac12\kappa(\partial_x^2 \hat {\mathcal Q}_*
  + \partial_y^2 \hat {\mathcal Q}_*) + \alpha \hat {\mathcal Q}_*
  + \tfrac12 \left(v(x) \partial_x \hat{\mathcal Q}_*
    + v(y) \partial_y \hat{\mathcal Q}_*\right)= \delta(x-y)\,.
\end{equation}

\begin{figure} [htb]
  \begin{center}
    \includegraphics[width=0.48\textwidth]{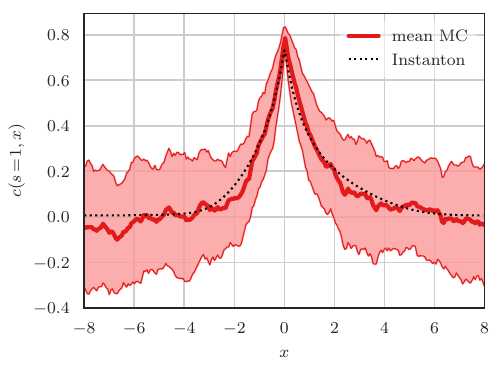}
    \hfill
    \includegraphics[width=0.48\textwidth]{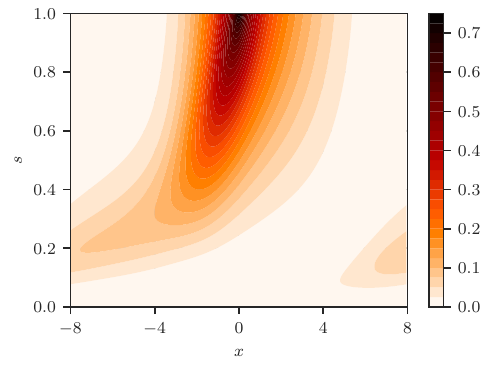}
  \end{center}
  \caption{Left: Comparison of the numerically computed instanton
    (black dots) with the filtered solution from direct numerical
    simulations (red) realizing a high amplitude event $z=0.7$ at the
    origin $x=0$, plus/minus one standard deviation (light red shaded
    region), for the nonlinear stochastic partial differential
    equation~(\ref{eq:SPDE}). The spatially dependent positive
    velocity field $v(x)$ leads to an asymmetry if the typical final
    configuration exceeding $z$ at $x=0$, which is visible both in the
    instanton and the direct simulation results. Right: Instanton
    trajectory along the arclength parameter $s$, showing the temporal
    evolution of the instanton into the large amplitude configuration
    at $s=1$. The positive velocity leads to a traveling wave
    instanton solution that slowly amplifies over time.}
\label{fig:nlpde_filter}
\end{figure}
One can numerically integrate the instanton
equations~(\ref{eq:inst-nonlinSPDE}) to obtain the most likely
configuration that achieves a given amplitude $z$ at $x=0$ at final
time $t=0$. Fourier transforms were used to calculate the spatial
derivatives. When solving the associated Riccati equation for
$\mathcal{Q}$, however, we chose an equally distant point grid to
discretize the time interval $[-T,0]$ for stability reasons. Here, the
time $T$ can be found from the geometric parametrization
\cite{heymann-vanden-eijnden:2008} and the instanton solution needs to
be interpolated onto this new point grid.  For the solution of the
Riccati equation on the equidistant grid, exponential time
differencing~\cite{kassam-trefethen:2005} was employed and the
diffusive term $\kappa(\partial_x^2+\partial_y^2){\mathcal{Q}}$ can be
treated for numerical efficiency using the 2-dimensional fast Fourier
transform.

In the left panel of Fig.~\ref{fig:nlpde_filter} we compare this
numerically computed instanton with the filtered instanton from direct
simulations of the SPDE~\eqref{eq:SPDE} using the method described in
\cite{grafke-grauer-schaefer:2015}. As can be seen, the mean
realization with $z=0.7$ at $x=0$ in the SPDE is very similar to the
instanton. In particular, both instanton and the typical sample from
the SPDE show an asymmetry around $x=0$ coming from the spatially
inhomogeneous velocity field~(\ref{eq:SPDE-velocity}). Also shown is
one standard deviation of the fluctuations around the instanton as
shaded red region. In the right panel of Fig.~\ref{fig:nlpde_filter}
we show the whole evolution of the instanton in arclength parameter
$s$ that compactifies the infinite time interval $t\in[0,\infty)$ into
$s\in[0,1]$. Clearly visible is the (inhomogeneous) movement of the
peak as it is advected with the positive velocity $v(x)$ given in
equation~(\ref{eq:SPDE-velocity}).
\begin{figure}
  \begin{center}
    \includegraphics[height=140pt]{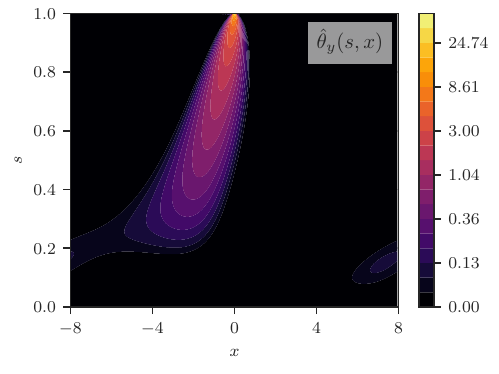}
    \includegraphics[height=140pt]{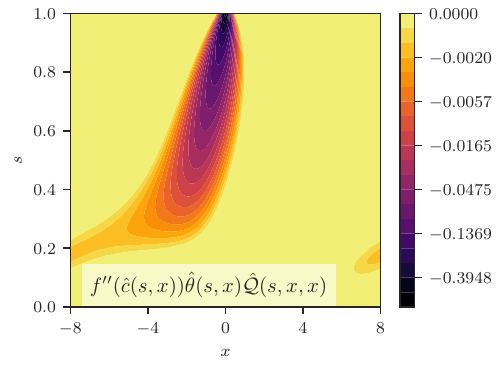}
  \end{center}
  \caption{Left: Conjugate momentum $\hat\theta(s,x)$ along the
    arclength parameter $s$. Right: The quantity
    $f''(\hat c(s,x))\hat \theta(s,x) \hat Q(s,x,x)$. The exponential
    of its integral along $x$ and $s$ enters the expression
    in~\eqref{eq:VPDEdef} for the prefactor component
    $\mathcal V(T,c_0)$. }
\label{fig:nlpde_riccati}
\end{figure}
In the left panel of Fig.~\ref{fig:nlpde_riccati} shows the conjugate
momentum $\hat\theta(s,x)$ along the arclength parameter $s$, and in
the right panel we show the quantity
$f''(\hat c(s,x))\hat \theta(s,x) \hat Q(s,x,x)$: the exponential of
its integral along $x$ and $s$ enters the expression
in~\eqref{eq:VPDEdef} for the prefactor component $\mathcal V(T,c_0)$.

The combination of both results inserted into
Proposition~\ref{thm:SPDEs} allows us to determine the
prefactor. In Fig.~\ref{fig:nlpde_example} we compare the result from
direct simulations of the SPDE~\eqref{eq:SPDE} to the predictions of
the instanton equations together with the prefactor, showing clear
agreement especially for large values of $z$, as expected. A
comparison is further made to the analytical result available for the
linear case, $\gamma=0$, which clearly shows that the nonlinear term
affects the probability. Similarly we compare against the case without
advection ($v(x)=0$), which is gradient, demonstrating that the
advective term also has a considerable (opposite) effect on the
probabilities. In this numerical example, we chose $N_x=256$ grid
points for the discretization in space for a domain of size $L=16$,
and $N_t=5000$ grid points in time for the direct simulations for a
temporal domain of size $T=25$. The instanton is computed with
$N_s=2000$ discretization points in the arclength
parameter. Additional parameters are $\kappa=1$, $\alpha = 0.6$, and
$\gamma=2$. The noise level is set to $\eps=0.1$, and we collected
$N_{\text{samples}}=10^5$ samples in the direct simulations.

\begin{figure}
  \begin{center}
    \includegraphics[width=0.48\textwidth]{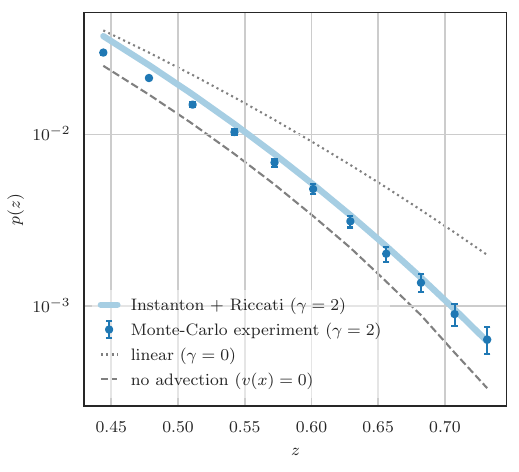}
  \end{center}
  \caption{\label{fig:inst_pde} Comparison of the predicted
    probabilities $p(z)$ of exceeding the threshold $z$. The dark blue
    dots shows the probability $p$ estimated via direct numerical
    simulations for the nonlinear case ($\gamma = 2$), while the light
    blue line is the theoretical prediction from instanton and Riccati
    equation. Clearly, the instantons, together with the corresponding
    prefactors, approximate these probabilities well. For comparison,
    we show the analytical prediction that can be obtained for the
    linear case ($\gamma=0$), highlighting the fact that the
    nonlinearity indeed plays a role for the tail probabilities in
    particular. Similarly, we show the situation without advection
    ($v(x)=0$), demonstrating that the advection term modifies the
    probabilities.}
\label{fig:nlpde_example}
\end{figure}

\section{Generalization to processes driven by a non-Gaussian noise}
\label{sec:nonGaussian}

It is possible to generalize the above results to continuous time
Markov jump processes that cannot be represented by an SDE with
additive Gaussian noise like in~(\ref{eq:sde}). The intuition is
similar: considering a WKB approximation to the BKE allows us to
obtain a Hamilton-Jacobi equation that defines the behavior of the
instanton. If we furthermore keep track of the prefactor in the WKB,
we will obtain an additional equation for it on the next order, which
will yield a generalization of the Riccati equations to obtain sharp
prefactor estimates.

Concretely, consider a continuous-time Markov jump process (MJP) on the
state space $\RR^d$ with the generator
\begin{equation}
  \label{eq:generator}
   L_\eps f(x) = \frac1\eps\sum_{r=1}^R a_r(x)\left( f(x+\nu_r \eps) - f(x)\right)
\end{equation}
which encodes a reaction network for $d\in\NN$ species with $R\in\NN$
reactions, each reaction $r$ leading to a change in species defined by
the vector $\nu_r\in\RR^d$, and happening with rate $a_r(x)\in\RR_+$
possibly depending on the state $x\in\RR^d$. Depending on the
microscopic model at hand, we can expand $a_r(x)$ in terms of orders
of $\eps$,
\begin{equation}
  a_r(x) = a_r^{(0)}(x) + \eps a_r^{(1)}(x) + \mathcal O(\eps^2)\,.
\end{equation}
If we insert the WKB ansatz $f(t,x) = Z(t,x) \exp(\eps^{-1} S(t,x))$
in the BKE
\begin{equation}
  \partial_t f  = L_\eps f\,,
\end{equation}
and collect term of successive orders in $\eps$, we obtain
\begin{align}
  \mathcal O(\eps^{-1}):&& \partial_t S &= -\sum_{r=1}^R a_r^{(0)}(x) (e^{\nu_r \cdot \nabla S}-1) = -H(x,\nabla S(x))\\
  \mathcal O(\eps^{0}):&& \partial_t Z &= -\sum_{r=1}^R a_r^{(0)}(x) e^{\nu_r\cdot \nabla S} \left(\nu_r \cdot \nabla Z + \tfrac12 Z \nu_r \cdot \nabla\nabla S\nu_r + a_r^{(1)}(x)\, Z\right)\,.\label{eq:K}
\end{align}
The instanton $\phi$ solves the Hamilton's equations
\begin{equation}
  \label{eq:nongauss-instanton}
  \begin{cases}
    \dot\phi = \nabla_\theta H(\phi,\theta), & \phi(0)=x\\
    \dot\theta = -\nabla_\phi H(\phi,\theta),
  \end{cases}
\end{equation}
where we additionally get a boundary condition $\phi(T)=y$ or
$\theta(T)=\nabla f(\phi(T))$ depending on the scenario under
consideration. For the generator~(\ref{eq:generator}), the Hamiltonian
is given by
\begin{equation}
  H(\phi,\theta) = \sum_{r=1}^R a_r^{(0)}(\phi)\left(e^{\nu_r\cdot \theta}-1\right)\,,
\end{equation}
but an arbitrary Hamiltonian is possible in general. If we evaluate
$Z(t,x)$ along the instanton, $G(t) = Z(t, \phi(t))$, we have
\begin{equation}
  \dot G(t) = \dot Z + \nabla Z \cdot \dot \phi = \dot Z + \sum_{r=1}^R \nu_r \cdot \nabla Z a_r^{(0)}(\phi) e^{\nu_r \cdot \nabla S}\,,
\end{equation}
and therefore, along the instanton, equation~(\ref{eq:K}) becomes
\begin{equation}
  \dot G = -\tfrac12 G \sum_{r=1}^R e^{\nu_r\cdot \nabla S}  \left(a_r^{(0)}(\phi) \nu_r \cdot \nabla\nabla S \nu_r + a_r^{(1)}(\phi)\right)\,,\quad G(T)=1\,.
\end{equation}
Written in terms of the Hamiltonian and the additional $\mathcal
O(\eps)$ drift term, this becomes
\begin{equation}
  \label{eq:nongauss-G}
  \dot G = -\tfrac12 G \left(\mathrm{tr}\left(H_{\theta\theta} W\right) + A^{(1)}\right)\,,\quad G(T)=1\,,
\end{equation}
with $A^{(1)} = \sum_r a_r^{(1)}(\phi)$. As before, we obtain an
evolution equation for $W=\nabla\nabla S$ by differencing the HJB
equation twice,
\begin{equation}
  \label{eq:nongauss-riccati}
  \dot W = -H_{\phi\phi} - H_{\phi\theta} W - W H_{\phi\theta}^\top - W H_{\theta\theta} W\,,
\end{equation}
to be solved with boundary conditions that depend on the scenario, and
where subscripts of the Hamiltonian denote differentiation. We can
also derive an equation for $Q$,
\begin{equation}
  \label{eq:nongauss-Q}
  \dot Q = Q H_{\phi\phi} Q + Q H_{\phi\theta} + H_{\phi\theta}^\top Q + H_{\theta\theta}\,.
\end{equation}
This allows one to compute sharp estimates for expectations,
probability densities, hitting probabilities and exit times in a
similar way as before, replacing the instanton
equations~(\ref{eq:hamilton_time_parametrization}) and its variants
with (\ref{eq:nongauss-instanton}), and the forward and backward
Riccati equations with~(\ref{eq:nongauss-riccati}) and~(\ref{eq:nongauss-Q}).

\subsubsection{Example: Continuous time Markov jump process}
\label{sec:exampl-cont-time-MJP}

Consider the following continuous-time MJP, inspired
by~\cite{bonaschi-peletier:2016}, in which a particle hops on a grid
with spacing $\eps$ , $X \in \eps\ZZ$ (see
Fig.~\ref{fig:queue-model}), i.e.~the spatial coordinate becomes
continuous for $\eps\to0$. Left and right jumps happen with a rate
\begin{equation}
  \label{eq:queue-rates}
  r_{\pm}^\eps(x) = \exp\left(-\eps^{-1} \left(E(x\pm \eps)-E(x)\right)\right)\,,
\end{equation}
and the generator is given by
\begin{equation}
  \label{eq:example-generator}
  L_\eps f = \eps^{-1} \Big(r_+(x)(f(x+\eps) - f(x)) + r_-(x)(f(x-\eps)-f(x))\Big)\,.
\end{equation}
By construction, this process is in detailed balance with respect to
the Gibbs distribution
\begin{equation}
  \label{eq:inv-meas-example}
  \mu_\infty(x)= Z^{-1} e^{-2\eps^{-1} E(x)}\,, \quad \text{where} \quad
  Z= \sum_{x\in \eps \ZZ}e^{- 2\eps^{-1} E(x)}
\end{equation}
i.e.~$E(x)$ plays the role of the free energy and $\tfrac12\eps$ that of the temperature.

In the continuum limit, $\eps\to0$, the rates~(\ref{eq:queue-rates})
can be expanded as
\begin{equation}
  r_\pm^\eps (x) = e^{\mp E'(x)} \left(1+\tfrac12 \eps E''(x)
    + \mathcal O(\eps^{2})\right) = r_\pm^{(0)}(x) + \eps  r_\pm^{(1)}
  +
  \mathcal O(\eps^2)
\end{equation}
with
\begin{equation}
  r_\pm^{(0)}(x) = e^{\mp  E'(x)}\qquad \text{and}\qquad r_\pm^{(1)}(x) = \frac{1}{2} E''(x) e^{\mp  E'(x)}\,.
\end{equation}
Correspondingly, the LDT Hamiltonian takes the form
\begin{equation}
  H(\phi,\theta) = r_+^{(0)}(\phi) \left(e^\theta-1\right) + r_-^{(0)}(\phi)\left(e^{-\theta}-1\right)\,.
\end{equation}

\begin{figure}
  \begin{center}
    \includegraphics[width=240pt]{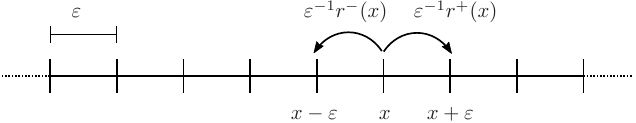}
  \end{center}
  \caption{Schematic depiction of the MJP with generator
    (\ref{eq:example-generator}): A particle at point $x\in\eps\ZZ$
    jumps with rates $r^\pm$ to the left and right, respectively. The
    rates are given in equation~(\ref{eq:queue-rates}). For
    $\eps\to0$, we compute the probability $P_T(z)$ of excursions
    larger than $z\in\RR$ at time $t=T$. }
  \label{fig:queue-model}
\end{figure}

We are interested in estimating the probability to observe a large
excursion $z\in\RR$ at time $t=T$,
\begin{equation}
  P_T(z) = \PP(X_T>z|X_0=0)
\end{equation}
for the MJP with generator~(\ref{eq:generator}), with the particle
starting at $X_0=0$ at $t=0$.
\begin{figure}
  \begin{center}
    \includegraphics[width=240pt]{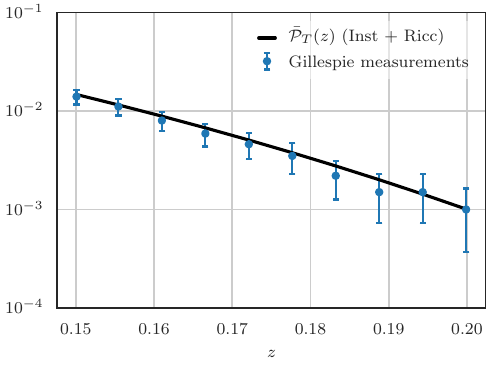}
  \end{center}
  \caption{Comparison between numerical estimate of $\bar P_T(z)$ and
    sampling estimate obtained from the Gillespie
    algorithm~\cite{gillespie:1977} for the original MJP for small
    $\eps$. Here, the parameters are $E(x) = \frac14 x^4$, 
    $T=5.12$, while for the instanton, $N_t = 512$, $\Delta
    t=10^{-2}$, and for the Gillespie algorithm $N=2\cdot 10^3$,
    $N_{\text{samples}} = 10^4$.}
  \label{fig:queue}
\end{figure}
Solving this problem numerically by our approach amounts to
performing the following steps:
\begin{enumerate}[(i)]
\item Solve the instanton equations
  \begin{equation}
    \begin{cases}
      \dot\phi = H_\theta(\phi,\theta) = e^{- E' + \theta} - e^{ E' -
        \theta},\qquad & \phi(0)=0\,,\\
      \dot \theta = -H_\phi(\phi,\theta) = -E'' e^{-E'}\left(e^\theta-1\right) + E'' e^{E'} \left(e^{-\theta}-1\right),& \phi(T)=z\,,
    \end{cases}
  \end{equation}
\item Solve the Riccati equation
  \begin{equation}
    \begin{aligned}
      \dot Q &= H_{\phi\phi}Q^2 + 2H_{\phi\theta} Q + H_{\theta\theta} \\
      &=\left(\left(- E''' + (E'')^2\right) e^{- E'} \left(e^\theta -1 \right)
        + \left(E''' + (E'')^2\right) e^{ E'} \left(e^{-\theta} - 1\right)\right)Q^2\\
      &\quad - E'' \left(e^{- E'+\theta} + e^{ E'-\theta}\right)Q+
      \left(e^{- E'+\theta} + e^{ E'-\theta}\right)\,,\qquad Q(0)=0
    \end{aligned}
  \end{equation}
  forward in time,
\item Assemble the full estimate as
  \begin{equation}
    \begin{aligned}
      \bar P_T(z) &= \left(\frac{\eps}{2\pi Q(T)
          \theta^2(T)}\right)^{1/2} \exp\left(\frac12 \int_0^T
        \left(H_{\phi\phi}(\phi(t),\theta(t)) Q(t) +
          A^{(1)}(\phi(t))\right)\,dt\right)\\
      &\qquad \times \exp\left(-\eps^{-1}
        \left(\int \theta d\phi - H(\theta,\phi)T\right)\right)
    \end{aligned}
  \end{equation}
\end{enumerate}

This procedure can be carried out for various $z\in\RR$ to, for
example, investigate the probability $P_T(z)$ for rare events,
i.e.~large $z$. Displayed as black solid line in
Fig.~\ref{fig:queue} is the estimate $\bar P_T(z)$ at
$T=5.12$. Here, we choose $E(x) = \frac14 x^4$. For the numerical
computation of the instanton, we employed a simple forward Euler
scheme integrating the instanton
equations~(\ref{eq:nongauss-instanton}) forward and backward multiple
times until convergence, with $N_t = 512$ time discretization points,
and $\Delta t=10^{-2}$ temporal resolution. The blue markers compare
the instanton prediction against a numerical computation of the actual
stochastic process for finite (but small) $\eps\ll1$ by using the
Gillespie algorithm~\cite{gillespie:1977}. Numerical parameters are
$\eps=5\cdot 10^{-4}$, and we took $N_{\text{samples}} = 10^4$ samples to
estimate $P_T(z)$.

\section{Conclusions}
\label{sec:conclu}

In this paper, we have proposed explicit formulas to calculate the
prefactor contribution of expectations, probabilities, probability
densities, exit probabilities, and mean first passage times for
stochastic processes, both at finite time and over the invariant
measure. The approach gives sharp estimates in the small noise limit,
corresponding to the next order (pre-exponential) term to the
Freidlin-Wentzell large deviation limit. This allows us to compute these
 probabilistic quantities in absolute terms, i.e.~including
normalization constants, for finite noise values, instead of merely
producing the exponential scaling. This feature makes the approach
valuable whenever full quantitative estimates of probabilities are
required, as is the case in almost all applications, for example in physics,
chemistry,  biology, and engineering.

From a physical point of view the prefactor formulas given above
represent an explicit evaluation of the fluctuation determinant. In
principle, these ideas have been formulated multiple times in various
contexts (compare section~\ref{sec:related}), from quantum field
theory over linear-quadratic control to calculus of variations. As
noted by Schulman~\cite{schulman:1981}:
\begin{quote}
  {\it ``Methods for handling the quadratic Lagrangian are legion and have
  been well developed since the earliest work on path integrals. Oddly
  enough, papers on the subject continue to appear and may give some
  historian of science material for a case history on the nondiffusion
  of knowledge.''}
\end{quote}
Importantly, though, numerical methods for the calculation of
prefactors in the general setup we consider here remain, to the best
of our knowledge, mostly nonexistent. Concretely, for the prefactor
terms we (i) formulate algorithms suitable for their explicit
numerical computation, (ii) phrase them for the more general class of
Lagrangians encountered in large deviation theory, and (iii) treat the
case of the invariant measure and infinite time horizon. Computations
on stochastic partial differential equations highlight the fact that
our results produce correct results even in the irreversible infinite
dimensional case, and are efficient enough to be used for quantitative
estimates of probabilistic quantities in regimes where direct sampling
is inaccessible.

\section*{Acknowledgments}

TG wants to thank Timo Schorlepp and Rainer Grauer for helpful
discussions. The work of TG was supported by the EPSRC projects
EP/T011866/1 and EP/V013319/1. The work of TS was supported by the NSF
grants DMS-1522737 and DMS-2012548 as well as the PSC-CUNY grant
TRADB-51-281. The work of EVE was supported by National Science
Foundation (NSF) Materials Research Science and Engineering Center
Program grant DMR-1420073, NSF grant DMS-1522767, and the Simons
Collaboration on Wave Turbulence, Grant No. 617006.

\bibliographystyle{abbrv}
\bibliography{paperbib}

\appendix

\section{The geometric instanton and Riccati equations}
\label{sec:init-cond-geom}

\subsection{Geometric instanton equations}
\label{sec:geom-inst-equat}

Here we give some additional information about the geometric instanton
equations~(\ref{eq:4infpdf}), referring the reader
to\cite{heymann-vanden-eijnden:2008,grafke-grauer-schaefer-etal:2014,grafke-schaefer-vanden-eijnden:2017,grafke-vanden-eijnden:2019}
for more information.  Recall that these equations read
\begin{equation}
  \label{eq:instanton}
  \begin{cases}
    \lambda \hat\phi' = b(\hat\phi) + a \hat\theta\\
    \lambda \hat\theta' = -(\nabla b(\hat\phi))^\top \hat\theta
  \end{cases}
\end{equation}
with boundary conditions $\hat \phi(0)=x_*$,
$\hat \theta(1) = \nabla f(\hat \phi(1))$ as in~\eqref{eq:4infpdf}, or
$\hat \phi(0)=x_*$, $\hat \phi(1) = y$ as in~\eqref{eq:4inf}. Here
$\lambda(s) = ds/dt$ is the reparametrization factor from physical
time $t\in[-\infty,0]$ to normalized arc-length
$s\in[0,1]$, and primes denote
derivatives with respect to~$s$. In practice, these equations can be
solved globally by a relaxation
method~\cite{heymann-vanden-eijnden:2008}. This amounts to solving the
first equation in $a\hat\theta$ and inserting the result in the second
equation multiplied by $a$ to deduce
\begin{equation}
  \label{eq:61}
  0 = \lambda^2 \hat\phi'' + \lambda \lambda' \hat\phi''
  - \lambda a \left( a^{-1} \nabla b(\hat\phi)
    - (a^{-1}\nabla b(\hat\phi))^\top\right) \hat \phi' 
  -a (\nabla b(\hat\phi))^\top a^{-1}  b(\hat\phi)
\end{equation}
with boundary conditions $\hat \phi(0)=x_*$,
$a \hat \phi'(1) = \nabla f(\hat \phi(1))$ or $\hat \phi(0)=x_*$,
$\hat \phi(1) = y$.  We can now find the stable fixed points of this
equation by introducing an artificial relaxation time $\tau$ and
evolving $\{\phi(s):s\in(0,1)\}$ globally from some initial guess via
\begin{equation}
  \label{eq:61a}
  \partial_\tau \hat\phi = \lambda^2 \hat\phi'' + \lambda \lambda' \hat\phi''
  - \lambda a \left( a^{-1} \nabla b(\hat\phi)
    - (a^{-1}\nabla b(\hat\phi))^\top\right) \hat \phi' 
  -a (\nabla b(\hat\phi))^\top a^{-1}  b(\hat\phi).
\end{equation}
while enforcing $\hat \phi(0)=x_*$,
$a \hat \phi'(1) = \nabla f(\hat \phi(1))$ or $\hat \phi(0)=x_*$,
$\hat \phi(1) = y$ for all $\tau\ge0$. This is the procedure that was
used in this paper to solve the geometric instanton equations.

Alternatively, in situations where we want to solve~\eqref{eq:61} with
boundary condition $\hat \phi(0)=x_*$,
$\hat \theta(1) = \nabla f(\hat \phi(1))$ we can also iteratively
solve the equation for $\hat \phi$ forward in $s$, then that for
$\hat \theta$ backward in time forward-integrate the
$\hat\phi$-equation~\cite{grafke-grauer-schaefer-etal:2014,grafke-schaefer-vanden-eijnden:2017}. In
this case, we encounter a similar problem as with the equation for
$\hat Q$ in~\eqref{eq:Q}: at $s=0$, corresponding to physical time
$t=-\infty$ when the instanton starts at the fixed-point
$x_* = \hat\phi(0)$, one has $b(\hat\phi(0))=b(x_*)=0$ and
$\hat\theta(0)=0$, as well as $\lambda(0)=0$. This means that the
equation for $\hat \phi$ in~\eqref{eq:instanton} reads $0=0$ at $s=0$
and it is \textit{a~priori} unclear how to start the integration. Here
too, this problem is only apparent since the equation for $\hat \phi$
can be written as
\begin{equation}
  \label{eq:62}
  \hat\phi'(s) = \lambda^{-1}(s) \left(b(\hat\phi(s)) + a \hat\theta(s)\right)
\end{equation}
and a direct application of l'H\^opital rule shows that $\hat\phi'(0) =
\lim_{s\to0} \phi'(s)$ satisfies
\begin{equation}
  \label{eq:62lim}
  \hat\phi'(0) = [\lambda'(0)] ^{-1}\left(\nabla b(x_*) \hat \phi'(0)
    + a \hat\theta'(0)\right)
\end{equation}
where we used $\hat\phi(0) = x_*$. Noticing that the matrix
$(\lambda'(0)\mathrm{Id} - \nabla b(x_*))$ is invertible since
$\lambda'(0)>0$ and $-\nabla b(x_*)$ is positive-definite as $x_*$ is
a stable fixed point of $\dot x = b(x)$ by assumption, this equation
has the unique solution
\begin{equation}
  \hat\phi'(0) = (\lambda'(0)\mathrm{Id} - \nabla b(x_*))^{-1} a \hat\theta'(0)\,.
\end{equation}
This solution can be computed numerically since we have access to
$\lambda'(0)$ and $\hat\theta'(0)$ through our knowledge of
$\lambda(s)$ and $\hat\theta(s)$ from the last backward
integration. In particular, since $\lambda(0)=0$ and
$\hat\theta(0)=0$, we can write, for arc-length stepsize $\Delta s>0$,
\begin{equation}
  \hat\theta'(0) =\hat\theta(\Delta s)/\Delta s + O(\Delta s),\qquad
  \lambda'(0)= \lambda(\Delta s)/\Delta s + O(\Delta s)\,.
\end{equation}
Once $\hat\phi'(0)$ has been computed, it can be used together with
$\hat\phi(0)=x_*$ to start off the numerical integration of the
$\hat\phi$-equation forward in time using
\begin{equation}
  \label{eq:63}
  \begin{aligned}
    \hat \phi(\Delta s) &= x_* + \hat\phi'(0) \Delta s + O(\Delta s^2)\\
    &= x_* +(\lambda(\Delta s) \mathrm{Id} - \Delta s \nabla
    b(x_*))^{-1} a \hat \theta(0) + O(\Delta s^2) .
  \end{aligned}
\end{equation}

\subsection{Geometric Riccati equations}
\label{sec:geom-ricc-equat}

As discussed in remark~\ref{remark:geometric-riccati-initial}, we can
similarly derive the initial arc-length derivative $\hat Q_y'(0)$ by
solving an additional Lyapunov equation. Concretely, let us start from
\begin{equation}
  \label{eq:Qinf-appendix}
  \hat Q'_{y} = \lambda^{-1} \left(\hat Q_{y} \hat K_y \hat Q_{y} +
  \hat Q_{y} (\nabla b(\hat \phi_{y}))^\top + (\nabla b(\hat
  \phi_{y})) \hat Q_{y} + a\right) \equiv \lambda^{-1} \mathfrak R.
\end{equation}
Since $\mathfrak R(0)=0$ and $\lambda(0)=0$, via l'H\^opital's rule it
follows that $\hat Q_y'(0) = \lambda'(0)^{-1} \mathfrak R'(0)$., which
can be written explicitly as
\begin{equation}
  \label{eq:Q-prime-Lyapunov-appendix}
  \mathfrak C\, \hat Q_y'(0) + \hat Q_y'(0)\, \mathfrak C^\top = \mathfrak K \,.
\end{equation}
where we defined
\begin{equation}
  \begin{aligned}
    \mathfrak C &= \tfrac12 \lambda'(0) \mathrm{Id} - \nabla b(x_*)\\
    \mathfrak K &= Q_{*} (\nabla \nabla b(x_*)\hat \theta_y'(0)) 
    Q_{*} +Q_{*} (\nabla \nabla b(x_*) \hat
    \phi_y'(0))^\top + (\nabla \nabla b(x_*)\hat \phi_y'(0)) Q_{*}\,.
  \end{aligned}
\end{equation}
If $\nabla\nabla b(x_*)=0$, we have $\mathfrak K=0$ and hence one
obtains $\hat Q_y'(0)=0$---this is the case e.g. in the examples
presented in sections~\ref{sec:exampl-invar-meas-R2-nonlin}
or~\ref{sec:exampl-nonl-irrev-R2-nonlin}. In this situation, the
approximation $\hat Q_y(\Delta s) = Q_*$ is correct up to order
$\mathcal O(\Delta s^2)$. If a more accurate approximation is needed,
we can consider the next order: Taking the derivative
of~\eqref{eq:Qinf-appendix} we deduce that
\begin{equation}
  \hat Q_y''(s) = \lambda(s)^{-1} \mathfrak R'(s) - \lambda(s)^{-2}
  \mathfrak R(s)
  \lambda'(s)\,,
\end{equation}
and hence, using  l'H\^opital's rule again together with $\mathfrak R(0)= \mathfrak R'(0) = 0$ and
$\lambda(0)=0$, we obtain
\begin{equation}
  \hat Q_y''(0) = \tfrac12 \lambda'(0)^{-1} \mathfrak R''(0)\,.
\end{equation}
For brevity, we will refrain to write down this equation explicitly,
but we note that it is again a Lyapunov equation, this time for
$\hat Q_y''(0)$, where all other terms are known. Knowledge of
$\hat Q_y''(0)$ in situations where $\hat Q_y'(0) =0$ allows one to
use
\begin{equation}
  \hat Q_y(\Delta s) = Q_* + \tfrac12 \Delta s^2 \hat Q_y''(0)
  +\mathcal O(\Delta s)^3\,.
\end{equation}

\section{Expressions for the solution of Riccati
  equations as expectations}
\label{sec:riccatiexpect}

The following two propositions give ways to express the solution of
 backward and forward Riccati equations of the type considered in text
 in terms of expectations over the solution of some SDE:
 \begin{proposition}
   \label{th:riccati1}
   Given $A: [0,T] \to \RR^{n\times n}$,
   $B: [0,T] \to\RR^{n\times n}$,
   $\eta: [0,T] \to \RR^n$, all in
   $C^1([0,T])$ and with $A$ symmetric, as well as
   $C\in \RR^{n\times n}$, with $C$ symmetric, and $\xi\in\RR^n$, the
   following equality holds
  \begin{equation}
  \label{eq:8aa}
  \begin{aligned}
    &\EE^z \exp\left(\int_0^T \left( \tfrac12 A(t) : Z_t Z_t+ \eta(t)
        \cdot Z_t\right) \,dt +\tfrac12 C : Z_T Z_T + \xi\cdot
      Z_T\right)\\
    & = G(0) \exp\left(\tfrac12 \langle z, W(0) z\rangle +
      r(0)\cdot z\right)\,,
  \end{aligned}
\end{equation}
where $Z_t$ solves the linear SDE
\begin{equation}
  \label{eq:45}
  dZ_t = B(t) Z_t dt + \sigma dW_t\,;
\end{equation}
and $W:[0,T]\to \RR^{n\times n}$,  $r: [0,T]\to \RR^{n}$, $G: [0,T]\to \RR$, solve
\begin{equation}
  \label{eq:46}
  \left\{
    \begin{aligned}
      & \dot W + B^T(t) W  + WB(t) +  W aW + A(t) =0, && W(T) =
      C\,,\\
      & \dot r + B^T(t) r + Wa r + \eta(t) = 0, && r(T) = \xi\,,\\
      &\dot G + \tfrac12 \tr (a W) G + \tfrac12 a: r r G =
      0, \qquad && G(T) =1\,.
    \end{aligned}
    \right.
\end{equation}
\end{proposition}
\noindent Note that the solution to the equation for $G(t)$
in~\eqref{eq:46} can be expressed as
\begin{equation}
  \label{eq:49}
  G(t) = \exp\left( \tfrac12 \int_t^T \left(\tr (a W(s)) + a: r(s)
    r(s) \right) ds \right)\,.
\end{equation}

\begin{proof}
  Let $v: [0,T]\times \RR^n\to \RR$ solve
  \begin{equation}
  \label{eq:pdea1}
  \partial_t v + \langle B(t) z, \nabla v\rangle +
  \tfrac12 a :\nabla \nabla v + \left( \tfrac12 A(t)  : z z + \eta(t)
    \cdot z\right) v = 0\,,
\end{equation}
for the final condition
\begin{equation}
  \label{eq:qa1}
  v(T,z)
  = \exp\left(\tfrac12 C : zz + \xi\cdot z\right)\,.
\end{equation}
Then: (i) computing
$d\left(v(t,Z_t) \exp\left(\int_0^t \left( \tfrac12 A(s) : Z_s Z_s+
      \eta(s) \cdot Z_s\right) \,ds\right)\right)$ via Ito's formula,
taking expectation, and integrating on $t\in[0,T]$ shows that the
solution to this equation at time $t=0$ can be expressed as the
expectation in~\eqref{eq:8aa}; and (ii) substituting
$G(t) \exp\left(\tfrac12 \langle z, W(t) z\rangle + r(t)\cdot
  z\right)$ in~\eqref{eq:pdea1} shows that this expression satisfies
this equation as well as~\eqref{eq:qa1} if $W(t)$,  $r(t)$, and $G(t)$
satisfy~\eqref{eq:46}.
\end{proof}

\begin{proposition}
  \label{th:riccati2}
  Using the same notations as in Proposition~\ref{th:riccati1}, let
  $Q: [0,T]\to \RR^{n\times n}$ solve
\begin{equation}
  \dot Q = B(t) Q  + QB^T(t) + QA(t) Q + a, \qquad Q(0) = Q_0\,,
\end{equation}
for some $Q_0= Q_0^T$, positive semidefinite (possibly zero). Then
\begin{equation}
  \label{eq:47}
  Q(t) = \EE Z_t^Q (Z_t^Q)^T\,,
\end{equation}
where $Z^Q_t$ solves the nonlinear (in the sense of McKean) SDE
\begin{equation}
  \label{eq:45Q}
  dZ^Q_t = B(t) Z^Q_t dt  + \tfrac12 Q A(t) Z^Q_t dt + \sigma dW_t\,,
\end{equation}
and the expectation in~\eqref{eq:47} is taken over solutions
to~\eqref{eq:45Q} with initial conditions drawn from a Gaussian
distribution with mean zero and covariance $Q_0$.
\end{proposition}

\begin{proof}
  Application of Ito's formula shows that
  \begin{equation}
    \label{eq:48}
    \begin{aligned}
      \frac{d}{dt} \EE [Z_t^Q (Z_t^Q)^T] & = B(t) \EE [Z_t^Q
      (Z_t^Q)^T]
      + \EE [Z_t^Q (Z_t^Q)^T] B^T(t) \\
      & + \tfrac12 Q A(t) \EE [Z_t^Q (Z_t^Q)^T] + \tfrac12 \EE [Z_t^Q
      (Z_t^Q)^T] A(t) Q+  a\,.
    \end{aligned}
  \end{equation}
  Since $\EE [Z_0^Q (Z_0^Q)^T] = Q_0 = Q(0)$ initially, this equation shows that
  $\EE [Z_t^Q (Z_t^Q)^T] = Q(t)$ for $t\ge0$.
\end{proof}

The following proposition offers a practical way to simulate~\eqref{eq:45Q}:
\begin{proposition}
  \label{th:MC}
  Let $\{Z^i_t\}_{i=1}^n$ solve
\begin{equation}
  \label{eq:50}
  dZ^i_t = B(t) Z^i_t dt  + \tfrac12 n^{-1} \sum_{j=1}^n \<Z^i_t ,A(t)
  Z^j_t\> Z^j_t dt + \sigma dW^i_t\,, \qquad i=1,\ldots, n\,,
\end{equation}
where $\{W^i_t\}_{i=1}^n$ is a set of independent Wiener
processes. Then if we drawn the initial conditions for~\eqref{eq:50}
independently from a Gaussian distribution with zero mean and
covariance $Q_0$, we have
\begin{equation}
  \label{eq:51}
  \frac1n \sum_{i=1}^n Z^i_t (Z^i_t)^T \to Q(t) \qquad \text{almost
    surely as \ $n\to\infty$}
\end{equation}
\end{proposition}

\begin{proof}
  The proposition is a direct consequence of a `propagation of chaos'
  argument (see e.g~\cite{sznitman91}) applied to~\eqref{eq:50}.
\end{proof}

\section{Link between the forward and backward Riccati
  equations}
\label{sec:riccatilink}

We have:
 \begin{proposition}
   \label{th:riccati1link}
   Using the same notations as in Proposition~\ref{th:riccati1}, let
   $W:[0,T]\to \RR^{n\times n}$ and $Q:[0,T]\to \RR^{n\times n}$ solve
\begin{equation}
  \label{eq:46bb}
  \left\{
    \begin{aligned}
      & \dot W + B^T(t) W  + WB(t) +  W aW + A(t) =0, && W(T) = W_T\,,\\
      & \dot Q  = Q B^T(t)  + B(t)Q +  a + Q A(t) Q =0, && Q(0)=Q_0\,,
    \end{aligned}
    \right.
  \end{equation}
  where $W_T\in \RR^{n\times n}$ and $Q_0\in \RR^{n\times n}$, both
  symmetric. Then the following identity holds
  \begin{equation}
    \label{eq:39}
    \det \left(\text{Id} - W_T Q(T)\right) =  \det \left(\text{Id} - W(0) Q_0\right)
    \exp\left( \int_0^T \tr (A(t) Q(t) - W(t) a) dt\right) 
  \end{equation}
\end{proposition}

\begin{proof}
  From~\eqref{eq:46bb}, it is easy to see that the matrix $W(t)Q(t)$ satisfies
  \begin{equation}
    \label{eq:41}
    \frac{d}{dt} (WQ) = (B^T(t) + Wa) (\text{Id}-WQ) -
    (\text{Id}-WQ) ( B^T(t)+ A(t) Q)\,.
  \end{equation}
  As a result, using the Jacobi formula, we deduce that
  \begin{equation}
    \label{eq:52}
    \begin{aligned}
      \frac{d}{dt} \log \det \left (\text{Id} - WQ\right)  = - \tr \left(  (\text{Id} - WQ)^{-1}
        (d/dt)(WQ)\right) = \tr \left(A(t) Q -Wa\right)\,.
      \end{aligned}
    \end{equation}
    Integrating both side on $t\in[0,T]$ and taking the exponential of
    the result gives~\eqref{eq:39}.
\end{proof}

\section{The Radon's Lemma for the Riccati
  equation}
\label{sec:riccati}

It is well-known that a matrix-Riccati equation can be equivalently
represented by a linear matrix equation. Sometimes, this
transformation is called Radon's Lemma. Consider a differential
equation of the form
\begin{equation}
\frac{d}{dt} \begin{pmatrix} \Phi \\ \Theta \end{pmatrix}  = 
  \begin{pmatrix}
    M_{11} & M_{12} \\
    M_{21} & M_{22}
  \end{pmatrix}
  \begin{pmatrix} \Phi \\ \Theta \end{pmatrix}\,,
\end{equation}
and set $W=\Theta\Phi^{-1}$. Then
\begin{eqnarray*}
\dot W &=& \dot \Theta \Phi^{-1} + \Theta \dot \Phi^{-1} = \dot \Theta \Phi^{-1} - \Theta \Phi^{-1} \dot \Phi \Phi^{-1} \\
           &=& (M_{21}\Phi + M_{22}\Theta) \Phi^{-1} - \Theta\Phi^{-1}(M_{11}\Phi + M_{12}\Theta)\Phi^{-1} \\
           &=& M_{21} + M_{22}W - WM_{11} - WM_{12}W\,.
\end{eqnarray*}
We can apply this to the instanton matrix-Riccati equation by
choosing $M_{21} = -\< \nabla \nabla b,\theta\>$,
$M_{22} = - (\nabla b)^T$, $M_{11} = \nabla b$, and $M_{12}=a$ to
obtain:
\begin{equation}
  \label{eq:radon}
\frac{d}{dt} \begin{pmatrix} \Phi \\ \Theta \end{pmatrix}  = 
  \begin{pmatrix}
    \nabla b & a \\
    -\< \nabla \nabla b,\theta\> & -(\nabla b)^T
  \end{pmatrix}
  \begin{pmatrix} \Phi \\ \Theta \end{pmatrix}\,,
\end{equation}
and as final conditions we can choose $\Theta(T)=W(T)$ and
$\Phi(T)=\text{Id}$. While it seems appealing to solve the Riccati
equation this way, in practice the issue is that the equation for $\Phi$
is well-posed forward in time, whereas that for $\Theta$ is well-posed
backward in time. This means that the system~\eqref{eq:radon} has to
be solved iteratively, and the final condition are not simple to
impose. This is why we did not use~\eqref{eq:radon} in this paper.

\section{Derivation of
  $\det_\perp H = (\hat n^\top H^{-1} \hat n) \det H $}
\label{appendix:comp_det_perp}

Given an invertible, positive definite $H=H^\top\in \RR^{n\times n}$ and
a unit vector $\hat n$ we define $\det_\perp H$ via
\begin{equation}
  \label{eq:A1}
  (2\pi)^{(n-1)/2} \left| {\det_\perp} H \right|^{-1/2} =\int_{P}
  e^{-\frac12 \< y,H y\>} d\sigma(y) =: A
\end{equation}
where $P = \{y : \langle\hat n, y\rangle = 0\}$. For $m>0$, let 
\begin{equation}
  \label{eq:A2}
  H_m = H + m \hat n \hat n^\top 
\end{equation}
Clearly
\begin{equation}
  \label{eq:A3}
  A = \int_{P}
  e^{-\frac12 \< y,H_m y\>} d\sigma(y)
\end{equation}
since $\langle\hat n, y\rangle = 0$ in $P$. At the same time we have
\begin{equation}
  \label{eq:A4}
  \begin{aligned}
    (2\pi)^{n/2} \left| \det H_m \right|^{-1/2} & =\int_{\RR^n}
    e^{-\tfrac12 \< u,H_m u\>} du \\
    & = (\hat t \cdot \hat n) \int_{P} \int_{\RR} e^{-\tfrac12 \<
      (y+s\hat t) ,H_m  (y+s\hat t)\>}  ds d\sigma(y) \\
    & = (\hat t \cdot \hat n) \int_{P} e^{-\tfrac12 \< y ,H_m y\>}
    d\sigma(y) \int_{\RR} e^{-\tfrac12 s^2 \< \hat t,H_m \hat t\>} ds
  \end{aligned}
\end{equation}
where we used $u=y + s \hat t$ to change integration variable with
\begin{equation}
  \label{eq:A5}
  \hat t = \frac{H^{-1} \hat n }{|H^{-1} \hat n|} \quad
  \Leftrightarrow \quad \hat n = \frac{H \hat t }{|H \hat t|}
\end{equation}
Comparing~\eqref{eq:A1} and~\eqref{eq:A5} we deduce
\begin{equation}
  \label{eq:A6}
  \left| {\det_\perp} H \right| = (\hat n \cdot \hat
  t)^2 |\hat t^\top H_m \hat t|^{-1}  \left| \det H_m \right|
\end{equation}
Since
\begin{equation}
  \label{eq:A7}
  \begin{aligned}
    t^\top H_m \hat t & = \hat t^\top (H + m \, \hat n \, \hat n^\top)
    \hat t\\
    & = \hat t^\top H \hat t + m (\hat n \cdot \hat t)^2
  \end{aligned}
\end{equation}
we have
\begin{equation}
  \label{eq:A8}
  (\hat n \cdot \hat t)^2 | t^\top H_m \hat t |^{-1} = \frac{\hat n \cdot
  \hat t}{|H \hat t| + m (\hat n \cdot \hat t)}
\end{equation}
and~\eqref{eq:A6} can be written as
\begin{equation}
  \label{eq:A6b}
  {\det_\perp} H = \frac{\hat n \cdot
  \hat t}{|H \hat t| + m (\hat n \cdot \hat t)}  \, \det H_m 
\end{equation}
Next, write~\eqref{eq:A2} as
\begin{equation}
  \label{eq:A9}
  \begin{aligned}
    H_m & = H \left( \text{Id} + m H^{-1} \hat n \, \hat n^\top \right) \\
    & = H \left( \text{Id} + m | H^{-1} \hat n| \, \hat t \, \hat n^\top \right) 
  \end{aligned}
\end{equation}
so that 
\begin{equation}
  \label{eq:A10}
  \det H_m = \det H \det \left( \text{Id} + m | H^{-1} \hat n| \, \hat t
    \,  \hat n^\top \right) 
\end{equation}
The matrix $\text{Id} + m | H^{-1} \hat n| \hat t  \hat n^\top$ has $n-1$
eigenvectors perpendicular to $\hat n$, each with eigenvalue 1, and 
one eigenvector $\hat t$ with eigenvalue $1
  + m | H^{-1} \hat n| (\hat n \cdot \hat t)$ since
\begin{equation}
  \label{eq:A11}
  (\text{Id} + m | H^{-1} \hat n| \, \hat t  \, \hat n^\top) \hat t = \left(1
  + m | H^{-1} \hat n| (\hat n \cdot \hat t) \right)\hat t 
\end{equation}
Therefore
\begin{equation}
  \label{eq:A6c}
  {\det_\perp} H = \frac{(\hat n \cdot
    \hat t) \left(1
      + m | H^{-1} \hat n| (\hat n \cdot \hat t) \right)
  }{|H \hat t| + m (\hat n \cdot \hat t)}  \, \det H 
\end{equation}
Since $|H^{-1} \hat n| =|H\hat t|^{-1}$ this can be written as
\begin{equation}
  \label{eq:A12}
  \begin{aligned}
   {\det_\perp} H & = \frac{(\hat n \cdot
      \hat t)}{|H \hat t|} \left|\det H \right|\\
    & = (\hat n^\top H^{-1} \hat n)  \, \det H 
  \end{aligned}
\end{equation}

\section{General form of the instanton and Riccati equations for SPDEs}
\label{sec:general}
Let us generalize~\eqref{eq:SPDE-generic} into
\begin{equation}
  \label{eq:SPDE-generic00}
  \partial_t u = \mathcal B[u] + \sqrt{\eps}\eta\,,\quad u(0)=u_0\,,
\end{equation}
for $t\in[0,\infty)$, $x\in\Omega\subseteq \RR^d$, and
$u:[0,\infty)\times\Omega\to \RR$, and where  $\mathcal B[u]$ is a (possibly
nonlinear) differential operator in the spatial variable $x$ and the
noise $\eta$ is white-in-time Gaussian with covariance
\begin{equation}
  \label{eq:SPDE-noise00}
  \EE \eta(t,x) \eta(t',x') = \delta(t-t') \chi(x,x')\,.
\end{equation}
If we consider again
probabilities that a linear observable exceeds a certain threshold, 
\begin{equation}
  \label{eq:probaPDE00}
  P^{u_0}_\eps(T,z) = \mathbb P^{u_0} \Big(\int_\Omega \phi(x) u(T,x)\,dx\ge z\Big)\,,
\end{equation}
we formally obtain a proposition analogous to Proposition~\ref{thm:SPDEs}:
\begin{proposition}[Probabilities for SPDEs -- general case]
  \label{thm:SPDEs00}
  Let the fields $u(t,x)$, $\theta(t,x)$ solve the instanton
  equations 
  \begin{equation}
    \left\{\begin{aligned}
      \partial_tu &= \mathcal B[u] + \int_\Omega \chi(x,y)\theta(t,y)\,dy\,,&&u(0) = u_0\,,\\
      \partial_t \theta &= -\int_\Omega \frac{\delta\mathcal
        B[u](t,y)}{\delta u(x)}
      \theta(t,y)\,dy\,,&&\theta(T) = \phi\,,
    \end{aligned}\right.
  \end{equation}
  and let $\mathcal Q(t,x,y)$ solve
  \begin{equation}
    \label{eq:Q-SPDE00}
    \begin{aligned}
      \partial_t \mathcal Q =& \int_{\Omega^3} \mathcal Q(t,x,z_1)
      \mathcal K(t,z_1,z_2) \mathcal Q(t,z_2,y)\,dz_1\,dz_2\,dz_3\\& +
      \int_\Omega \frac{\delta \mathcal B[u](t,x)}{\delta u(t,z)}
      \mathcal Q(t,y,z)\,dz + \int_\Omega \frac{\delta \mathcal
        B[u](t,y)}{\delta u(t,z)} \mathcal Q(t,z,x)\,dz + \chi(x,y)\,,
    \end{aligned}
  \end{equation}
  with $\mathcal Q(0) = 0$ and where we denote
  \begin{equation}
    \mathcal K(t,x,y) = \int_\Omega \frac{\delta \mathcal B[u](t,z)}{\delta u(t,x) \delta u(t,y)} \theta(t,z)\,dz\,.
  \end{equation}
  Then the probability $P^{u_0}_\eps(T,z)$ in~\eqref {eq:probaPDE00} satisfies
  \begin{equation}
    \lim_{\eps\to0} \frac{P^{u_0}_\eps(T,z)}{\bar P^{u_0}_\eps(T,z)} = 1\,,
  \end{equation}
  where
  \begin{equation}
    P^{u_0}_\eps(T,z) = (2\pi)^{-1/2} \eps^{1/2} \mathcal V(T,u_0) \exp\left(-\frac1{2\eps} \int_0^T \int_{\Omega^2} \theta(t,x) \chi(x,y) \theta(t,y)\,dx\,dy\,dt\right)
  \end{equation}
  with
  \begin{equation}
    \mathcal V(T,u_0) = \left(\int_{\Omega^2} \theta(T,x) \mathcal Q(T,x,y) \theta(T,y)\,dx\,dy\right)^{1/2} \exp\left(\frac12 \int_0^T \int_{\Omega} \mathcal K(t,x,x) \mathcal Q(t,x,x)\,dx\,dt\right)\,.
  \end{equation}
\end{proposition}
\end{document}